\documentclass[reqno,10pt]{article}

\usepackage{bm}
\usepackage{amsopn}
\usepackage[inline]{enumitem}

\usepackage[utf8]{inputenc}
\usepackage{mathtools}
\usepackage{amssymb} 
\usepackage{makeidx}
\allowdisplaybreaks

\usepackage{multirow}
\usepackage{makecell}
\usepackage{rotating}

\newcommand{\CenteredVcell}[2]{\begin{centering}\begin{sideways}\parbox[c]{#1}{\begin{centering}{\hfill #2 \hfill\textcolor{white}{.}}\end{centering}}\end{sideways}\end{centering}}

\usepackage[english]{babel}
\usepackage{graphicx}
\usepackage{amsthm}
\usepackage{amsfonts,amsmath,amssymb}
\usepackage{oldgerm}
\usepackage{bbold}

\usepackage{mathrsfs}
\usepackage[active]{srcltx}
\usepackage{verbatim}
\usepackage{enumitem} %
\usepackage[toc,page]{appendix}
\usepackage{aliascnt,bbm}
\usepackage{array}
\usepackage[colorlinks = true, citecolor = blue]{hyperref}
\usepackage[a4paper]{geometry}
\usepackage[textwidth=4cm,textsize=footnotesize]{todonotes}
\usepackage{xargs}
\usepackage{cellspace}
\usepackage{algorithm}
\usepackage{algorithmic}

\usepackage[Symbolsmallscale]{upgreek}

\usepackage{tikz}
\usepackage{pgfplots}
\usepackage{pgfkeys}
\usetikzlibrary{spy}
\usetikzlibrary[patterns] %
\usetikzlibrary{calc,decorations.pathreplacing} %
\usetikzlibrary{shadows.blur} %
\usetikzlibrary{shapes.symbols}

\usepackage{accents}
\usepackage{dsfont}
\usepackage{aliascnt}
\usepackage{cleveref}
\usepackage{lmodern}%
\makeatletter

\crefname{theorem}{theorem}{Theorems}
\Crefname{Theorem}{Theorem}{Theorems}

\newaliascnt{lemma}{theorem}
\newtheorem{lemma}[lemma]{Lemma}
\aliascntresetthe{lemma}
\crefname{lemma}{lemma}{lemmas}
\Crefname{Lemma}{Lemma}{Lemmas}

\newaliascnt{proposition}{theorem}
\newtheorem{proposition}[proposition]{Proposition}
\aliascntresetthe{proposition}
\crefname{proposition}{proposition}{propositions}
\Crefname{Proposition}{Proposition}{Propositions}

\newaliascnt{remark}{theorem}

\aliascntresetthe{remark}
\crefname{remark}{remark}{remarks}
\Crefname{Remark}{Remark}{Remarks}

\crefname{figure}{figure}{figures}
\Crefname{Figure}{Figure}{Figures}

\newtheorem{assumption}{\textbf{H}\hspace{-3pt}}
\Crefname{assumption}{\textbf{H}\hspace{-3pt}}{\textbf{H}\hspace{-3pt}}
\crefname{assumption}{\textbf{H}}{\textbf{H}}

\Crefname{assumptionG}{\textbf{G}\hspace{-3pt}}{\textbf{G}\hspace{-3pt}}
\crefname{assumptionG}{\textbf{G}}{\textbf{G}}

\Crefname{assumptionA}{\textbf{A}\hspace{-3pt}}{\textbf{A}\hspace{-3pt}}
\crefname{assumptionA}{\textbf{A}}{\textbf{A}}

\makeatother

\usepackage{autonum}

\newcommand{\ceilLigne}[1]{\lceil #1 \rceil}
\newcommand{\itgamma}{\ceilLigne{1/\delta}}

\def\vx{x}
\def\posteriorw{\pi}
\def\posteriorwepsdelta{\pi_{\vareps, \delta}}
\def\posteriorwepsdeltac{\pi_{\vareps, \delta}}

\def\posteriorweps{\pi_{\vareps}}
\def\prior{\mu}

\def\rmf{\mathrm{f}}

\def\constanteM{\upkappa}

\def\bfX{\mathbf{X}}
\def\bbfX{\bar{\mathbf{X}}}
\def\bfB{\mathbf{B}}

\def\MAP{\mathrm{MAP}}

\def\dim{d}

\def\dimY{m} %

\def\Tg{\mathcal{T}_{\gamma}}

\newcommand{\cball}[2]{\overline{\operatorname{B}}(#1,#2)}

\def\bfDd{\mathbf{D}_{\mathrm{d}}}
\newcommand{\tup}[1]{\textup{#1}}
\newcommand{\mtt}{\mathtt{m}}

\newcommand{\prox}{\operatorname{prox}}
\newcommand{\Kker}{\mathrm{K}}

\newcommand{\Rker}{\mathrm{R}}
\newcommand{\Qker}{\mathrm{Q}}

\newcommand{\Pker}{\mathrm{P}}

\newcommand{\bvareps}{\vareps_0}

\def\w{{w}}

\newcommandx{\norm}[2][1=]{\ifthenelse{\equal{#1}{}}{\left\Vert #2 \right\Vert}{\left\Vert #2 \right\Vert^{#1}}}
\newcommandx{\normLigne}[2][1=]{\ifthenelse{\equal{#1}{}}{\Vert #2 \Vert}{\Vert #2\Vert^{#1}}}

\def\msa{\mathsf{A}}

\def\msk{\mathsf{K}}

\def\msc{\mathsf{C}}

\def\msu{\mathsf{U}}

\def\msx{\mathsf{X}}
\def\msy{\mathsf{Y}}

\def\mcbb{\mathcal{B}}  %
\newcommand{\mcb}[1]{\mathcal{B}(#1)}

\def\mcy{\mathcal{Y}}
\def\mcx{\mathcal{X}}

\def\rset{\mathbb{R}}

\def\nset{\mathbb{N}}
\def\nsets{\mathbb{N}^*}

\def\rmd{\mathrm{d}}

\def\rme{\mathrm{e}}

\def\rmc{\mathrm{C}}

\def\rma{\mathrm{a}}

\def\rmA{\mathrm{A}}

\newcommand{\cc}{\mathrm{c}}

\newcommand{\R}{\mathbb R}

\newcommandx{\functionspace}[2][1=+]{\mathbb{F}_{#1}(#2)}
\newcommand{\argmax}{\operatorname*{arg\,max}}
\newcommand{\argmin}{\operatorname*{arg\,min}}

\newcommandx{\VarDeux}[3][3=]{\operatorname{Var}^{#3}_{#1}\left\{#2 \right\}}

\newcommand{\1}{\mathbbm{1}}

\newcommand{\LeftEqNo}{\let\veqno\@@leqno}

\newcommand{\floor}[1]{\left\lfloor #1 \right\rfloor}

\newcommand{\N}{\ensuremath{\mathbb{N}}}

\newcommand{\PE}{\mathbb{E}}

\newcommand{\abs}[1]{\left\vert #1 \right\vert}
\newcommand{\absLigne}[1]{\vert #1 \vert}
\newcommand{\tvnorm}[1]{\left \| #1 \right \|_{\mathrm{TV}}}
\newcommand{\tvnormsq}[1]{\left \| #1 \right \|_{\mathrm{TV}}^{1/2}}
\newcommand{\tvnormLigne}[1]{\| #1 \|_{\mathrm{TV}}}
\newcommandx{\VnormLigne}[2][1=V]{\| #2 \|_{#1}}
\newcommandx{\Vnorm}[2][1=V]{\left\| #2 \right\|_{#1}}

\newcommand{\parenthese}[1]{\left(#1 \right)}
\newcommand{\parentheseLigne}[1]{(#1 )}
\newcommand{\parentheseDeux}[1]{\left[ #1 \right]}
\newcommand{\parentheseDeuxLigne}[1]{[ #1 ]}
\newcommand{\defEns}[1]{\left\lbrace #1 \right\rbrace }
\newcommand{\defEnsLigne}[1]{\lbrace #1 \rbrace }

\newcommand{\ps}[2]{\left\langle#1,#2 \right\rangle}

\newcommand{\proba}[1]{\mathbb{P}\left( #1 \right)}
\newcommand{\probasq}[1]{\mathbb{P}^{1/2}\left( #1 \right)}

\newcommand{\probaLigne}[1]{\mathbb{P}( #1 )}

\newcommandx\probaMarkovTilde[2][2=]
{\ifthenelse{\equal{#2}{}}{{\widetilde{\mathbb{P}}_{#1}}}{\widetilde{\mathbb{P}}_{#1}\left[ #2\right]}}

\newcommand{\expe}[1]{\PE \left[ #1 \right]}
\newcommand{\expesq}[1]{\PE^{1/2} \left[ #1 \right]}
\newcommand{\expeLignesq}[1]{\PE^{1/2} [ #1 ]}

\newcommand{\expeLigne}[1]{\PE [ #1 ]}

\def\eqsp{\;}
\newcommand{\coint}[1]{\left[#1\right)}
\newcommand{\ocint}[1]{\left(#1\right]}
\newcommand{\ooint}[1]{\left(#1\right)}
\newcommand{\ccint}[1]{\left[#1\right]}

\newcommand{\ocintLigne}[1]{(#1]}

\newcommandx{\weight}[2][2=n]{\omega_{#1,#2}^N}

\newcommand{\ball}[2]{\operatorname{B}(#1,#2)}

\newcommand{\diameter}{\operatorname{diam}}

\newcommandx\sequence[3][2=,3=]
{\ifthenelse{\equal{#3}{}}{\ensuremath{\{ #1_{#2}\}}}{\ensuremath{\{ #1_{#2}, \eqsp #2 \in #3 \}}}}
\newcommandx\sequenceD[3][2=,3=]
{\ifthenelse{\equal{#3}{}}{\ensuremath{\{ #1_{#2}\}}}{\ensuremath{( #1)_{ #2 \in #3} }}}

\newcommandx{\sequencen}[2][2=n\in\N]{\ensuremath{\{ #1_n, \eqsp #2 \}}}
\newcommandx\sequenceDouble[4][3=,4=]
{\ifthenelse{\equal{#3}{}}{\ensuremath{\{ (#1_{#3},#2_{#3}) \}}}{\ensuremath{\{  (#1_{#3},#2_{#3}), \eqsp #3 \in #4 \}}}}
\newcommandx{\sequencenDouble}[3][3=n\in\N]{\ensuremath{\{ (#1_{n},#2_{n}), \eqsp #3 \}}}

\newcommand{\wrt}{w.r.t.}

\def\eg{e.g.}

\newcommand{\opnorm}[1]{{\left\vert\kern-0.25ex\left\vert\kern-0.25ex\left\vert #1 
    \right\vert\kern-0.25ex\right\vert\kern-0.25ex\right\vert}}

\def\generator{\mathcal{A}}

\def\Id{\operatorname{Id}}

\newcommandx{\CPE}[3][1=]{{\mathbb E}_{#1}\left[#2 \left \vert #3 \right. \right]} %
\newcommandx{\CPELigne}[3][1=]{{\mathbb E}_{#1}[#2  \vert #3 ]} %
\newcommandx{\CPVar}[3][1=]{\mathrm{Var}^{#3}_{#1}\left\{ #2 \right\}}
\newcommand{\CPP}[3][]
{\ifthenelse{\equal{#1}{}}{{\mathbb P}\left(\left. #2 \, \right| #3 \right)}{{\mathbb P}_{#1}\left(\left. #2 \, \right | #3 \right)}}

\newcommandx{\osc}[2][1=]{\mathrm{osc}_{#1}(#2)}

\def\Id{\operatorname{Id}}

\def\domain{\mathrm{D}}

\def\w{w}

\def\bgamma{\bar{\gamma}}

\newcommand{\ensemble}[2]{\left\{#1\,:\eqsp #2\right\}}
\newcommand{\ensembleLigne}[2]{\{#1\,:\eqsp #2\}}

\def\diag{\Delta_{\rset^d}}

\newcommand\coupling[2]{\Gamma(\mu,\nu)}

\newcommand{\complementary}{\mathrm{c}}

\renewcommand{\geq}{\geqslant}
\renewcommand{\leq}{\leqslant}

\def\diam{\mathrm{diam}}

\def\Leb{\mathrm{Leb}}

\def\vareps{\varepsilon}

\def\Phibf{\mathbf{\Phi}}

\newcommandx{\KLbig}[2]{\mathrm{KL}\left( #1 \middle| #2 \right)}

  \newcommand \rmP {\mathrm{P}}
    \newcommand \rmm {\mathrm{m}}

 \newcommand \eps {\varepsilon}
 
\newcommand \ie {{\em i.e. }}

\newcommand \KL {\mathrm{KL}}
\newcommand \KLsqrt {\mathrm{KL}^{1/2}}

\def\pnpula{PnP-ULA}

\def\pnpsgd{PnP-SGD}
\def\Mtt{\mathtt{M}}
\def\Ktt{\mathtt{K}}

\def\Ltt{\mathtt{L}}

\def\rmm{\mathrm{m}}
\newcommand{\rref}[1]{\textup{\Cref{#1}}}
\def\bdelta{\bar{\delta}}
\newcommand{\wassersteinD}[1][1=\distance]{\mathbf{W}_{#1}}
\def\Pens{\mathscr{P}}

\def\bfDd{\mathbf{D}_{\mathrm{d}}}
\def\bfDc{\mathbf{D}_{\mathrm{c}}}

\def\floor#1{\lfloor #1 \rfloor}
\def\1{\bm{1}}

\def\eps{{\epsilon}}

\def\rmA{{\mathbf{A}}}

\def\rmP{{\mathbf{P}}}

\def\vx{{{x}}}  
\def\vy{{{y}}}

\DeclareMathAlphabet{\mathsfit}{\encodingdefault}{\sfdefault}{m}{sl}
\SetMathAlphabet{\mathsfit}{bold}{\encodingdefault}{\sfdefault}{bx}{n}

\newcommand{\Fdata}{\ensuremath{F}}
\newcommand{\Gprior}{\ensuremath{U}}

\usepackage{subfig}
\usepackage{mwe}

\newlength{\tempheight}
\newlength{\tempwidth}

\newcommand{\rowname}[1]%
{\rotatebox{90}{\makebox[\tempheight][c]{\textbf{#1}}}}

\newcommand{\columnname}[1]%
{\makebox[\tempwidth][c]{\textbf{#1}}}

\title{Bayesian imaging using Plug \& Play priors: when Langevin meets Tweedie~\thanks{VDB was partially supported by EPSRC grant EP/R034710/1. RL was partially supported by grants from Région Ile-De-France. AD acknowledges support  of the Lagrange Mathematical and Computing Research Center. MP acknowledges support by EPSRC grant EP/T007346/1. JD and AA acknowledge support from the French Research Agency through the PostProdLEAP project (ANR-19-CE23-0027-01). Computer experiments for this work ran on a Titan Xp GPU donated by NVIDIA, as well as on HPC resources from GENCI-IDRIS (Grant 2020-AD011011641).}}

\author{
	\hspace{1.5cm} Rémi Laumont~\footnotemark[2]~\footnotemark[3]%
	\footnote{These authors contributed equally}%
	\footnote{Université de Paris, MAP5 UMR 8145, F-75006 Paris, France}%
	\and Valentin De Bortoli~\footnotemark[2]~\footnote{Department of Statistics
		University of Oxford
		24-29 St Giles
		OX1 3LB, Oxford
		United Kingdom}%
	\and \newline Andr\'{e}s Almansa~\footnotemark[3]
	\and Julie Delon~\footnotemark[3]~\footnotemark[4]%
	\footnote{Institut Universitaire de France (IUF)}
	\and Alain Durmus~\footnote{Centre Borelli, UMR 9010, 
		\'Ecole Normale Supérieure Paris-Saclay}%
	\and Marcelo Pereyra~\footnote{School of Mathematical and Computer Sciences, Heriot-Watt University \& Maxwell Institute for Mathematical Sciences, Edinburgh, United Kingdom}
}

\begin{document}
	\maketitle

	\begin{abstract}
		Since the seminal work of Venkatakrishnan et
		al.~\cite{venkatakrishnan2013plug} in 2013, \emph{Plug \& Play} (PnP) methods have
		become ubiquitous in Bayesian imaging. These methods derive estimators for inverse problems in imaging by combining an explicit likelihood function with a prior that is implicitly defined by an image denoising algorithm. 
		In the case of optimisation schemes, some recent
		works guarantee the convergence to a fixed point, albeit not necessarily a maximum-a-posteriori Bayesian
		estimate. In the case of Monte Carlo sampling schemes for general Bayesian computation, to the best of our knowledge there
		is no known proof of convergence. Algorithm convergence issues aside, there are important open questions
		regarding whether the underlying Bayesian models and estimators are well
		defined, well-posed, and have the basic regularity properties required to
		support efficient Bayesian computation schemes. This paper develops theory for Bayesian analysis and computation with PnP priors. We introduce  PnP-ULA
		(Plug \& Play Unadjusted Langevin Algorithm) for Monte Carlo sampling and minimum mean squared error
		estimation. Using recent results on the quantitative convergence of Markov
		chains, we establish detailed convergence guarantees for this algorithm
		under realistic assumptions on the denoising operators used, with special
		attention to denoisers based on deep neural networks.  We also show that these
		algorithms approximately target a decision-theoretically optimal Bayesian
		model that is well-posed and meaningful from a frequentist viewpoint. PnP-ULA is demonstrated on several canonical problems such as image deblurring and inpainting, where it is used for point estimation as well as for uncertainty visualisation
		and quantification. 

	\end{abstract}

	\section{Introduction}
	
	\subsection{Bayesian inference in imaging inverse problems}
Most inverse problems in imaging aim at reconstructing an unknown image
$\vx \in \mathbb{R}^d$ from a degraded observation $\vy \in \rset^\dimY$ under some assumptions on their relationship. For
example, many works consider observation models of the form
$\vy = \rmA(\vx) + n$, where $\rmA: \ \rset^d \to \rset^\dimY$ is a degradation
operator modelling deterministic instrumental aspects of the observation
process, and $n$ is an unknown (stochastic) noise term taking values in
$\rset^\dimY$. The operator $\rmA$ can be known or not, and is usually assumed to be
linear (e.g., $\rmA$ can represent blur, missing pixels, a projection,
etc.). %

The estimation of $\vx$ from $\vy$ is usually ill-posed or
ill-conditioned\footnote{That is, either the estimation problem does not admit a
  unique solution, or there exists a unique solution but it is not Lipschitz
  continuous w.r.t. to perturbations in the data $y$.} and additional
assumptions on the unknown $\vx$ are required in order to deliver meaningful
estimates. The Bayesian statistical paradigm provides a natural framework to
regularise such estimation problems. The relationship between $\vx$ and $\vy$ is
described by a statistical model with likelihood function $p(\vy|\vx)$, and the
knowledge about $\vx$ is encoded by the \emph{prior} distribution for $\vx$,
typically specified via a density function $p(\vx)$ %
or by its potential $U(\vx)= -\log p(\vx)$. Similarly, in some cases the
likelihood $p(\vy|\vx)$ is specified via the potential
$\Fdata(\vx,\vy) = -\log p(\vy|\vx)$. %
The likelihood and prior define the joint distribution with density
$p(\vx,\vy)=p(\vy|\vx)p(\vx)$, from which we derive the \emph{posterior}
distribution with density $p(\vx|\vy)$ where for any $x \in \rset^d, y \in \rset^\dimY$
\begin{equation}
  \textstyle{p(x|y) = p(y|x)p(x) / \int_{\rset^d} p(y|\tilde{x})p(\tilde{x}) \rmd \tilde{x} \eqsp ,}
\label{eq:posterior}
\end{equation}
which underpins all inference about $\vx$ given the observation $\vy$. Most imaging methods seek to derive estimators reaching some 
 kind of consensus between prior and likelihood, as for instance the Minimum Mean Square Error (MMSE) or Maximum A Posteriori
(MAP) estimators
\begin{eqnarray}
\label{eq:MAP}
\hat{\vx}_{\textsc{map}} &=& \textstyle{\argmax_{x \in \rset^d} p(\vx|\vy) \! =\! \argmin_{x \in \rset^d} \left \{\Fdata(\vx,\vy) + \Gprior(\vx)\right \} \eqsp ,}  \\
\label{eq:MMSE}
\hat{\vx}_{\textsc{mmse}} &=&  \textstyle{\argmin_{u \in \rset^d} \CPELigne{\|x - u\|^2}{\vy} = \CPELigne{x}{\vy} = \int_{\rset^d} \tilde{x} p(\tilde{x}|y) \textrm{d}\tilde{x} \eqsp .}
\end{eqnarray}

The quality of the inference about $\vx$ given $\vy$ depends on how accurately
the specified prior represents the true marginal distribution for $\vx$. Most
works in the Bayesian imaging literature consider relatively simple priors
promoting sparsity in transformed domains or piece-wise regularity (e.g.,
involving the $\ell_1$ norm or the total-variation pseudo-norm
\cite{Rudin1992,Chambolle04,Louchet2013,Pereyra2016}), Markov random fields
\cite{MRF-MIT-2011}, or learning-based priors like patch-based Gaussian or
Gaussian mixture models \cite{Zoran2011,yu2011solving, Aguerrebere2014b,
  Teodoro2018scene,houdard2018high}. Special attention is given in the
literature to models that have specific factorisation structures or that are
log-concave, as this enables the use of Bayesian computation algorithms that
scale efficiently to high-dimensions and which have detailed convergence
guarantees, \cite{Pereyra2016, durmus2018efficient,
  repetti_pereyra_2019,girolami2011riemann,chen2014stochastic}.%
\subsection{Bayesian computation in imaging inverse problems}
There is a vast literature on Bayesian computation methodology for models
related to imaging sciences (see, e.g., \cite{pereyra2015survey}). Here, we
briefly summarise efficient high-dimensional Bayesian computation strategies
derived from the Langevin stochastic differential equation (SDE)

\begin{equation}\label{LangevinSDE_intro}
\begin{split}
\textrm{d}\bfX_t &= \nabla \log p(\bfX_t|\vy) + \sqrt{2}\textrm{d}\bfB_t  = \nabla \log p(y|\bfX_t) + \nabla \log p(\bfX_t) + \sqrt{2}\textrm{d}\bfB_t \eqsp ,
\end{split}
\end{equation}
where $(\bfB_t)_{t\geq 0}$ is a $d$-dimensional Brownian motion. When $p(x|y)$
is proper and smooth, with $\vx \mapsto \nabla \log p(\vx|\vy)$
Lipschitz continuous\footnote{That is, there exists $L \geq 0$ such that for any
  $x_1, x_2 \in \mathbb{R}^d$,
  $\|\nabla \log p(x_1|y) - \nabla \log p(x_2|y)\| \leq L\|x_1 - x_2\|$}, then,
for any initial condition $\bfX_0 \in \rset^d$, the SDE
\eqref{LangevinSDE_intro} has a unique strong solution $(\bfX_t)_{t\geq 0}$ that
admits the posterior of interest $p(x|y)$ as unique stationary density
\cite{roberts1996exponential}. In addition, for any initial condition
$\bfX_0 \in \rset^d$ the distribution of $\bfX_t$ converges towards the
posterior distribution in total variation. %
Although solving \eqref{LangevinSDE_intro} in continuous time is generally not
possible, we can use discrete time approximations of \eqref{LangevinSDE_intro}
to generate samples that are approximately distributed according to $p(x|y)$. A
natural choice is the Unadjusted Langevin algorithm (ULA) Markov chain
$(X_{k})_{k\geq 0}$ obtained from an Euler-Maruyama discretisation of
\eqref{LangevinSDE_intro}, given by $X_0 \in \rset^d$ and the following
recursion for all $k \in \nset$
\begin{equation}\label{ULA_intro}
X_{k+1} = X_{k} + \delta \nabla \log p(y|{X}_k) + \delta \nabla \log p({X}_k) + \sqrt{2\delta} Z_{k+1} \eqsp , 
\end{equation}
where $\ensembleLigne{Z_k}{k \in \nset}$ is a family of i.i.d Gaussian random
variables with zero mean and identity covariance matrix and
$\delta > 0$ is a step-size which controls a
trade-off between asymptotic accuracy and convergence speed
\cite{dalalyan2014theoretical,durmus2017nonasymp}. The approximation error
involved in discretising \eqref{LangevinSDE_intro} can be asymptotically removed
at the expense of additional computation by combining \eqref{ULA_intro} with a
Metropolis-Hastings correction step, leading to the so-called
Metropolis-adjusted Langevin Algorithm (MALA)
\cite{roberts1996exponential}.

When the prior density $p(x)$ is log-concave but not smooth, one can still use ULA by approximating the gradient
of $U(x) = -\log p(x)$
in \eqref{ULA_intro}
by the gradient of the smooth Moreau-Yosida envelope %
$U_\lambda(x)$,
given for any $x \in \mathbb{R}^d$ and $\lambda > 0$ by
$\nabla U_\lambda(x) = \frac{1}{\lambda}(x - \prox_U^\lambda(x))$. \footnote{Recall: The Moreau-Yosida envelope is defined as $U_\lambda(x) = \inf_{\tilde{x}} U(\tilde{x}) + \frac{1}{2\lambda}\|x-\tilde{x}\|^2$
and the proximal operator is defined as 
$\prox_U^\lambda (x) = \argmin_{\tilde{x} \in \mathbb{R}^d} U(\tilde{x}) + \frac{1}{2\lambda}\|x-\tilde{x}\|_2^2$.}
For example, one could use the Moreau-Yosida ULA \cite{durmus2018efficient}, given by $X_0 \in \rset^d$ and the following
recursion for all $k \in \nset$
\begin{equation}\label{intro_MYULA}
X_{k+1} = X_{k} + \delta \nabla \log p(y|{X}_k) + \frac{\delta}{\lambda} \left[\prox_U^\lambda({X}_k)-{X}_k\right] + \sqrt{2\delta} Z_{k+1} \eqsp.
\end{equation}
Notice that $\prox_U^\lambda$ is equivalent to MAP denoising under the prior $p(x)$, for additive white Gaussian noise with noise variance $\lambda$. The \emph{Plug \& Play} ULA methods studied in this paper are closely related to \eqref{intro_MYULA}, with a state-of-the-art Gaussian denoiser ``plugged'' in lieu of $\prox_U^\lambda$.
However, instead of approximating $\nabla U$ via a Moreau-Yosida envelope as above, we use Tweedie's identity \eqref{eq:Tweedie} relating $\nabla U$ to an MMSE denoiser (see Section~\ref{sec:BayesianPnP}).

\subsection{Machine learning and Plug \& Play approaches in imaging inverse problems}
In an apparently different direction, machine learning approaches have recently
gained a considerable importance in the field of imaging inverse problems,
particularly strategies based on deep neural networks. Indeed, neural networks
can be trained as regressors to learn the function
$\vy \mapsto \hat{\vx}_{\textsc{mmse}}$ empirically from a huge dataset of
examples $\{\vx^\prime_i,\vy^\prime_i\}_{i=1}^N$, where $N \in \nset$ is the
size of the training dataset. Many recent works on the topic report
unprecedented accuracy. This training can be
agnostic~\cite{dong2014learning,zhang2017beyond,zhang2018ffdnet,
  gharbi2016deep,schwartz2018deepisp,gao2019dynamic} or exploit the knowledge of
$\rmA$ in the network architecture via unrolled optimization
techniques~\cite{gregor2010learning,Chen2017,diamond2017unrolled,gilton2019neumann}. However,
solutions encoded by end-to-end neural networks are mostly problem specific and
not easily adapted to reflect changes in the problem (e.g., in instrumental
settings). There also exist concerns regarding the stability of such approaches
for general reconstruction problem \cite{antun2020instabilities,antun2021can}.

A natural
strategy to reconcile the strengths of the Bayesian paradigm and neural networks
is provided by \emph{Plug \& Play} approaches.  These data-driven regularisation
approaches learn an implicit representation of the prior density $p(x)$ (or its
potential $U(x) = - \log p(x)$) while keeping an explicit likelihood density,
which is usually assumed to be known and calibrated
\cite{arridge_maass_oktem_schonlieb_2019}. More precisely, using a denoising
algorithm $D_\varepsilon$, \emph{Plug \& Play} approaches seek to derive an
approximation of the gradient $\nabla\Gprior$ (called the Stein score)
\cite{Bigdeli2017,Bigdeli2017a} or $\prox_\Gprior$
\cite{meinhardt2017learning,Zhang2017,chan2017plug,kamilov2017plug,ryu2019plug},
which can for instance been used within an iterative minimisation scheme to approximate
$\hat{\vx}_{\MAP}$, or within a Monte Carlo sampling scheme to approximate
$\hat{\vx}_{\textsc{mmse}}$~\cite{Alain2012, guo2019agem,
  kadkhodaie2020solving}. To the best of our knowledge, the idea of leveraging a denoising algorithm to approximate the score $\nabla\Gprior$ within a iterative Monte Carlo scheme was first proposed in the seminal paper \cite{Alain2012} in the context of generative modelling with denoising auto-encoders, where the authors present a Monte Carlo scheme that can be viewed as an approximate \emph{Plug \& Play} MALA. This scheme was recently combined with an expectation maximisation approach and applied to Bayesian inference for inverse problems in imaging in \cite{guo2019agem}. Similarly, the recent work \cite{kadkhodaie2020solving} proposes to solve imaging inverse problems by using a \emph{Plug \& Play} stochastic gradient strategy that has close connections to an unadjusted version of the MALA scheme of \cite{Alain2012}. While these approaches have shown some remarkable
  empirical performance, they rely on hybrid algorithms that are not always well understood
  and that in some cases fail to converge. Indeed, their convergence properties
  remain an important open question, especially when $D_\varepsilon$ is
  implemented as a neural network that is not a gradient mapping. These algorithms are better understood when interpreted as fixed-point algorithms seeking to reach a set of equilibrium equations between the denoiser and the data fidelity term~\cite{buzzard2018plug}.
Our understanding of the convergence properties of hybrid optimisation methods has advanced significantly recently %
~\cite{ryu2019plug,Xu2020,sun2020scalable,hurault2021gradient}, but these questions remain
largely unexplored in the context of stochastic Bayesian algorithms, to compute
$\hat{\vx}_{\textsc{mmse}}$ or perform other forms of statistical inference.

The use of \emph{Plug \& Play} operators has also been investigated in the context of
Approximate Message Passing (AMP) computation methods (see \cite{donoho2009message} for
an introduction to AMP focused on compressed sensing and \cite{pnp_ahmad2020} for a survey on PnP-AMP in the context of magnetic resonance imaging), particularly for applications involving randomised forward operators where it is possible to characterise AMP schemes in detail (see, e.g., \cite{bayati2011dynamics,javanmard2013state,metzler2016denoising,chen2017bm3d}). This is an active area of research, and recent works have extended the approach to Vector AMP (VAMP) strategies and characterised their behaviour for a wider class of problems \cite{fletcher2019plug}.

Approaches based on
score matching techniques \cite{song2019generative,ho2020denoising} have also shown
promising results recently
\cite{kawar2021stochastic,kawar2021snips}. These methods are linked with
\emph{Plug \& Play} approaches as they also estimate a Stein score. However,
they do not rely on the asymptotic convergence of a diffusion, but instead aim at inverting a noising process stemming from an optimal transport problem \cite{debortoli2020convergence}. The recent work \cite{kawar2021snips} is particularly relevant in this context as it considers a range of imaging inverse problems, where it exploits the structure of the forward operator to perform posterior sampling in a coarse-to-fine manner. This also allows the use of multivariate step-sizes that are specific to each scale and ensure stability. However, to the best of our knowledge, the convergence properties of \cite{kawar2021snips} have not been studied yet.

\subsection{Contributions summary}
This paper presents a formal framework for Bayesian analysis and computation with \emph{Plug \& Play} priors. We propose two \emph{Plug \& Play} ULAs, with detailed convergence guarantees under realistic assumptions on the denoiser used. We also study important questions regarding whether the underlying Bayesian models and estimators are well defined, well-posed, and have the basic regularity properties required to
support efficient Bayesian computation schemes. We pay particular attention to denoisers based on deep neural networks, and report extensive numerical experiments with a
specific neural network denoiser~\cite{ryu2019plug} shown to satisfy our
convergence guarantees.

The remainder of the paper is organized as follows. \Cref{sec:marcelo} defines notation, introduces our framework for studying Bayesian inference methods with \emph{Plug \& Play} priors, and presents two \emph{Plug \& Play} ULAs for Bayesian computation in imaging problems. This is then followed by a detailed theoretical analysis of \emph{Plug \& Play} Bayesian models and algorithms in \Cref{sec:theory}.
\Cref{sec:experimental} demonstrates the proposed approach with experiments related to non-blind image deblurring and image inpainting, where we perform point estimation and uncertainty visualisation analyses, and report comparisons with the \emph{Plug \& Play} Stochastic Gradient Descent method of \cite{laumont2021maximum}.
Conclusions and perspectives for future work are finally reported in \Cref{sec:conclusion}.

	\section{Bayesian inference with Plug \& Play priors: theory methods and algorithms}
	\label{sec:marcelo}

\subsection{Bayesian modelling and analysis with \emph{Plug \& Play} priors}\label{sec:BayesianPnP}
This section presents a formal framework for Bayesian analysis and computation
with \emph{Plug \& Play} priors. As explained previously, we are interested in
the estimation of the unknown image $x$ from an observation $y$ when the problem
is ill-conditioned or ill-posed, resulting in significant uncertainty about the
value of $x$. The Bayesian framework addresses this difficulty by using prior
knowledge about the marginal distribution of $x$ in order to reduce the
uncertainty about $x|y$ and make the estimation problem well posed. In the
Bayesian \emph{Plug \& Play} approach, instead of explicitly specifying the
marginal distribution of $x$, we introduce prior knowledge about $x$ by
specifying an image denoising operator $D_\vareps$ for recovering $x$ from a
noisy observation $x_\vareps \sim \mathcal{N}(x,\vareps\Id)$ with noise variance
$\vareps > 0$. A case of particular relevance in this context is when
$D_\vareps$ is implemented by a neural network, trained by using a set of clean
images $\{x^\prime_i\}_{i=1}^N$.

A central challenge in the formalisation of Bayesian inference with \emph{Plug \& Play} priors is that the denoiser $D_\vareps$ used is generally not directly related to a marginal distribution for $x$, so it is not possible to derive an explicit posterior for $x|y$ from $D_\vareps$. As a result, it is not clear that plugging $D_\vareps$ into gradient-based algorithms such as ULA leads to a well-defined or convergent scheme that is targeting a meaningful Bayesian model.

To overcome this difficulty, in this paper we analyse \emph{Plug \& Play} Bayesian models through the prism of \emph{M-complete} Bayesian modelling
\cite{bernardo_smith_bayesian_theory}. Accordingly, there exists a true -albeit
unknown and intractable- marginal distribution for $x$ and posterior
distribution for $\vx|\vy$. If it were possible, basing inferences on these true marginal and
posterior distributions would be optimal both in terms of point
estimation and in terms of delivering Bayesian probabilities that are valid from
a frequentist viewpoint. We henceforth use $\mu$ to denote this optimal prior
distribution for $\vx$ on $(\mathbb{R}^d, \mathcal{B}(\mathbb{R}^d))$ - where $\mathcal{B}(\rset^d)$ denotes the Borel $\sigma$-field of
$\rset^d$, and when
$\mu$ admits a density w.r.t. the Lebesgue measure on $\mathbb{R}^d$, we denote
it by $p^\star$. In the latter case, the posterior distribution for
$\vx|\vy$ associated with the marginal $\mu$ also admits a density that is given
for any $x \in \rset^d$ and $y \in \rset^\dimY$ by
\begin{equation}\label{posterior_star}
p^\star(x|y) = p(y|x)p^\star(x) / \textstyle{\int_{\rset^d} p(y|\tilde{\vx})p^\star(\tilde{\vx})\textrm{d}\tilde{\vx}}\, .    
\end{equation}
\footnote{Strictly speaking, the true likelihood $p^{\star}(y|x)$ may also be unknown, this is particularly relevant in the case of blind or myopic inverse imaging problems. For simplicity,  we restrict our experiments and theoretical development to the case where $p(y|x)$ represents the true likelihood. Generalizations of our approach to the blind or semi-blind setting are discussed, e.g. by \cite{guo2019agem} - formalising these generalisations is an important perspective for future work.}
Unlike most Bayesian
imaging approaches that operate implicitly in an \emph{M-closed} manner and treat
their postulated Bayesian models as true models (see
\cite{bernardo_smith_bayesian_theory} for more details), we explicitly regard $p^\star$ (or
more precisely $\mu$) as a fundamental property of the unknown $x$, and
models used for inference as operational approximations of $p^\star$ specified by the practitioner (either
analytically, algorithmically, or from training data). This distinction will be useful
for using the oracle posterior \eqref{posterior_star} as a reference, and \emph{Plug \& Play} Bayesian algorithms based on a denoiser $D_\vareps$ as approximations to reference algorithms to perform inference w.r.t. $p^\star$. The accuracy of the \emph{Plug \& Play} approximations will depend chiefly on the closeness between $D_\vareps$ and an optimal denoiser $D^\star_\vareps$ derived form $p^\star$ that we define shortly.

In this conceptual construction, the marginal $\mu$ naturally depends on the
imaging application considered. It could be the distribution of natural images
of the size and resolution of $\vx$, or that of a  class of images
related to a specific application. And in problems where there is training
data $\{x^\prime_i\}_{i=1}^N$ available, we regard
$\{x^\prime_i\}_{i=1}^N$ as samples from $\mu$. Lastly, we note that the
posterior for $x|y$ remains well defined when $\mu$ does not admit a density;
this is important to provide robustness to situations where $p^\star$ is nearly
degenerate or improper. For clarity, our presentation assumes that $p^\star$
exists, although this is not strictly required \footnote{Operating without
  densities requires measure disintegration concepts that are technical
  \cite{schwartzdesintegration}.}.

Notice that because $\mu$ is unknown, we cannot verify that $p^\star(x|y)$ satisfies the basic
desiderata for gradient-based Bayesian computation: i.e., $p^\star(x|y)$ need not be
proper and differentiable, with $\nabla \log p^\star(\vx|\vy)$ Lipschitz
continuous. To guarantee that gradient-based algorithms that target approximations of $p^\star(\vx|\vy)$ are well defined by construction, we introduce a regularised oracle $\mu_\vareps$
obtained via the convolution of $\mu$ with a Gaussian smoothing kernel with
bandwidth $\vareps > 0$. Indeed, by construction, $\mu_\vareps$ has a smooth proper
density $p_{\vareps}$ given for any $x \in \rset^d$ and $\vareps > 0$
by%
$$
\textstyle{p^\star_\vareps (x) = (2\uppi\vareps)^{-d/2} \int_{\mathbb{R}^d} \exp{[-\|x-\tilde{x}\|_2^2/(2\vareps)]} p^\star(\tilde{\vx}) \textrm{d}\tilde{\vx}\eqsp.}
$$
Equipped with this regularised marginal distribution, we use Bayes' theorem to involve the likelihood $p(\vy|\vx)$ and derive the
posterior density $p_{\vareps}^\star(x|y)$, given for any $\vareps > 0$ and
$x \in \rset^d$ by
\begin{equation}
  \label{eq:approx}
\textstyle{p^\star_{\vareps}(x|y) = p(y|x){p^\star_\vareps}(x)/ \int_{\rset^d} p(y|\tilde{x})p^\star_\vareps(\tilde{x}) \textrm{d}\tilde{x} \eqsp , }
\end{equation}
which inherits the regularity properties required for gradient-based Bayesian computation when the likelihood satisfies the following standard conditions:
\begin{assumption}
\label{assum:post}
For any $y \in \rset^\dimY$, $\sup_{x \in \rset^d} p(y|x) < +\infty$,
$p(y|\cdot) \in \rmc^1(\rset^d, \ooint{0, +\infty})$ and there exists
$\Ltt_y > 0$ such that $\nabla \log (p(y|\cdot))$ is $\Ltt_y$ Lipschitz
continuous.
\end{assumption}
More precisely, \Cref{prop:prop_eps} below establishes that the regularised prior
$p_\vareps^\star(x)$ and posterior $p_\vareps^\star(x|y)$ are proper, smooth,
and that they can be made arbitrarily close to the original oracle models
$p^\star(x)$ and $p^\star(x|y)$ by reducing $\vareps$, with the approximation
error vanishing as $\vareps \rightarrow 0$. 

\begin{proposition}
\label{prop:prop_eps}
Assume \rref{assum:post}. Then, for any $\vareps > 0$ and $y \in \rset^\dimY$, the following hold:
\begin{enumerate}[label=(\alph*), leftmargin=4em]
    \item \label{item:a_proper} $p_{\vareps}^\star$  and $p_{\vareps}^\star(\cdot|y)$ are proper.
    \item \label{item:b_der} For any $k \in \nset$,
      $p_{\vareps}^\star \in \rmc^k(\rset^d)$. In addition, if
      $p(y|\cdot) \in \rmc^k(\rset^d)$ then
      $p_{\vareps}^\star(\cdot| y) \in \rmc^{k}(\rset^d, \rset)$.
    \item \label{item:integration} Let $k \in \nset$. If $\int_{\rset^d} \norm{\tilde{x}}^k p^\star(x)\rmd \tilde{x} < +\infty$ then $\int_{\rset^d} \norm{\tilde{x}}^k p_\vareps^\star(\tilde{x}|y) \rmd \tilde{x} < +\infty$.
    \item \label{item:qual_bounds} $\lim_{\vareps \to 0} \| p_{\vareps}^\star(\cdot |y) - p^{\star}(\cdot|y) \|_1 = 0$. 
    \item \label{item:quant_bounds} In addition, if there exist $\kappa, \upbeta \geq 0$ such that for any $x \in \rset^d$, $\| p^{\star} - p^{\star}( \cdot - x) \|_1 \leq  \norm{x}^{\upbeta}$, then there exists $C \geq 0$ such that $\| p_{\vareps}^\star(\cdot |y) - p^{\star}(\cdot|y) \|_1 \leq C \vareps^{\upbeta / 2}$.
\end{enumerate}
\end{proposition}

\begin{proof}
    The proof is postponed to \Cref{prop:posterior_eps:proof}.
\end{proof}

Under \rref{assum:post} and $p(y|\cdot) \in \rmc^1(\rset^d)$, $x \mapsto \nabla \log p_\vareps^\star(x|y)$ is
well-defined and continuous. However, $x \mapsto \nabla \log p_\vareps^\star(x|y)$ might
not be Lipschitz continuous and hence the Langevin SDE \eqref{LangevinSDE_intro} might not have a strong solution. This requires an additional assumption on $\mu$.

To study the Lipschitz continuity of $x \mapsto \nabla \log p_\vareps^\star(x|y)$, as well as to set the
grounds for \emph{Plug \& Play} methods that define priors implicitly through a
denoising algorithm, we introduce the oracle MMSE denoiser $D^\star_\vareps$
defined for any $x \in \rset^d$ and $\vareps > 0$ by
$$
\textstyle{D^\star_\vareps(x) = (2\uppi\vareps)^{-d/2} \int_{\mathbb{R}^d} \tilde{x} \exp{[-\|x-\tilde{x}\|^2/(2\vareps)]} p^\star(\tilde{\vx}) \textrm{d}\tilde{\vx} \eqsp.}
$$
Under the assumption that the expected mean square error (MSE) is finite, $D^\star_\vareps$ is the MMSE estimator to recover an image $x \sim \mu$ from a noisy
observation $x_\vareps \sim \mathcal{N}(x,\vareps\Id)$
\cite{robert2007bayesian}. Again, this optimal denoiser is a fundamental property of $x$ and it is generally intractable. Motivated by the fact that state-of-the-art image denoisers are close-to-optimal in terms of MSE, in Section 2.3 we will characterise the accuracy of \emph{Plug \& Play} Bayesian methods for approximate inference w.r.t. $p_\vareps^\star(x|y)$ and $p^\star(x|y)$ as a function of the closeness between the denoiser $D_\eps$ used and the reference $D^\star_\vareps$.

To relate the gradient $x \mapsto \nabla \log p_\vareps^\star(x)$ and
$D^\star_\vareps$, we use Tweedie's identity
\cite{efron2011tweedie} which states that for all $x \in \rset^d$ %
\begin{equation}\label{eq:Tweedie}
\vareps \nabla \log p_\vareps^\star(x) =  D^\star_\vareps(x) - x \eqsp ,
\end{equation}
and hence $x \mapsto \nabla \log p_\vareps^\star(x|y)$ is Lipschitz continuous if
and only if $D^\star_\vareps$ has this property. We argue that this is a natural
assumption on $D^\star_\vareps$, as it is essentially equivalent to assuming
that the denoising problem underpinning $D^\star_\vareps$ is well-posed in the sense of Hadamard (recall
that an inverse problem is said to be well posed if its solution is unique and
Lipschitz continuous w.r.t to the observation \cite{stuart_2010}). As established in \Cref{prop:reg_p_eps} below, this happens when the
expected MSE involved in using $D^\star_\vareps$ to recover $\vx$ from
$x_\vareps \sim \mathcal{N}(x,\vareps\Id)$, where $\vx$ has marginal $\mu$, is
finite and uniformly upper bounded for all $\vx_\vareps \in
\mathbb{R}^d$. %

\begin{proposition}
  \label{prop:reg_p_eps}
  Assume \rref{assum:post}. Let $\vareps > 0$.  $\nabla \log p_{\vareps}^\star$ is
  Lipschitz continuous if and only if there exists $C \geq 0$ such that for any
  $x_{\vareps} \in \rset^d$
\begin{equation}
\label{eq:var}
    \textstyle{\int_{\rset^d} \norm{{x} - D_{\vareps}^{\star}(x_\vareps)}^2 g_{\vareps}({x} | x_\vareps) \rmd {x} \leq C \eqsp ,}
\end{equation}
where $g_{\vareps}(\cdot|x_{\vareps})$ is the density of the conditional
distribution of the unknown image $x \in \rset^d$ with marginal
$\mu$, given a noisy observation $x_{\vareps} \sim
\mathcal{N}(x,\vareps\Id)$. See \Cref{sec:convergence_pnpula} for details.
\end{proposition}

\begin{proof}
  The proof is postponed to \Cref{lemma:lip_peps}. 
\end{proof}

These results can be generalised to hold under the weaker assumption that the
expected MSE for $D^\star_\vareps$ is finite but not uniformly bounded, as in
this case $x \mapsto \nabla \log p_\vareps^\star(x|y)$ is locally instead of
globally Lipschitz continuous (we postpone this technical extension to future
work). The pathological case where $D^\star_\vareps$ does not have a finite MSE
 arises when $\mu$ is such that the denoising problem does not admit a Bayesian
 estimator w.r.t. to the MSE loss. In summary, the gradient $x \mapsto \nabla \log p_\vareps^\star(x|y)$ is Lipschitz continuous when $\mu$ carries enough information to make the problem of Bayesian image denoising under Gaussian additive noise well posed.
 
Notice that by using Tweedie's identity, we can express a ULA recursion for sampling approximately from $p^\star_{\vareps}(\vx|\vy)$ as follows:
\begin{equation}\label{MYULAe}
\begin{split}
X_{k+1} =& X_{k} + \delta \nabla \log p(y|{X}_k) + (\delta/\vareps) \left(D^\star_\vareps({X}_k)- X_{k}\right) + \sqrt{2\delta} Z_{k+1} \eqsp . 
\end{split}
\end{equation}
where we recall that $\ensembleLigne{Z_k}{k \in \nset}$ are i.i.d standard
Gaussian random variables on $\mathbb{R}^d$ and
$\delta > 0$ is a positive step-size. Under standard assumptions on $\delta$, the sequence generated by \eqref{MYULAe} is a Markov chain which admits an invariant probability distribution whose density is provably close to $p^\star_{\vareps}(\vx|\vy)$, with $\delta$ controlling a
trade-off between asymptotic accuracy and convergence speed. In the following section we present \emph{Plug \& Play} ULAs that arise from replacing $D^\star_\vareps$ in \eqref{MYULAe} with a denoiser $D_\vareps$ that is tractable.

Before concluding this section, we study whether the oracle
$p^\star(x|y)$ is itself well-posed, i.e., if
$p^\star(x|y)$ changes continuously w.r.t. $y$ under a suitable probability
metric (see \cite{latzs2020well}). We answer positively to this question in
\Cref{prop:stability_observation_simplified} which states that, under mild
assumptions on the likelihood, $p^\star(x|y)$ is locally Lipschitz continuous
w.r.t. $y$ for an appropriate metric. This stability result implies, for
example, that the MMSE estimator derived from $p^\star(x|y)$ is locally
Lipschitz continuous w.r.t. $y$, and hence stable w.r.t. small perturbations of
$y$. Note that a similar property holds for the regularised posterior
$p_\vareps^\star(x|y)$. In particular,
\Cref{prop:stability_observation_simplified} holds for Gaussian likelihoods (see
\Cref{sec:theory} for details).

\begin{proposition}
  \label{prop:stability_observation_simplified}
  Assume that there exist $\Phi_1: \ \rset^d \to \coint{0,+\infty}$ and $\Phi_2: \ \rset^\dimY \to \coint{0, +\infty}$
  such that for any $x \in \rset^d$ and $y_1, y_2 \in \rset^\dimY$
  \begin{equation}
    \label{eq:condition_reg}
    \norm{\log(p(y_1|x)) - \log(p(y_2|x))} \leq (\Phi_1(x) + \Phi_2(y_1) + \Phi_2(y_2)) \norm{y_1 - y_2}\eqsp ,
  \end{equation}
  and for any $c > 0$, 
  $\int_{\rset^d} (1+\Phi_1(\tilde{x})) \exp[c \Phi_1(\tilde{x})] p^\star(x) \rmd \tilde{x} < +\infty$.
  Then $y \mapsto p^\star(\cdot|y)$ is locally Lipschitz w.r.t 
  $\normLigne{\cdot}_1$, \ie, for any compact set $\msk$ there exists $C_{\msk} \geq 0$ such that for any $y_1, y_2 \in \msk$ , $\norm{p^\star(\cdot | y_1) - p^\star(\cdot|y_2)}_1 \leq C_{\msk} \norm{y_1 - y_2} $.
\end{proposition}

\begin{proof}
The proof is a straightforward application of \Cref{prop:stability_observation}. 
\end{proof}

To conclude, starting from the decision-theoretically optimal model
$p^\star(x|y)$, we have constructed a regularised approximation
$p_\vareps^\star(x|y)$ that is proper and smooth by construction, with gradients that are explicitely related to denoising operators by Tweedie's formula. Under mild
assumptions on $p(y|x)$, the approximation $p_\vareps^\star(x|y)$ is well-posed
and can be made arbitrarily close to the oracle
$p^\star(x|y)$ by controlling $\vareps$. Moreover, we established that
$x \mapsto \nabla \log p_\vareps^\star(x)$ is Lipschitz continuous when the
problem of Gaussian image denoising for $\mu$ under the MSE loss is well
posed. This allows imagining convergent gradient-based algorithms for performing Bayesian
computation for $p_\vareps^\star(x|y)$, setting the basis for \emph{Plug \& Play} ULA schemes that mimic these idealised algorithms by using a tractable denoiser $D_\epsilon$ such as neural network, trained to optimise MSE performance and hence to approximate the oracle MSE denoiser $D^\star_\epsilon$.

\subsection{Bayesian computation with \emph{Plug \& Play} priors} 
\label{sec:tractable_bayesian}
We are now ready to study \emph{Plug \& Play} ULA schemes to perform approximate inference w.r.t. $p_\vareps^\star(x|y)$ (and hence indirectly w.r.t. $p^\star(x|y)$). We use \eqref{MYULAe} as starting point, with $D^\star_\vareps$ replaced by a surrogate denoiser $D_\vareps$, but also modify \eqref{MYULAe} to guarantee geometrically fast convergence\footnote{Geometric convergence is highly desirable property in large-scale problems and guarantees that the generated Markov chains can be used for Monte Carlo integration. } to a neighbourhood of $p_\vareps^\star(x|y)$. In particular, geometrically fast convergence is achieved here by modifying far-tail probabilities to prevent the Markov chain from becoming too diffusive as it explores the tails of $p_\vareps^\star(x|y)$. We consider two alternatives to guarantee geometric
convergence with markedly different bias-variance trade-offs: one with excellent
accuracy guarantees but that requires using a small step-size
$\delta$ and hence has a higher computational cost, and
another one that allows taking a larger step-size $\delta$ to
improve convergence speed at the expense of weaker guarantees in terms of
estimation bias. 

First, in the spirit of Moreau-Yosida regularised ULA \cite{durmus2018efficient}, we define \emph{Plug \& Play} ULA (PnP-ULA) as the following recursion: given an initial state $X_0 \in \rset^d$ and
for any $k \in \nset$, 
\begin{equation}\label{PNPULAe}
\begin{split}
\textrm{(PnP-ULA)}\quad X_{k+1} =& X_{k} + \delta \nabla \log p(y|{X}_k) + (\delta/\vareps) \left(D_\vareps({X}_k)-X_{k}\right)\\ &+  (\delta/\lambda) (\Pi_\msc (X_k) -X_k) + \sqrt{2\delta} Z_{k+1} \eqsp , 
\end{split}
\end{equation}
where $\msc \subset \mathbb{R}^d$ is some large compact convex set that contains most of the prior probability mass of $x$,  $\Pi_{\msc}$ is the
projection operator onto $\msc$ w.r.t the Euclidean scalar product on
$\rset^d$, and $\lambda > 0$ is a tail regularisation parameter that is set such that the drift in PnP-ULA satisfies a certain growth condition as $\|\vx\| \rightarrow \infty$ (see \Cref{sec:theory} for details). 

An alternative strategy (which we call Projected \pnpula , \ie P\pnpula, see
\Cref{algo:ppnp-ula}) is to modify PnP-ULA to include a hard projection onto $\msc$,
\ie \ $(X_k)_{k \in \nset}$ is defined by $X_0 \in \msc$ and the following
recursion for any $k \in \nset$
\begin{equation}\label{PMYULAe}
\begin{split}
X_{k+1} = \Pi_\msc \left[X_{k} + \delta \nabla \log p(y|{X}_k) + (\delta/\vareps) (D_\vareps({X}_k) - X_{k})+  \sqrt{2\delta} Z_{k+1}\right] \eqsp , 
\end{split}
\end{equation}
where we notice that, by construction, the chain cannot exit $\msc$ because of
the action of the projection operator $\Pi_\msc$. The hard projection guarantees
geometric convergence with weaker restrictions on $\delta$ and
hence PPnP-ULA can be tuned to converge significantly faster than PnP-ULA,
albeit with a potentially larger bias. These two schemes are summarised in
\Cref{algo:pnp-ula} and \Cref{algo:ppnp-ula}
below. Note the presence of a regularisation parameter $\alpha$ in these algorithms, which permits to balance the weights between the prior and data terms. For the sake of simplicity, this parameter is set to $\alpha=1$ in \Cref{sec:theory} and \Cref{sec:experimental} but will be taken into account in the supplementary material \Cref{sec:supplemantary}. \Cref{sec:convergence_pnpula} and \Cref{sec:proj-altern} present detailed
convergence results for \pnpula \  and P\pnpula. Implementation guidelines,
including suggestions for how to set the algorithm parameters of \pnpula \  and
P\pnpula \  are provided in \Cref{sec:experimental}.

\begin{algorithm}  
	\caption{PnP-ULA}\label{algo:pnp-ula}
	\begin{algorithmic}
		
		\REQUIRE $n \in \nset$, $y \in \rset^\dimY$, $\vareps, \lambda, \alpha, \delta > 0$, $\msc\subset \rset^d$  convex and compact
		
		\ENSURE $2 \lambda (2 \Ltt_y + \alpha \Ltt / \vareps) \leq 1$ and $\delta < (1/3)(\Ltt_y + 1/\lambda + \alpha \Ltt/\vareps)^{-1}$
		
		\STATE \textbf{Initialization:} Set $X_0 \in \rset^d$ and $k=0$.
		
		\FOR{$k=0:N$}
		
		\STATE $Z_{k+1} \sim \mathcal{N}(0, \Id)$
		
		\STATE $X_{k+1} = X_k + \delta \nabla \log(p(y| X_k)) + (\alpha \delta / \vareps) (D_{\vareps}(X_k) -X_k) + (\delta / \lambda) (\Pi_{\msc} (X_k) -X_k) + \sqrt{2 \delta} Z_{k+1}$

		\ENDFOR

		\STATE \textbf{return} $\ensembleLigne{X_k}{k \in \{0, \dots, N+1\}}$

	\end{algorithmic}
\end{algorithm}

\begin{algorithm}  
	\caption{PPnP-ULA}\label{algo:ppnp-ula}
	\begin{algorithmic}
		
		\REQUIRE $n \in \nset$, $y \in \rset^\dimY$, $\vareps, \lambda, \alpha, \delta > 0$, $\msc\subset \rset^d$ convex and compact

		\STATE \textbf{Initialization:} Set $X_0 \in \msc$ and $k=0$.
		
		\FOR{$k=0:N$}
		
		\STATE $Z_{k+1} \sim \mathcal{N}(0, \Id)$
		
		\STATE $X_{k+1} = \Pi_{\msc} \left(X_k + \delta \nabla \log(p(y| X_k)) + (\alpha \delta / \vareps) (D_{\vareps}(X_k) -X_k) + \sqrt{2 \delta} Z_{k+1}\right)$

		\ENDFOR

		\STATE \textbf{return} $\ensembleLigne{X_k}{k \in \{0, \dots, N+1\}}$

	\end{algorithmic}
\end{algorithm}

Lastly, it is worth mentioning that \Cref{algo:pnp-ula} and \Cref{algo:ppnp-ula}
can be straightforwardly modified to incorporate additional regularisation
terms. More precisely, one could consider a prior defined as the (normalised)
product of a \emph{Plug \& Play} term and an explicit analytical term. In that
case, one should simply modify the recursion defining the Markov chain by adding
the gradient associated with the analytical term. In a manner akin to
\cite{durmus2018efficient}, analytical terms that are not smooth are involved
via their proximal operator.

Before concluding this section, it is worth
emphasising that, in addition to being important in their own right,
\Cref{algo:pnp-ula} and \Cref{algo:ppnp-ula} and the
associated theoretical results set the grounds for analysing more advanced
stochastic simulation and optimisation schemes for performing Bayesian inference
with \emph{Plug \& Play} priors, in particular accelerated optimisation and
sampling algorithms \cite{pereyra2020accelerating}. This is an important
perspective for future work.

	\section{Theoretical analysis}
	\label{sec:theory}

In this section, we provide a theoretical study of the long-time behaviour of
\pnpula , see \Cref{algo:pnp-ula} \ and P\pnpula , see \Cref{algo:ppnp-ula}.  %
For any $\vareps > 0$ we recall that
$p_{\vareps}^\star$ is given by the Gaussian smoothing of $p$ with level $\vareps$,
for any $x \in \rset^d$ by
\begin{equation}
  \textstyle{
    p^\star_{\vareps}(x) = (2 \uppi \vareps)^{-d/2} \int_{\rset^{\dim}} \exp[-\norm{x - \tilde{x}}^2 / (2 \vareps)] \ p^\star(\tilde{x}) \rmd \tilde{x} \eqsp .
    }
\end{equation}
One typical example of likelihood function that we consider in our numerical
illustration, see \Cref{sec:experimental}, is
$p(y|x) \propto \exp[-\norm{\rmA x - y}^2/(2 \sigma^2)]$ for any $x \in \rset^d$ with
$\sigma > 0$ and $\rmA \in \rset^{m \times \dim}$. 
We define $\posteriorw$ the target
posterior distribution given for any $x \in \rset^d$ by
$(\rmd \posteriorw / \rmd \Leb)(x) = p^\star(x|y)$.  We also consider the family
of probability distributions $\ensembleLigne{\posteriorweps}{\vareps >0}$ given
for any $\vareps > 0$ and $x \in \rset^{\dim}$ by
\begin{equation}
  \label{eq:posterior_eps}
  (\rmd \posteriorweps / \rmd \Leb) (x) = \left. p(y|x) p_{\vareps}^{\star}(x) \middle/ \int_{\rset^{\dim}} p(y|\tilde{x}) p_{\vareps}^{\star}(\tilde{x}) \rmd \tilde{x} \right. \eqsp .
\end{equation}
Note that in the supplementary material \Cref{sec:supplemantary} we
investigate the general setting where $p_\vareps^\star$ is replaced by
$(p_\vareps^\star)^\alpha$ for some $\alpha > 0$ that acts as a regularisation parameter. We divide our
study into two parts.  We recall that $\pi_\vareps$ is well-defined for any
$\vareps > 0$ under \rref{assum:post}, see \Cref{prop:prop_eps}.  We start with some notation in \Cref{sec:notation}. We then 
establish non-asymptotic bounds between the iterates of \pnpula \ and
$\pi_{\vareps}$ with respect to the total variation distance for any
$\vareps > 0$, in \Cref{sec:convergence_pnpula}. Finally, in
\Cref{sec:proj-altern} we establish similar results for
P\pnpula.%

\subsection{Notation}
\label{sec:notation}
Denote by $\mathcal{B}(\rset^d)$ the Borel $\sigma$-field of
$\rset^d$, and for $f : \rset^d \to \rset$ measurable,
$\Vnorm[\infty]{f}= \sup_{\tilde{x} \in \rset^d} \abs{f(\tilde{x})}$.  For $\mu$ a
probability measure on $(\rset^d, \mathcal{B}(\rset^d))$ and $f$ a
$\mu$-integrable function, denote by $\mu(f)$ the integral of $f$
\wrt~$\mu$. For $f: \ \rset^d \to \rset$ measurable and $V: \ \rset^d \to [1, \infty )$ measurable, the $V$-norm  of
$f$ is given by $\Vnorm[V]{f}= \sup_{\tilde{x} \in \rset^d} |f(\tilde{x})|/V(\tilde{x})$. Let
$\xi$ be a finite signed measure on $(\rset^d,\mcbb(\rset^d))$. The
$V$-total variation distance of $\xi$ is defined as
\begin{equation}
\textstyle{\Vnorm[V]{\xi} = \sup_{\Vnorm[V]{f} \leq 1}  \abs{\int_{\rset^d } f(\tilde{x}) \rmd \xi (\tilde{x})} \eqsp.}
\end{equation}
If $V = 1$, then $\Vnorm[V]{\cdot}$ is the total variation denoted by
$\tvnorm{\cdot}$.  Let $\msu$ be an open set of $\rset^d$.  For any pair of
measurable spaces $(\msx, \mcx)$ and $(\msy, \mcy)$, measurable function
$f : \ (\msx, \mcx) \to (\msy, \mcy)$ and measure $\mu$ on $(\msx, \mcx)$ we
denote by $f_{\#} \mu$ the pushforward measure of $\mu$ on $(\msy, \mcy)$ given
for any $\msa \in \mcy$ by $f_{\#}\mu(\msa) = \mu(f^{-1}(\msa))$. We denote
$\Pens(\rset^d)$ the set of probability measures over $(\rset^d, \mcb{\rset^d})$
and for any $m \in \nset$,
$\Pens_m(\rset^d) = \ensembleLigne{\nu \in \Pens(\rset^d)}{\int_{\rset^d}
  \normLigne{\tilde{x}}^m \rmd \nu(\tilde{x}) < +\infty}$.

We denote by $\rmc^{k}(\msu, \rset^\dimY)$ and $\rmc^{k}_c(\msu, \rset^\dimY)$
the set of $\rset^\dimY$-valued $k$-differentiable functions, respectively the
set of compactly supported $\rset^\dimY$-valued and $k$-differentiable
functions.  Let $f : \msu \to \rset$, we denote by $\nabla f$, the gradient of
$f$ if it exists. $f$ is said to be $\mtt$-convex with $\mtt \geq 0$ if for all
$x_1,x_2 \in \rset^d$ and $t \in \ccint{0,1}$,
\begin{equation}
  f(t x_1 + (1-t) x_2) \leq t f(x_1)  + (1-t) f(x_2) -\mtt t(1-t)  \norm[2]{x_1-x_2}/2  \eqsp.
\end{equation}
For any $a \in \rset^d$ and $R > 0$, denote $\ball{a}{R}$ the open 
ball centered at $a$ with radius $R$.  %
Let $(\msx, \mcx)$ and $(\msy, \mcy)$ be two measurable spaces. A
Markov kernel $\Pker$ is a mapping
$\Kker: \ \msx \times \mcy \to \ccint{0, 1}$ such that for any
$\tilde{x} \in \msx$, $\Pker(\tilde{x}, \cdot)$ is a probability measure and for any
$\msa \in \mcy$, $\Pker(\cdot, \msa)$ is measurable. For any
probability measure $\mu$ on $(\msx, \mcx)$ and measurable function
$f : \msy \to \rset_+$ we denote
$\mu \Pker = \int_{\msx} \Pker(x, \cdot) \rmd \mu(x)$ and
$\Pker f = \int_{\msy} f(y) \Pker(\cdot , \rmd y)$.  In what follows
the Dirac mass at $\tilde{x} \in \rset^{d}$ is denoted by $\updelta_{\tilde{x}}$.  For
any $\tilde{x} \in \rset^{\dim}$, we denote
$\tau_{\tilde{x}}: \ \rset^{\dim} \to \rset^{\dim}$ the translation operator
given for any $\tilde{x}' \in \rset^{\dim}$ by
$\tau_{\tilde{x}}(\tilde{x}') = \tilde{x}' - \tilde{x}$. 
The complement of a set $\msa \subset \rset^d$, is denoted by
$\msa^{\complementary}$.  All densities are w.r.t. the Lebesgue
measure (denoted $\Leb$) unless stated otherwise.
 For all
convex and closed set $\msc \subset \rset^d$, we define $\Pi_{\msc}$ the
projection operator onto $\msc$ w.r.t the Euclidean scalar product on
$\rset^d$. For any matrix $\rma \in \rset^{d_1\times d_2}$ with
$d_1, d_2 \in \nset$, we denote $\rma^\top \in \rset^{d_2 \times d_1}$ its
adjoint.

\subsection{Convergence of \pnpula}
\label{sec:convergence_pnpula}

In this section, we fix $\vareps > 0$ and derive quantitative bounds between the
iterates of \pnpula \ and $\pi_{\vareps}$ with respect to the total
variation distance. To address this issue, we first show that \pnpula \ is geometrically
ergodic and establish non-asymptotic bounds between the corresponding Markov
kernel and its invariant distribution. Second, we analyse the distance between
this stationary distribution and $\pi_{\vareps}$. %

For any $\vareps > 0$ we define
$g_{\vareps}: \ \rset^d \times \rset^d \to \coint{0,+\infty}$ for any
$x_1, x_2 \in \rset^d$ by
\begin{equation}
  \label{eq:cond_prior}
  g_{\vareps}(x_1 | x_2) = \left. p^\star(x_1) \exp[- \norm{x_{2} - x_1}^2/ (2\vareps)]  \middle/ \int_{\rset^{\dim}} p^\star(\tilde{x}) \exp[- \norm{x_2 - \tilde{x}}^2/(2 \vareps)] \rmd \tilde{x} \right. \eqsp .
\end{equation}
Note that $g(\cdot | X_{\vareps})$ is the density with respect to the Lebesgue
measure of the distribution of $X$ given $X_{\vareps}$, where $X$ is sampled
according to the prior distribution $\prior$ (with density $p^\star$) and
$X_{\vareps} = X + \vareps^{1/2} Z$ where $Z$ is a Gaussian random variable with
zero mean and identity covariance matrix.  Throughout, this section, we consider
the following assumption on the family of denoising operators
$\ensembleLigne{D_{\vareps}}{\vareps > 0}$ which will ensure that \pnpula \
approximately targets $\posteriorweps$. %

\begin{assumption}[R]
  \label{assum:neural_net}
  We have that
  $\int_{\rset^d} \normLigne{\tilde{x}}^2 p^\star(\tilde{x}) \rmd \tilde{x} <
  +\infty$.  In addition, there exist $\bvareps> 0$, $\Mtt_R \geq 0$ and
  $\Ltt \geq 0$ such that for any $\vareps \in \ocint{0, \bvareps}$,
  $x_{1}, x_{2} \in \rset^d$ and $x \in \cball{0}{R}$ we have
  \begin{equation}
    \label{eq:denoiser_cond}
      \norm{(\Id - D_{\vareps})(x_{1}) - (\Id
        - D_{\vareps})(x_{2})} \leq \Ltt \norm{x_{1} - x_{2}} \eqsp , \qquad  
      \norm{D_{\vareps}(x) - D_{\vareps}^{\star}(x)} \leq \Mtt_R \eqsp ,
  \end{equation}
  where we recall that \begin{equation}
    \label{eq:def_d_star}
    \textstyle{
      D_{\vareps}^{\star}(x_{1}) =  \int_{\rset^{\dim}} \tilde{x} \ g_{\vareps}(\tilde{x} | x_{1}) \rmd \tilde{x} \eqsp .
      }
    \end{equation}
  \end{assumption}
  The Lipschitz continuity condition in \eqref{eq:denoiser_cond} will be useful for establishing the stability and geometric convergence of the Markov chain generated by PnP-ULA. This condition can be explicitly enforced during training by using an appropriate regularization of the neural network weights \cite{ryu2019plug,miyato2018spectral}. Regarding the second condition in
  	\eqref{eq:denoiser_cond}, $\Mtt_R$ is a bound on the error involved in using $D_\vareps$ as an approximation of $D^\star_\vareps$ for images of magnitude $R$ (i.e., for any $x \in \cball{0}{R}$), and it will be useful for bounding the bias resulting from using PnP-ULA for inference w.r.t. $\pi_\vareps$ (recall that the bias vanishes as $\Mtt_R\rightarrow 0$ and $\delta \rightarrow 0$). For denoisers represented by neural networks, one can promote a small value of $\Mtt_R$ during training by using an appropriate loss function. More precisely, consider a neural network
  	$f_{w}: \ \rset^d \to \rset^d$, parameterized by its weights and bias gathered
  	in $w \in \mathcal{W}$ where $\mathcal{W}$ is some measurable space, for any
  	$\vareps > 0$, one could target empirical approximation of a loss of the form $\ell_{\vareps}: \ \mathcal{W} \to \coint{0, +\infty}$
  	given for any $w \in \mathcal{W}$ by
  	$\ell_{\vareps}(w) = \int_{\rset^{\dim} \times \rset^d} \normLigne{x -
  		f_{w}(x_{\vareps})}^2p_{\vareps}^\star(x_{\vareps}) g_{\vareps}(x |
  	x_{\vareps}) \rmd x_{\vareps} \rmd x $.  Note that such a loss is considered
  	in the Noise2Noise network introduced in \cite{lehtinen2018noise2noise}. %
  
  With regards to the theoretical limitations stemming from representing $D_\vareps$ by a deep neural network, universal approximation theorems (see e.g., \cite[Section 4.7]{bach2017breaking}) suggest that $\Mtt_R$ could be arbitrarily low in principle. For a given architecture and training strategy, and if there exists $\tilde{\Mtt}_R \geq 0$ such
  	that
  	$\inf_{w \in \mathcal{W}} \sup_{x \in \cball{0}{R}} \tilde{\Mtt}_R^{-1}
  	\normLigne{f_{w}(x) - D_{\vareps}^{\star}(x)} \} \leq 1$ then the second
  	condition in \eqref{eq:denoiser_cond} holds upon letting
  	$D_{\vareps} = f_{w^{\dagger}}$ for an appropriate choice of weights
  	$w^{\dagger} \in \mathcal{W}$. This last inequality can be established using
  	universal approximation theorems such as \cite[Section
  	4.7]{bach2017breaking}. Moreover, for any other $w \in \mathcal{W}$,
  	$\ell_{\vareps}(w) \geq \int_{\rset^{\dim} \times \rset^d} \normLigne{x -
  		D_{\vareps}^{\star}(x_{\vareps})}^2p_{\vareps}^\star(x_{\vareps})
  	g_{\vareps}(x | x_{\vareps}) \rmd x \rmd x_{\vareps} =
  	\ell^\star_{\vareps}$, since for any $x_{\vareps} \in \rset^d$,
  	$D_{\vareps}^{\star}(x_{\vareps}) = \int_{\rset^{\dim}} \tilde{x} \
  	g_{\vareps}(\tilde{x} | x_{\vareps}) \rmd \tilde{x}$, see
  	\eqref{eq:def_d_star}. Consider ${w}^\dagger \in \mathcal{W}$ obtained after
  	numerically minimizing $\ell_{\vareps}$ and satisfying
  	$\ell_{\vareps}({w}^\dagger) \leq \ell_{\vareps}^\star + \eta$ with
  	$\eta > 0$.  In this case, the following result ensures that
  	\eqref{eq:denoiser_cond} is satisfied with $\Mtt_R$ of order
  	$\eta^{1/(2d +2)}$ for any $R > 0$ and letting
  	$D_{\vareps} = f_{w^{\dagger}}$.
  
      \begin{proposition}
        \label{prop:fun_res}
        Assume that for any $w \in \mathcal{W}$
        \begin{equation}
          \label{eq:bound_nerual}
          \textstyle{
          \int_{\rset^d} (\norm{x}^2 + \norm{f_{w}(x_{\vareps})}^2)  p_{\vareps}^\star(x_{\vareps}) g_{\vareps}(x |
      x_{\vareps}) \rmd x \rmd x_{\vareps}  < +\infty \eqsp .}
        \end{equation}
        Let $R, \eta > 0$ and $\w^\dagger \in \mathcal{W}$ such that
        $\ell_{\vareps}(w^\star) \leq \ell_{\vareps}^\star + \eta$.  In
        addition, assume that 
        \begin{equation}
          \sup_{x_1, x_2 \in \cball{0}{2R}} \defEns{ \norm{x_2 - x_1}^{-1} (\norm{f_{w^{\dagger}}(x_2) - f_{w^{\dagger}}(x_1)} + \norm{D_{\vareps}^{\star}(x_2) - D_{\vareps}^{\star}(x_1)})} < +\infty  \eqsp ,
        \end{equation}
        where $D_{\vareps}^{\star}$ is given in \eqref{eq:def_d_star}.  Then
        there exists $C_R, \bar{\eta}_R\geq 0$ such that if
        $\eta \in \ocint{0, \bar{\eta}_R}$ then for any
        $\tilde{x} \in \cball{0}{R}$,
        $\norm{f_{w^{\dagger}}(\tilde{x}) - D_{\vareps}^{\star}(\tilde{x})} \leq C_R \eta^{1/(2d +2)}$.
      \end{proposition}

      \begin{proof}
        The proof is postponed to \Cref{sec:fun_res_proof}.
      \end{proof}

      Note that \eqref{eq:bound_nerual} is satisfied if for any
      $w \in \mathcal{W}$,
      $\sup_{x \in \rset^{d}} \normLigne{f_w(x)}(1+ \normLigne{x})^{-1} <
      +\infty$ and \Cref{assum:neural_net} holds . 
      
      We recall that \pnpula , see \Cref{algo:pnp-ula}, is given by the following
      recursion: $X_0 \in \rset^d$ and for any $k \in \nset$
\begin{align}
  \label{eq:pnpula}
  X_{k+1} &= X_k + \delta b_{\vareps}(X_k)  + \sqrt{2 \delta} Z_{k+1} \eqsp , \\
  b_{\vareps}(x) &= \nabla \log p(y|x) +  P_{\vareps}(x) + (\prox_{\lambda}(\iota_{\msc})(x) -x)/\lambda \eqsp , \quad   P_{\vareps}(x) = (D_{\vareps}(x) - x) / \vareps \eqsp ,
\end{align} 
where $\delta > 0$ is a step-size, $\vareps, \lambda >0$ are
hyperparameters of the algorithm, $\msc \subset \rset^d$ is a closed convex set,
$\ensembleLigne{Z_k}{k \in \nset}$ a family of i.i.d. Gaussian random variables
with zero mean and identity covariance matrix and
$\prox_{\lambda}(\iota_{\msc})$ the proximal operator of $\iota_{\msc}$ with
step-size $\lambda$, see \cite[Definition 12.23]{bauschke2011convex}, where
$\iota_{\msc}$ is the convex indicator of $\msc$ defined for
$x \in \rset^{\dim}$ by $\iota_{\msc} = +\infty$ if $x \notin \msc$ and $0$ if
$x \in \msc$. Note that for any $x \in \rset^d$ we have
$\prox_{\lambda}(\iota_{\msc})(x) = \Pi_{\msc}(x)$, where
$\Pi_{\msc}$ is the projection onto $\msc$. 

In what follows, for any $\delta >0$ and
$\msc \subset \rset^{\dim}$ closed and convex, we denote by
$\Rker_{\vareps, \delta}: \ \rset^d \times \mcb{\rset^d} \to \ccint{0,1}$
the Markov kernel associated with the recursion \eqref{eq:pnpula} and given for
any $x \in \rset^d$ and $\msa \in \mcb{\rset^d}$ by
\begin{equation}
  \label{eq:pnpula_kernel}
  \Rker_{\vareps, \delta}(x, \msa) = (2 \uppi)^{-d/2} \int_{\rset^{\dim}} \1_{\msa}(x + \delta b_{\vareps}(x) + \sqrt{2\delta} z) \exp[- \norm{z}^2/2] \rmd z \eqsp .
\end{equation}
Note that for ease of notation, we do not explicitly highlight the dependency of
$\Rker_{\vareps, \delta}$ and $b_{\vareps}$ with respect to the hyperparameter
$\lambda > 0$ and $\msc$.

Here we consider the case where
  $x \mapsto \log p(y|x)$ %
  satisfies a one-sided Lipschitz condition, \ie \ we consider
  the following condition.
\begin{assumption}
  \label{assum:one_sided}
  There exists $\mtt \in \rset$ such that for any $x_1, x_2 \in \rset^d$ we have
  \begin{equation}
    \langle \nabla \log p(y|x_2) - \nabla \log p(y|x_1) , x_2 - x_1 \rangle \leq -\mtt \norm{x_2 - x_1}^2 \eqsp . 
  \end{equation}
\end{assumption}
We refer to the supplementary material
\Cref{sec:strongly-log-concave} for refined
convergence rates in the case where $x \mapsto \log p(y|x) $ is strongly
$\mtt$-concave.  Note that if \Cref{assum:one_sided} is satisfied with
$\mtt > 0$ then $x \mapsto \log p(y|x)$ is $\mtt$-concave. Assume
\Cref{assum:post} then \Cref{assum:one_sided} holds for $\mtt =
-\Ltt_y$. However, it is possible that $\mtt > -\Ltt_y$ which leads to better
convergence rates for \pnpula . As a result even when \Cref{assum:post} holds we
still consider \Cref{assum:one_sided}. In order to deal with
\Cref{assum:one_sided} in the case where $\mtt \leq 0$, we set
$\msc \subset \rset^d$ to be some convex compact set fixed by the user.  Doing
so, we ensure the stability of the Markov chain. The choice of $\msc$ in
practice is discussed in \Cref{sec:experimental}. In our imaging experiments, we recall that for any $x \in \rset^d$ we have,
$p(y|x) \propto \exp[-\norm{\rmA x - y}^2/(2 \sigma^2) ]$. If $\rmA$ is not invertible
then $x \mapsto \log p(y|x)$ is not $\mtt$-concave with $\mtt > 0$.  This is the
case, in our deblurring experiment when the convolution kernel has zeros in the Fourier domain.

We start with the following result which
ensures that the Markov chain \eqref{eq:pnpula} is geometrically
ergodic under \rref{assum:neural_net} for the Wasserstein metric
$\wassersteinD[1]$ and in $V$-norm for $V: \ \rset^{\dim} \to \coint{1, +\infty}$ given for any $x \in \rset^{\dim}$ by
\begin{equation}
  \label{eq:V_def}
V(x) = 1 + \norm{x}^2 \eqsp .  
\end{equation}

\begin{proposition}
  \label{prop:ergo_C}
  Assume \rref{assum:post}, \tup{\Cref{assum:neural_net}($R$)} for some $R > 0$
  and \rref{assum:one_sided}. Let $\lambda >0$,
  $\vareps \in \ocint{0, \vareps_0}$ such that
  $2\lambda (\Ltt_y + \Ltt /\vareps - \min(\mtt, 0)) \leq 1$ and
  $\bdelta = (1/3)(\Ltt_y + \Ltt /\vareps + 1 / \lambda)^{-1}$.  Then for any
  $\msc \subset \rset^d$ convex and compact with $0 \in \msc$, there exist
  $A_{1, \msc} \geq 0$ and $\rho_{1, \msc} \in \coint{0, 1}$ such that for any
  $\delta \in \ocintLigne{0, \bdelta}$, $x_1, x_2 \in \rset^{\dim}$ and
  $k \in \nset$ we have
  \begin{align}
    \Vnorm{\updelta_{x_1} \Rker_{\vareps, \delta}^k-  \updelta_{x_2} \Rker_{\vareps, \delta}^k} &\leq A_{1, \msc} \rho_{1, \msc}^{k \delta} (V^2(x_1) + V^2(x_2))  \eqsp , \\
                                                                                           \wassersteinD[1](\updelta_{x_1} \Rker_{\vareps, \delta}^k, \updelta_{x_2} \Rker_{\vareps, \delta}^k) &\leq A_{1, \msc} \rho_{1, \msc}^{k \delta} \norm{x_1-x_2}  \eqsp ,
  \end{align}
  where $V$ is given in \eqref{eq:V_def}.
\end{proposition}

\begin{proof}
  The proof is postponed to \Cref{prop:ergo:proof}.
\end{proof}

The constants $A_{1, \msc}$ and $\rho_{1, \msc}$ do not depend on the dimension
$d$ but only on the parameters $\mtt, \Ltt, \Ltt_y, \vareps$ and $\msc$. Note
that a similar result can be established for $\wassersteinD[p]$ for any
$p \in \nsets$ instead of $\wassersteinD[1]$. Under the conditions of
\Cref{prop:ergo_C} we have for any $\nu_1, \nu_2 \in \Pens_1(\rset^{\dim})$
\begin{align}
  \label{eq:contract_general}
      \Vnorm{\nu_1 \Rker_{\vareps, \delta}^k-  \nu_2 \Rker_{\vareps, \delta}^k} &\leq A_{1, \msc} \rho_{1, \msc}^{k \delta} \parenthese{\int_{\rset^{\dim}} V^2(\tilde{x})  \rmd \nu_1(\tilde{x}) + \int_{\rset^{\dim}} V^2(\tilde{x})  \rmd \nu_2(\tilde{x}) }  \eqsp , \\
                                                                                           \wassersteinD[1](\nu_1 \Rker_{\vareps, \delta}^k, \nu_2 \Rker_{\vareps, \delta}^k) &\leq A_{1, \msc} \rho_{1, \msc}^{k \delta} \parenthese{\int_{\rset^{\dim}} \norm{\tilde{x}} \rmd \nu_1(\tilde{x}) + \int_{\rset^{\dim}} \norm{\tilde{x}} \rmd \nu_2(\tilde{x})}  \eqsp .
\end{align}

First, $(\Pens_1(\rset^{\dim}), \wassersteinD[1])$ is a complete metric space
\cite[Theorem 6.18]{villani2009optimal}. Second, for any
$\delta \in \ocintLigne{0, \bdelta}$, there exists $\rmm \in \nsets$ such that
$\rmf^{\rmm}$ is contractive with
$\rmf: \ \Pens_1(\rset^{\dim}) \to \Pens_1(\rset^{\dim})$ given for any
$\nu \in \Pens_1(\rset^{\dim})$ by $\rmf(\nu) = \nu \Rker_{\vareps, \delta}$
using \Cref{prop:ergo_C}. Therefore we can apply the Picard fixed point theorem
and we obtain that $\Rker_{\vareps, \delta}$ admits an invariant probability
measure $\posteriorwepsdelta \in \Pens_1(\rset^{\dim})$.

Therefore, since $\posteriorwepsdelta$ is an invariant
probability measure for $\Rker_{\vareps, \delta}$ and
$\posteriorwepsdelta \in \Pens_1(\rset^{\dim})$, using
\eqref{eq:contract_general}, we have for any
$\nu \in \Pens_1(\rset^{\dim})$
\begin{align}
      \Vnorm{\nu \Rker_{\vareps, \delta}^k-  \posteriorwepsdelta} &\leq A_{1, \msc} \rho_{1, \msc}^{k \delta} \parenthese{\int_{\rset^{\dim}} V^2(\tilde{x}) \rmd \nu(\tilde{x}) + \int_{\rset^{\dim}} V^2(\tilde{x}) \rmd \posteriorwepsdelta(\tilde{x}) }  \eqsp , \\
                                                                                           \wassersteinD[1](\nu \Rker_{\vareps, \delta}^k, \posteriorwepsdelta) &\leq A_{1, \msc} \rho_{1, \msc}^{k \delta} \parenthese{\int_{\rset^{\dim}} \norm{\tilde{x}} \rmd \nu(\tilde{x}) + \int_{\rset^{\dim}} \norm{\tilde{x}} \rmd \posteriorwepsdelta(\tilde{x})}  \eqsp .
\end{align}
Combining this result with the fact that for any $t \geq 0$,
$(1-\rme^{-t})^{-1} \leq 1 + t^{-1}$, we get that for any $n \in \nsets$ and 
$h: \ \rset^{\dim} \to \rset$ measurable such that $\sup_{x \in \rset^d} \{ (1 + \norm{x}^2)^{-1} \abs{h (x)}\} < +\infty$ 
\begin{align}& \abs{n^{-1}\sum_{k=1}^n\expeLigne{h(X_k)} - \int_{\rset^d} h(\tilde{x})
  \rmd \posteriorwepsdelta(\tilde{x}) } \\
  & \qquad \qquad \leq A_{1, \msc} (\bdelta +
  \log^{-1}(1/\rho_{1, \msc})) \left. \parenthese{V^2(x) + \int_{\rset^{\dim}}
      V^2(\tilde{x}) \rmd \posteriorwepsdelta(\tilde{x})} \middle/ (n \delta)
  \right. \eqsp ,
  \end{align}
  where $(X_k)_{k \in \nset}$ is the Markov chain given by \eqref{eq:pnpula} with starting point $X_0 = x \in \rset^d$.

  In the rest of this section we evaluate how close the invariant measure $\pi_{\vareps, \delta}$ is
  to $\pi_{\vareps}$.  Our proof will rely on the following assumption which is
  necessary to ensure that $x \ \mapsto \log p_{\vareps}^\star(x)$ has Lipschitz
  gradients, see \Cref{prop:reg_p_eps}.
  \begin{assumption}
    \label{assum:cov}
    For any $\vareps > 0$, there exists $\Ktt_{\vareps} \geq 0$ such that for any $x \in \rset^d$,
    \begin{align}
      &\int_{\rset^d} \norm{\tilde{x} - \int_{\rset^d} \tilde{x}' g_{\vareps}(\tilde{x}'| x) \rmd \tilde{x}'}^2 g_{\vareps}(\tilde{x}| x) \rmd \tilde{x} \leq \Ktt_{\vareps} \eqsp ,
    \end{align}
    with $g_{\vareps}$ given in \eqref{eq:cond_prior}.
  \end{assumption}
  We emphasize that \Cref{assum:cov} is not needed to establish the convergence
  of the Markov chain. However, we impose it in order to compare the stationary
  distribution of \pnpula \ with the target distribution
  $\posteriorweps$. Depending on the prior distribution density $p^\star$,
  \Cref{assum:cov} may be checked by hand. %
  Finally, note that \Cref{assum:cov} can be
  extended to cover the case where the prior distribution $\prior$ does not
  admit a density with respect to the Lebesgue measure.

  In the following proposition, we show that we can control the
  distance between $\posteriorwepsdelta$ and $\posteriorweps$
  based on the previous observations.

\begin{proposition}
  \label{prop:bias_C}
  Assume \rref{assum:post}, \tup{\Cref{assum:neural_net}($R$)} for some $R > 0$,
  \rref{assum:one_sided} and \rref{assum:cov}. Moreover, let
  $\vareps \in \ocint{0, \vareps_0}$ and assume that
  $\int_{\rset^{\dim}} (1+\norm{\tilde{x}}^4) p_{\vareps}^\star(\tilde{x}) \rmd
  \tilde{x} < + \infty$. Let $\lambda >0$ such that
  $2\lambda (\Ltt_y + ( /\vareps)\max(\Ltt, 1 + \Ktt_{\vareps}/ \vareps)
  -\min(\mtt,0)) \leq 1$ and
  $\bdelta = (1/3)(\Ltt_y + \Ltt /\vareps + 1 / \lambda)^{-1}$. Then for any
  $\delta \in \ocintLigne{0, \bdelta}$ and $\msc$ convex and compact with
  $0 \in \msc$, $\Rker_{\vareps, \delta}$ admits an invariant probability
  measure $\posteriorwepsdeltac$. In addition, there exists $B_0 \geq 0$ such
  that for any $\msc$ convex compact with $\cball{0}{R_{\msc}} \subset \msc$ and
  $R_{\msc} > 0$, there exists $B_{1, \msc} \geq 0$ such that for any
  $\delta \in \ocintLigne{0, \bdelta}$
  \begin{equation}
    \Vnorm{\posteriorwepsdeltac- \posteriorweps} \leq B_0 R_{\msc}^{-1} + B_{1, \msc} (\delta^{1/2} + \Mtt_R + \exp[-R])  \eqsp ,
  \end{equation}
  where $V$ is given in \eqref{eq:V_def}.  
\end{proposition}

\begin{proof}
  The proof is postponed to \Cref{prop:bias:proof}.
\end{proof}

We now combine \Cref{prop:ergo_C} and \Cref{prop:bias_C} in order to control the
bias of the Monte Carlo estimator obtained using \pnpula . %
In the supplementary material
\Cref{sec:posterior_approx} we also provide bounds
on
$\absLigne{n^{-1}\sum_{k=1}^n \expeLigne{h(X_k)} - \int_{\rset^{\dim}} h(\tilde{x}) \rmd
  \posteriorw(\tilde{x}) }$ by controlling $\Vnorm{\pi - \pi_\vareps}$.

\begin{proposition}
  \label{prop:bias_control_final_C}
  Assume \tup{\Cref{assum:post}}, \tup{\Cref{assum:neural_net}($R$)} for some
  $R > 0$, \rref{assum:one_sided} and \rref{assum:cov}. Moreover, let $>0$,
  $\vareps \in \ocint{0, \vareps_0}$ and assume that
  $\int_{\rset^{\dim}} (1+\norm{\tilde{x}}^4) p_{\vareps}^\star(\tilde{x})
  \rmd \tilde{x} < + \infty$.  Let $\lambda >0$ such that
  $2\lambda (\Ltt_y + (1/\vareps) \max(\Ltt, 1 + \Ktt_{\vareps}/ \vareps) - \min(\mtt,0))
  \leq 1$ and
  $\bdelta = (1/3)(\Ltt_y +  \Ltt /\vareps + 1 / \lambda)^{-1}$. Then
  there exists $C_{1, \vareps} > 0$ such that for any $\msc$ convex compact
  with $\cball{0}{R_{\msc}} \subset \msc$ and $R_{\msc} > 0$ there exists
  $C_{2, \vareps}$ such that for any $h : \ \rset^{\dim} \to \rset$ measurable
  with $\sup_{x \in \rset^d} \{\abs{h(x)}(1 + \norm{x}^{2})^{-1} \} \leq 1$, $n \in \nsets$,
  $\delta \in \ocintLigne{0, \bdelta}$ we have
  \begin{multline}
    \abs{n^{-1}\sum_{k=1}^n \expe{h(X_k)} - \int_{\rset^{\dim}} h(\tilde{x}) \rmd \posteriorweps(\tilde{x}) } \\\leq \defEns{C_{1, \vareps} R_{\msc}^{-1} + C_{2, \vareps, \msc} (\delta^{1/2} + \Mtt_R + \exp[-R] + (n \delta)^{-1})  } (1 + \norm{x}^4)  \eqsp .
  \end{multline}

\end{proposition}

\begin{proof}
  The proof is straightforward combining \Cref{prop:ergo_C} and \Cref{prop:bias_C}.
\end{proof}

\subsection{Convergence guarantees for P\pnpula}
\label{sec:proj-altern}

We now study the Projected \emph{Plug \& Play} Unadjusted Langevin Algorithm
(P\pnpula ). It is given by the following recursion: $X_0 \in \msc$ and for any
$k \in \nset$
\begin{align}
  \label{eq:ppnpula}
  X_{k+1} &= \Pi_\msc(X_k + \delta b_{\vareps}(X_k)  + \sqrt{2 \delta} Z_{k+1}) \eqsp , \\
  b_{\vareps}(x) &= \nabla \log p(y|x) +  P_\vareps(x)   \eqsp , \quad   P_\vareps(x) = (D_{\vareps}(x) - x) / \vareps \eqsp ,
\end{align}
where $\delta > 0$ is a step-size, $\vareps >0$ is
an hyperparameter of the algorithm, $\msc \subset \rset^d$ is a closed convex set,
$\ensembleLigne{Z_k}{k \in \nset}$ a family of i.i.d. Gaussian random variables
with zero mean and identity covariance matrix and where $\Pi_{\msc}$ is the
projection onto $\msc$.
In what follows, for any $\delta >0$ and
$\msc \subset \rset^{\dim}$ closed and convex, we denote by
$\Qker_{\vareps, \delta}: \ \rset^d \times \mcb{\rset^d} \to \ccint{0,1}$
the Markov kernel associated with the recursion \eqref{eq:ppnpula} and given for
any $x \in \rset^d$ and $\msa \in \mcb{\rset^d}$ by
\begin{equation}
  \label{eq:ppnpula_kernel}
  \Qker_{\vareps, \delta}(x, \msa) = (2 \uppi)^{-d/2} \int_{\rset^{\dim}} \1_{\Pi_{\msc}^{-1}(\msa)}(x + \delta b_{\vareps}(x) + \sqrt{2\delta} z) \exp[- \norm{z}^2/2] \rmd z \eqsp .
\end{equation}
Note that for ease of notation, we do not explicitly highlight the dependency of
$\Qker_{\vareps, \delta}$ and $b_{\vareps}$ with respect to the hyperparameter $\msc$. 

First, we have the following result which ensures that P\pnpula \ is
geometrically ergodic for all step-sizes.

\begin{proposition}
  \label{prop:ergo_C_proj}
  Assume \rref{assum:post}, \tup{\Cref{assum:neural_net}($R$)} for some $R >
  0$. Let $\lambda, \vareps, \bdelta >0$.  Then for any $\msc \subset \rset^d$
  convex and compact with $0 \in \msc$, there exist $\tilde{A}_{\msc} \geq 0$
  and $\tilde{\rho}_{\msc} \in \coint{0, 1}$ such that for any
  $\delta \in \ocintLigne{0, \bdelta}$, $x_1, x_2 \in \msc$ and $k \in \nset$ we
  have
  \begin{align}
    \tvnormLigne{\updelta_{x_1} \Qker_{\vareps, \delta}^k-  \updelta_{x_2} \Qker_{\vareps, \delta}^k} &\leq \tilde{A}_{\msc} \tilde{\rho}_{\msc}^{k \delta}   \eqsp .
  \end{align}
\end{proposition}

\begin{proof}
    The proof is postponed to \Cref{prop:ergo_C_proj:proof}.
\end{proof}

In particular $\Qker_{\vareps, \delta}$ admits an invariant probability measure
$\pi_{\vareps, \delta}^\msc$.  The next proposition ensures that for small
enough step-size $\delta$ the invariant measures of \pnpula \ and P\pnpula \ are
close if the compact convex set $\msc$ has a large diameter.

\begin{proposition}
  \label{prop:disc_invariant_K}
  Assume \rref{assum:post}, 
  \tup{\Cref{assum:neural_net}($R$)} for some $R > 0$ and
  \rref{assum:one_sided}. In addition, assume that there exists $\tilde{\mtt}, c > 0$ such that for $\msc = \rset^d$ and for any $\vareps > 0$ and $x \in \rset^d$, $\langle b_\vareps(x), x \rangle \leq -\tilde{\mtt} \norm{x}^2 + c$. Let $\lambda >0$,
  $\vareps \in \ocint{0, \vareps_0}$ such that
  $2\lambda (\Ltt_y +  \Ltt /\vareps - \min(\mtt, 0)) \leq 1$.  Then there
  exist $\bar{A} \geq 0$ and $\eta, \bdelta > 0$ such that for any
  $\msc \subset \rset^d$ convex and compact with $0 \in \msc$ and
  $\cball{0}{R_{\msc} / 2} \subset \msc \subset \cball{0}{R_{\msc}}$ and
  $\delta \in \ocintLigne{0, \bdelta}$ we have
  \begin{equation}
    \tvnormLigne{\pi_{\vareps, \delta} - \pi_{\vareps, \delta}^\msc} \leq \bar{A} \exp[-\eta R_{\msc}] \eqsp ,   
  \end{equation}
  where $\pi_{\vareps, \delta}$ is the invariant measure of $\Rker_{\vareps, \delta}$ and
  $\pi_{\vareps, \delta}^\msc$ is the invariant measure of $\Qker_{\vareps, \delta}$.
\end{proposition}

\begin{proof}
  The proof is postponed to \Cref{prop:disc_invariant_K:proof}.
\end{proof}

It is worth mentioning at this point that in our experiments, see \Cref{sec:experimental}, the probability of the iterates $(X_n)_{n \in \nset}$ leaving $\msc$ with \pnpula \ or with P\pnpula \ is so low that the projection constraint is not activated. As a result, if implemented with the same step-size both algorithms produce the same results. We do not suggest completely removing the constraints as this is important to theoretically guarantee the geometric ergodicity of the algorithms.

Regarding the choice of the step-size, we observe that the bound $\bdelta = (1/3)(\Ltt_y + \Ltt /\vareps + 1 / \lambda)^{-1}$ used in \pnpula \ is
conservative and our experiments suggest that \pnpula \ is stable for larger step-sizes.

	\section{Experimental study}\label{sec:experimental}
	This section illustrates the behaviour of \pnpula \ and P\pnpula \ with two classical imaging inverse problems: non-blind image deblurring  and inpainting. For these two problems, we first analyse in detail the convergence  of the Markov chain generated by \pnpula \ for different test images. This is then followed by a comparison between the MMSE Bayesian point estimator, as calculated by using \pnpula \ and P\pnpula \, and the MAP estimator provided by the recent \pnpsgd~method \cite{laumont2021maximum}. We refer the reader to \cite{laumont2021maximum} for comparisons with PnP-ADMM \cite{ryu2019plug}. To simplify comparisons, for all experiments and algorithms, the operator $D_{\vareps}$ is chosen as the pretrained denoising neural network introduced in \cite{ryu2019plug}, for which $(D_{\vareps}-\Id)$ is $\Ltt$-Lipschitz with $\Ltt<1$.

For the deblurring experiments,  the observation model takes the form 
\begin{equation}\label{eq:inverse_problem} 
y = \rmA x + n \eqsp ,
\end{equation}
where $x \in \mathbb{R}^d$ is the unknown original image, $y \in \mathbb{R}^\dimY$ the observed image, $n$ is a realization of a Gaussian i.i.d. centered noise with variance $\sigma^2 \Id$ (with $\sigma^2 = (1/255)^2$), and $\rmA$ is a $9 \times 9$ box blur operator. The log-likelihood for this case writes $\log p(y|x) = -\|\rmA x-y \|^2/(2\sigma^2)$. 

In the inpainting experiments, we seek to recover $x \in \mathbb{R}^d$  from  $y = \rmA x$ where the matrix $\rmA$ is a $m\times d$ matrix containing $m$ randomly selected rows of the $d\times d$ identity matrix. We focus on a case where $80\%$ of the image pixels are hidden and the observed pixels are measured without any noise. Because the posterior density for $ x|y$ is degenerate, we run PnP-ULA on the posterior $\tilde x|y$ where $\tilde x := \rmP x \in \R^n$ denotes the vector of
$n=d-m$ unobserved pixels of $x$, and map samples to the pixel space by using the affine mapping
$f_y:\R^n \to \R^d$ defined for any $\tilde x  \in \rset^n$ and $y \in \rset^m$ by
\[f_y(\tilde x) = \rmP^\top \tilde x + \rmA^\top y. \] Note that we can write the log-posterior $\tilde{U}_\vareps (\tilde x)= -\log p_\vareps(\tilde{x}|y)$ on the set $\rset^n$ of hidden pixels in terms of $f_y$ and the log-prior $U_\vareps(x) = -\log p_\vareps(x)$ on the set $\rset^d$:
$$\tilde{U}_\vareps = U_\vareps \circ f_y .$$
Using the chain rule and Tweedie's formula, we have that for any $x \in \rset^d$ and $y \in \rset^m$
\begin{equation}
  b_\vareps(\tilde x)  = -\nabla \tilde U_\vareps (\tilde x) = -\rmP \nabla U_\vareps (f_y (\tilde x)) = (1/\vareps) \rmP  (D_\vareps - \Id)( f_y(\tilde x)) \eqsp . 
\end{equation}
Since $\rmP$ and $f_y$ are 1-Lipschitz, $b_\vareps = -\nabla \tilde{U}_{\vareps}$ is also Lipschitz with constant $ \tilde{\Ltt} \leq (\Ltt/\vareps)$. %

\Cref{fig:original_images} shows the six test images of size  $256 \times 256$ pixels that were used in the experiments. We have selected these six images for their diversity in composition, content and level of detail (some images are predominantly composed of piece-wise constant regions, whereas others are rich in complex textures). This diversity will highlight strengths and limitations of the chosen denoiser as an image prior. \Cref{fig:blurred_images} depicts the corresponding blurred images and  \Cref{fig:toipaint_images} the images to inpaint.

\begin{figure}[h]
    \centering
    \begin{tabular}{ccc}
    \includegraphics[width=0.25\textwidth]{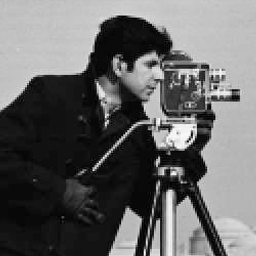} &
    \includegraphics[width=0.25\textwidth]{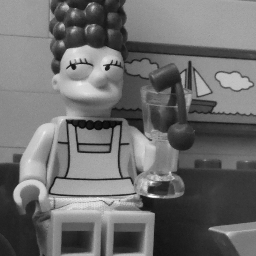} &
    \includegraphics[width=0.25\textwidth]{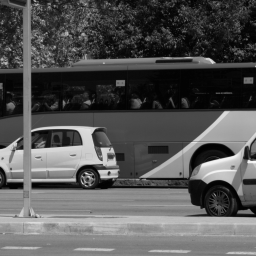}  \\
    \texttt{Cameraman.} & \texttt{Simpson.} & \texttt{Traffic.} \\
    \includegraphics[width=0.25\textwidth]{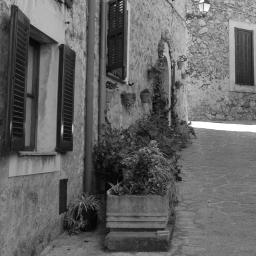} &
    \includegraphics[width=0.25\textwidth]{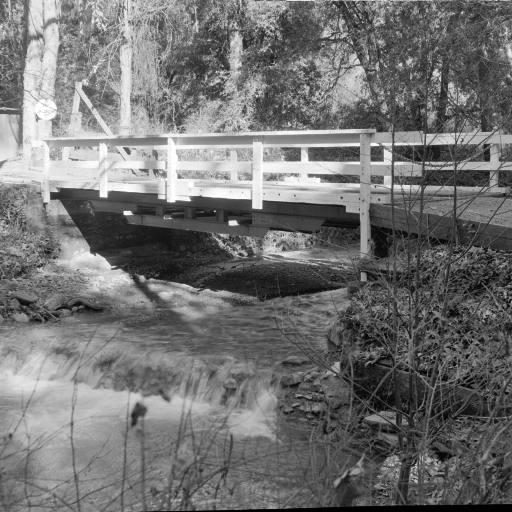} &
    \includegraphics[width=0.25\textwidth]{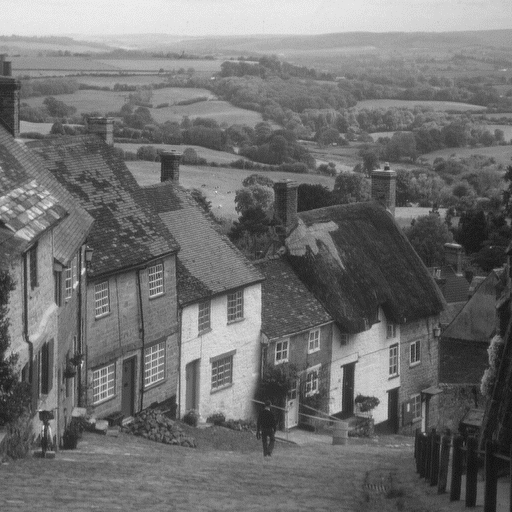} \\
    \texttt{Alley.} & \texttt{Bridge.} & \texttt{Goldhill.}
    \end{tabular}
	\caption{Original images used for the deblurring and inpainting experiments.}
	\label{fig:original_images}
\end{figure}

\begin{figure}[h]
    \centering
    \begin{tabular}{ccc}
        \includegraphics[width=0.25\textwidth]{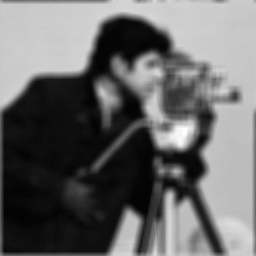} &
        \includegraphics[width=0.25\textwidth]{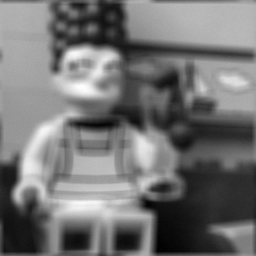} &
        \includegraphics[width=0.25\textwidth]{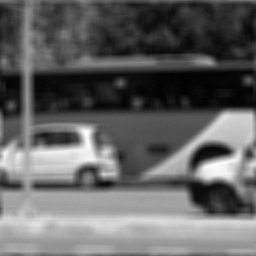}
        \\
        \begin{small}PSNR=20.30/SSIM=0.70\end{small} & \begin{small}PSNR=22.44/SSIM=0.66\end{small} & \begin{small}PSNR=20.34/SSIM=0.49\end{small}
         \\
        \includegraphics[width=0.25\textwidth]{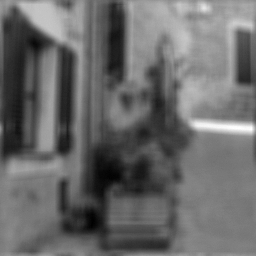} &  
        \includegraphics[width=0.25\textwidth]{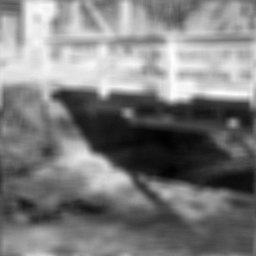} &
        \includegraphics[width=0.25\textwidth]{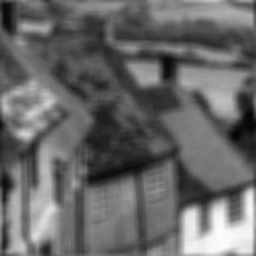}
        \\
        \begin{small}PSNR=22.64/SSIM=0.46\end{small} & \begin{small}PSNR=21.84/SSIM=0.49\end{small} & \begin{small}PSNR=22.61/SSIM=0.45\end{small} 
    \end{tabular}
    \caption{Images of \Cref{fig:original_images}, blurred using a $9 \times 9$-box-filter operator and corrupted by an additive Gaussian white noise with standard deviation $\sigma=1/255$.}
	\label{fig:blurred_images}
\end{figure}

\begin{figure}[h]
    \centering
    \begin{tabular}{ccc}
        \includegraphics[width=0.25\textwidth]{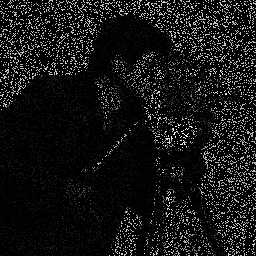} &
        \includegraphics[width=0.25\textwidth]{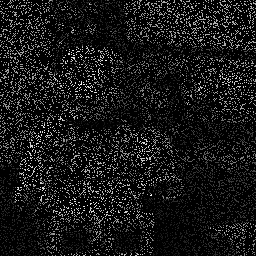} &
        \includegraphics[width=0.25\textwidth]{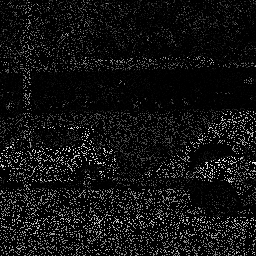} 
        \\
        \begin{small}PSNR=6.69/SSIM=0.11\end{small} & \begin{small}PSNR=7.43/SSIM=0.04\end{small} & \begin{small}PSNR=8.35/SSIM=0.09 \end{small}
        \\
        \includegraphics[width=0.25\textwidth]{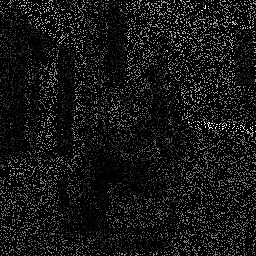} &  
        \includegraphics[width=0.25\textwidth]{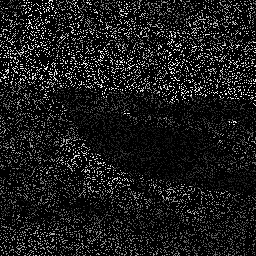} &
        \includegraphics[width=0.25\textwidth]{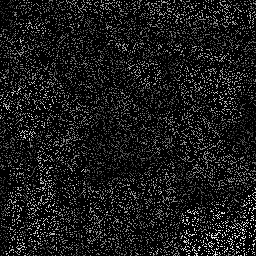}
        \\
        \begin{small}PSNR=8.27/SSIM=0.004\end{small} & \begin{small}PSNR=5.71/SSIM=0.004\end{small} & \begin{small}PSNR=6.61/SSIM=0.03\end{small}
    \end{tabular}
    \caption{Images of \Cref{fig:original_images}, with 80\% missing pixels.}
	\label{fig:toipaint_images}
\end{figure}

\subsection{Implementation guidelines and parameter setting}
In the following, we provide some simple and robust rules in order to set the parameters of the different algorithms, in particular the discretization step-size $\delta$ and the tail regularization parameter $\lambda$.

\paragraph{Choice of the denoiser} The theory presented in \Cref{sec:theory} requires that $D_\vareps$ satisfies \Cref{assum:neural_net}($R$). As default choice, we recommend using a pretrained denoising neural network
such as the one described in \cite{ryu2019plug}. The Lipschitz constant of the
network is controlled during training by using spectral normalization and therefore
the first condition of \Cref{assum:neural_net}($R$) holds. Moreover, the loss function used to train the
network is given by $\ell_\vareps$ as introduced in
\Cref{sec:convergence_pnpula}. Therefore, under the conditions of
\Cref{prop:fun_res}, we get that the second condition of
\Cref{assum:neural_net}($R$) holds.%

\paragraph{Step-size $\delta$} The parameter $\delta$ controls the asymptotic
accuracy of PnP-ULA and PPnP-ULA, as well as the speed of convergence to
stationarity. This leads to the following bias-variance trade-off. For large
values of $\delta$, the Markov chain has low auto-correlation and converges
quickly to its stationary regime. Consequently, the Monte Carlo estimates
computed from the chain exhibit low asymptotic variance, at the expense of some
asymptotic bias. On the contrary, small values of $\delta$ produce a Markov
chain that explores the parameter space less efficiently, but more
accurately. As a result, the asymptotic bias is smaller, but the variance is
larger. In the context of inverse problems that are high-dimensional and
ill-posed, properly exploring the solution space can take a large number of
iterations. For this reason, we recommend using large values of $\delta$, at the
expense of some bias. In addition, in \pnpula, $\delta$ is also subject to a
numerical stability constraint related to the inverse of the Lipschitz constant
of $b_\vareps(x) = \nabla \log p_{\vareps}(x|y)$; namely, we require
$\delta<(1/3)\operatorname{Lip}(b_\vareps)^{-1}$ where 
\[
\operatorname{Lip}(b_\vareps) = \begin{cases}
\alpha \Ltt/\vareps + 1/\lambda & \text{for the inpainting problem} \\
\alpha \Ltt /\vareps + \Ltt_y + 1/\lambda& \text{otherwise}
\end{cases}
\]
 where ${\Ltt}$
and $\Ltt_y$ are respectively the Lipschitz constant of the denoiser residual
$(D_{\vareps} -\Id)$ and the Lipschitz constant of the log-likelihood
gradient. In our experiments, $\Ltt=1$ and
$\Ltt_y = \|\rmA^\top \rmA \| / \sigma^2$, so we choose $\delta$ just below the
upper bound
$\delta_{th} = 1/3(\operatorname{Lip}(b_\vareps))^{-1}$
where
$\rmA^\top$ is the adjoint of $\rmA$. For PPnP-ULA, we set
$\delta <(\Ltt /\vareps + \Ltt_y)^{-1}$ (resp. $\delta <(\Ltt/\vareps)^{-1}$ for inpainting)  to prevent excessive bias.

\paragraph{Parameter $\lambda$} The parameter $\lambda$ controls the tail behaviour of the target density. As previously explained, it must be set so that the tails of the target density decay sufficiently fast to ensure convergence at a geometric rate, a key property for guaranteeing that the Monte Carlo estimates computed from the chain are consistent and subject to a Central Limit Theorem with the standard $\mathcal{O}(\sqrt{k})$ rate. More precisely, we require $\lambda \in (0,1/2(\Ltt/\vareps + 2\Ltt_y ))$. Within this admissible range, if $\lambda$ is too small this limits the maximal $\delta$ and leads to a slow Markov chain. For this reason, we recommend setting $\lambda$ as large as possible below $(2\Ltt/\vareps + 4\Ltt_y )^{-1}$. %

\paragraph{Other parameters} The compact set $\msc$ is defined as $\msc = \ccint{-1,2}^d$, even if in practice no samples where generated outside of $\msc$ in all our experiments, which suggests that the tail decay conditions hold without explicitly enforcing them.
In all our experiments, we set the noise level of the denoiser $D_{\vareps}$ to $\vareps= (5/255)^2$. The initialization $X_0$ can be set to a random vector. In our experiments (where $m=d$), we chose $X_0 = y$ in order to reduce the number of burn-in iterations. For $m \neq d$ we could use $X_0 = \rmA^{\top}  y$ instead. Concerning the regularization parameter $\alpha$, by default we set $\alpha = 1$, but in some cases it is possible to marginally improve the results by fine tuning it.
All algorithms are implemented using Python and the PyTorch library, and run on an Intel Xeon CPU E5-2609 server with a Nvidia Titan XP graphic card or on Idris' Jean-Zay servers featuring Intel Cascade Lake 6248 CPUs with a single Nvidia Tesla V100 SXM2 GPU. Reported running times correspond to the Xeon + Titan XP configuration.

\subsection{Convergence analysis of PnP-ULA in non-blind image deblurring and inpainting} 
\label{sec:cv_study_deblurring}

When using a sampling algorithm such as \pnpula\ on a new problem, it is essential to check that the  state space is correctly explored. 
In order to provide a thorough convergence study, we first run the algorithm for $25\times 10^6$ iterations. We use  a burn-in period of $2.5\times 10^6$ iterations, and consider only the samples computed after this burn-in period to study the  Markov chain in close-to-stationary regime. In section~\ref{sec:deblurring_inpainting}, we will see that much less iterations are required if the goal is only to compute point estimators with \pnpula.  %
For simplicity, the algorithm is always initialized with the observation $y$ in our experiments with \pnpula \ (for inpainting, this means that unknown pixels are initialized to the value $0$).  

There is no fully comprehensive way to empirically characterise the convergence properties of a high-dimensional Markov chain, as different statistics computed from the same chain align differently with the eigenfunctions of the Markov kernel and hence exhibit different convergence speeds. In problems of small dimension, we would calculate and analyse the $d$-dimensional multivariate autocorrelation function (ACF) of the Markov chain, but this is not feasible in imaging problems. In problems of moderate dimension, one could characterise the range of convergence speeds by first estimating the posterior covariance matrix (which, for $256\times 256$ images, would be a $256^2\times 256^2$ matrix) and then performing a principal component analysis on this matrix to identify the directions with smallest and largest uncertainty, as these would provide a good indication of the subspaces where the chain converges the fastest and the slowest. However, computing the posterior covariance matrix is also not possible in imaging problems because of the dimensionality involved. Here we focus on approximations of the posterior covariance which make sense for the particular inverse problem we study. More precisely, we use the  diagonalization basis of the inverse operator, {\em i.e.} the Fourier basis for the deblurring experiments, and the basis formed by the unknown pixels for the inpainting experiments. Under the assumption that the posterior covariance is mostly determined by the likelihood, this strategy allows broadly identifying the linear statistics that converge fastest and slowest, without requiring the estimation and manipulation of prohibitively large matrices.

\begin{figure}[h!]
	\centering
	\begin{tabular}{ccc}
	   
		\includegraphics[width=0.28\textwidth]{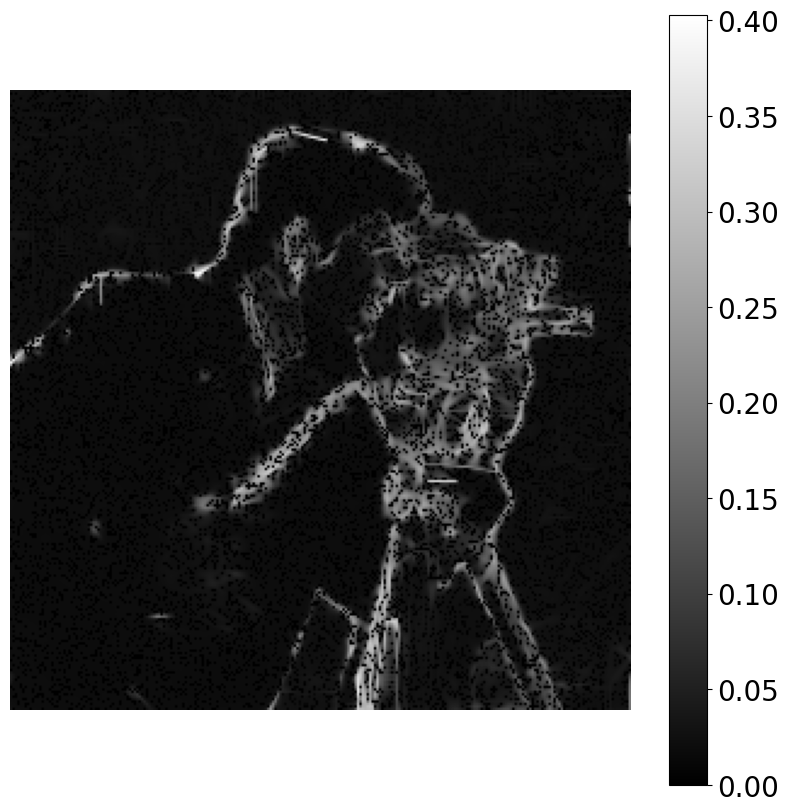} &
		
		\includegraphics[width=0.28\textwidth]{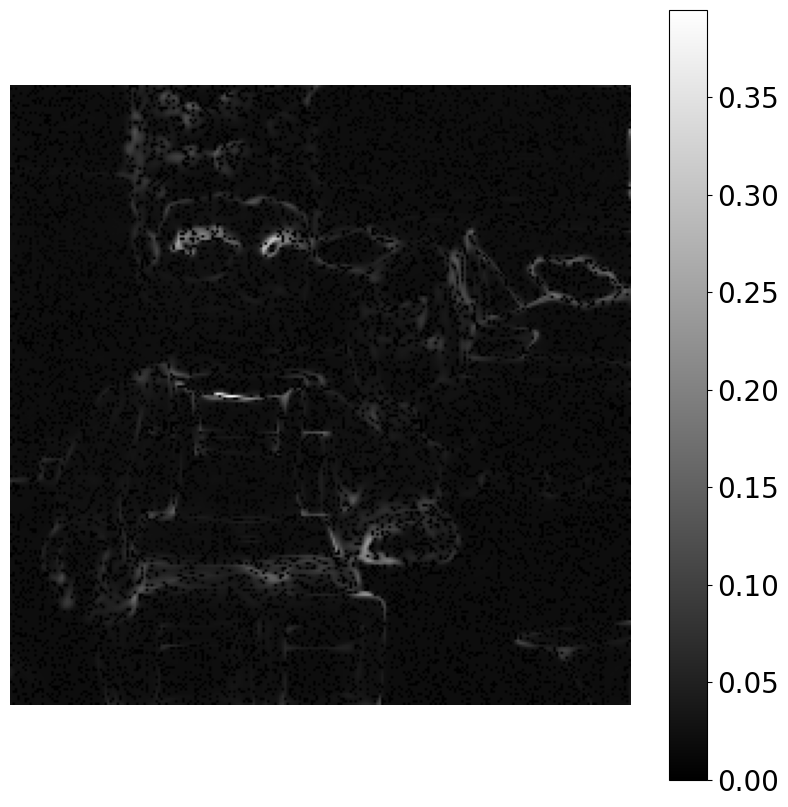} &
	
		\includegraphics[width=0.28\textwidth]{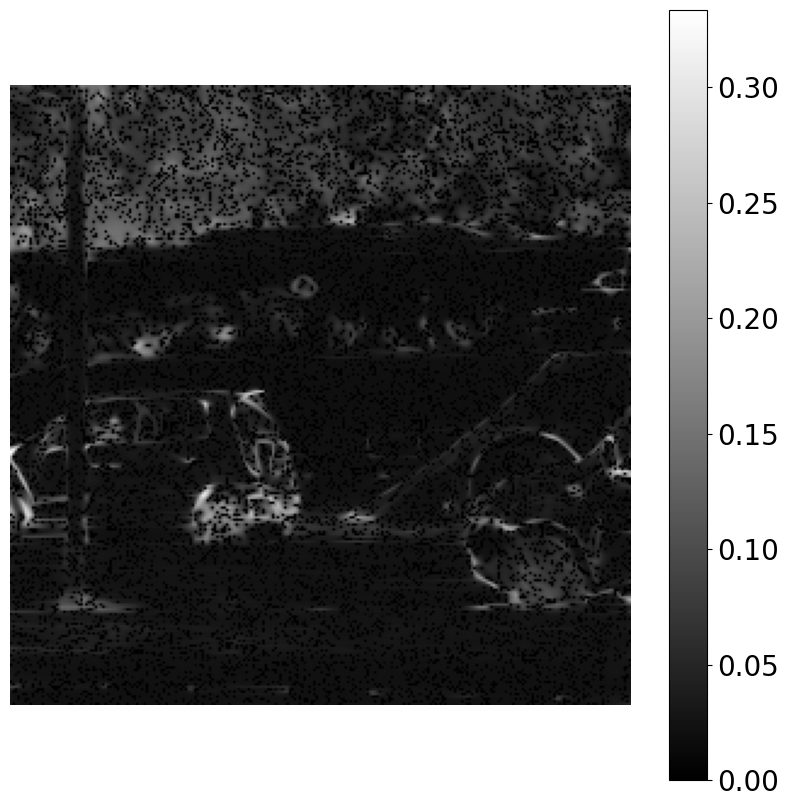} 
		
	    \\
	    \includegraphics[width=0.28\textwidth]{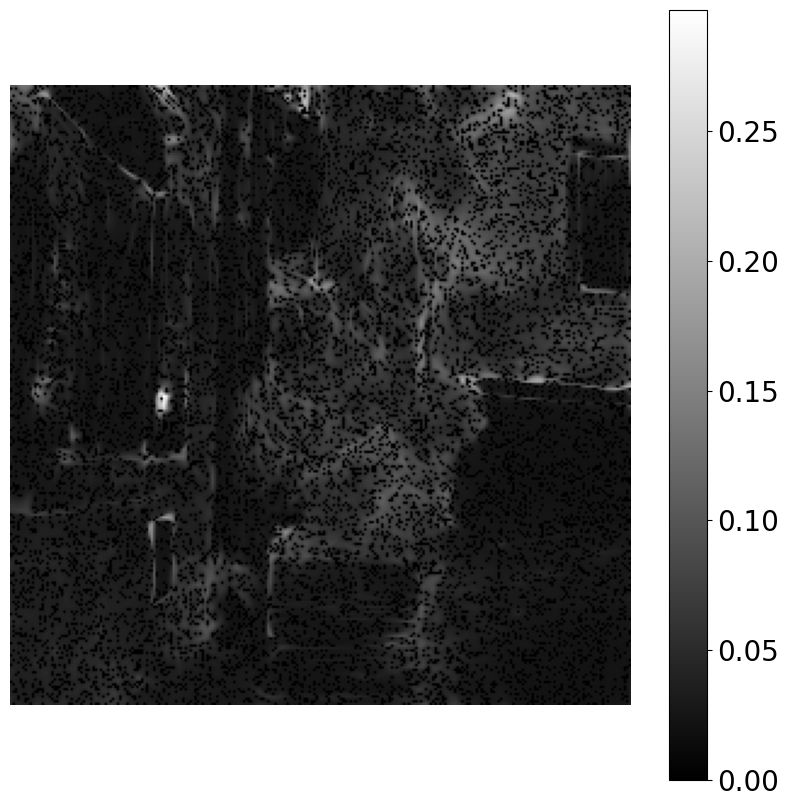} &
		
		\includegraphics[width=0.28\textwidth]{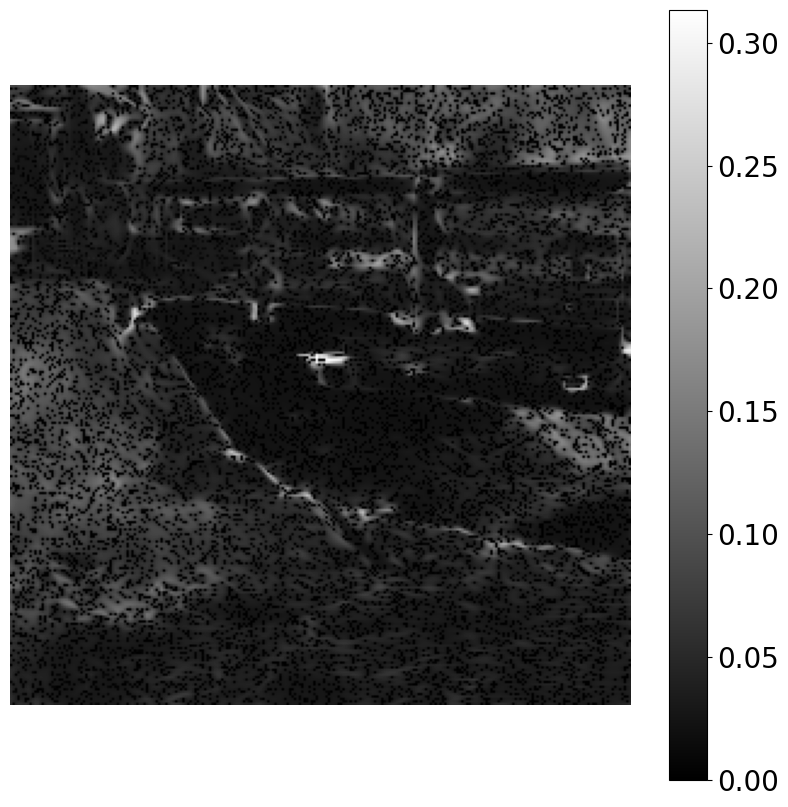} &

		\includegraphics[width=0.28\textwidth]{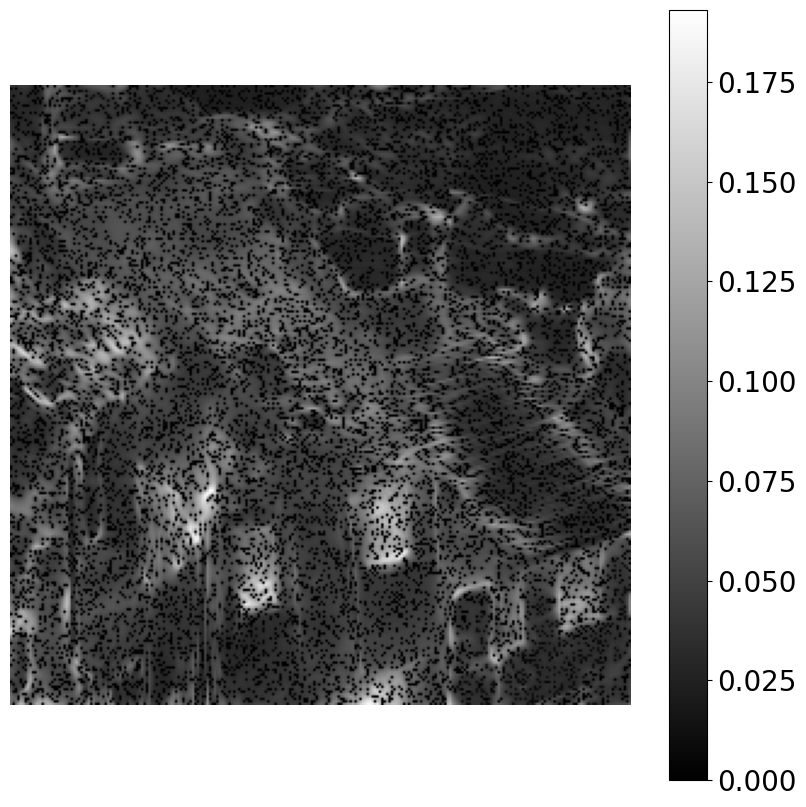}

	\end{tabular}
	\caption{Marginal posterior standard deviation of the unobserved pixels for the inpainting problem. Uncertainty is located around edges and in textured areas.}
	\label{fig:inpainting_std}
\end{figure}

\paragraph{Inpainting}
We first focus on the inpainting problem. 
\Cref{fig:inpainting_std} shows a map of the pixel-wise marginal standard deviations, for all images. We observe that pixels in homogeneous regions have low uncertainty, while pixels on textured regions, edges, or complex structures (a reflection on the window shutter in the Alley image for instance) are the most uncertain.

\begin{figure}
    \centering
    \begin{tabular}{ccc}
    
    \includegraphics[width = 0.3 \textwidth]{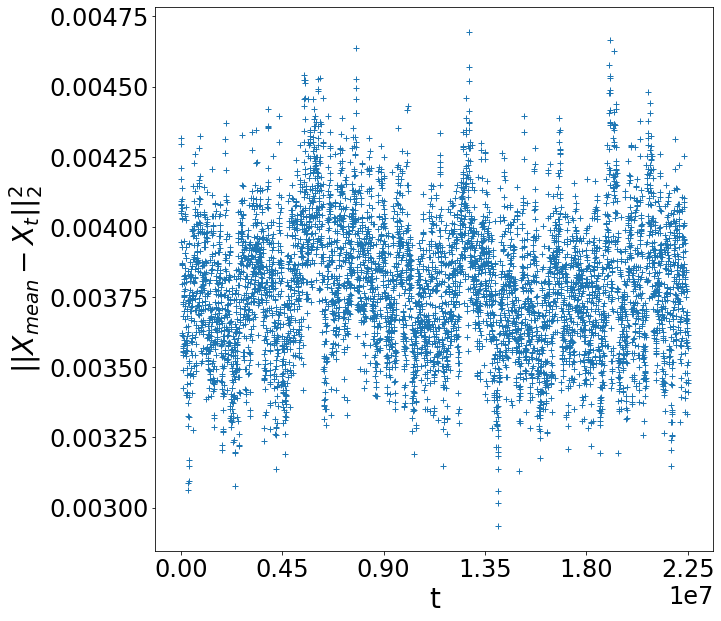} 
    &
    \includegraphics[width = 0.3 \textwidth]{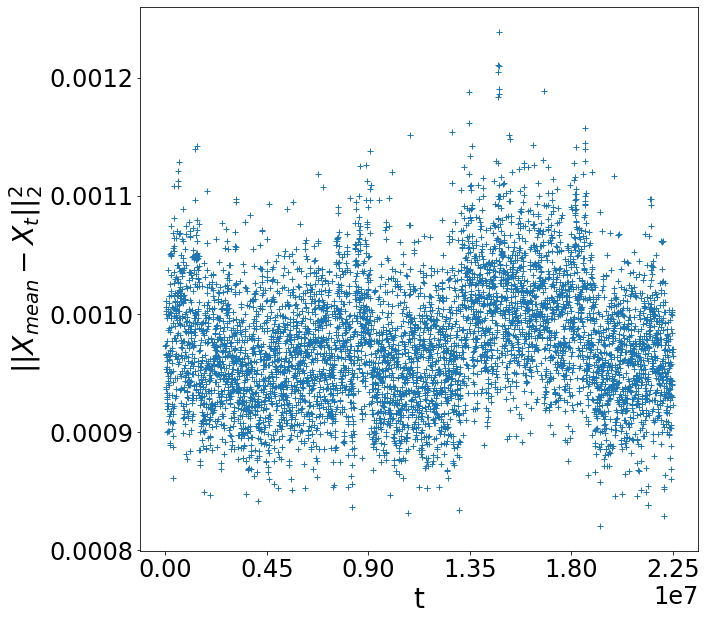} &
    \includegraphics[width = 0.3 \textwidth]{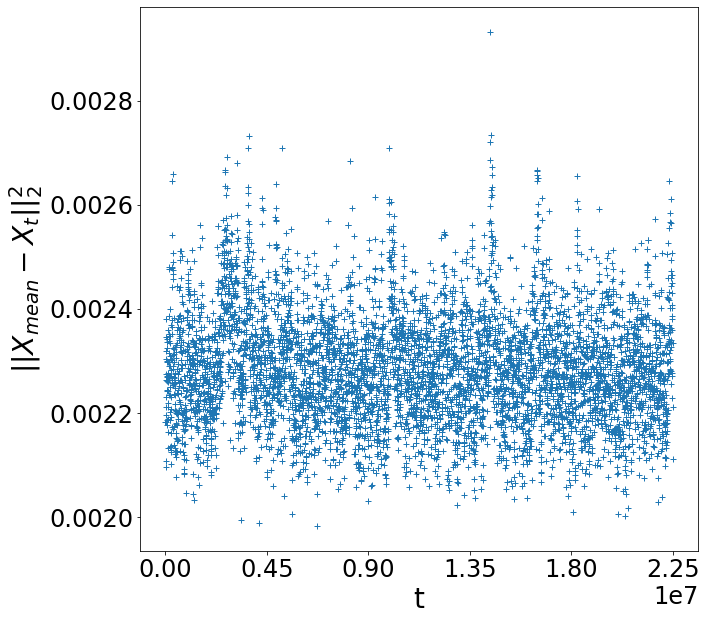}
    \\
    \texttt{Cameraman.} & \texttt{Simpson.} & \texttt{Traffic.}
    \\
     \includegraphics[width = 0.3 \textwidth]{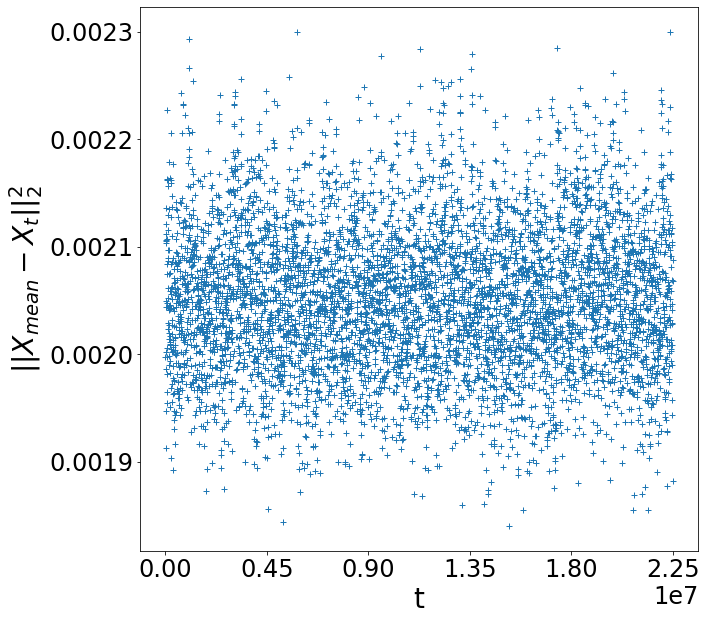} &  
    \includegraphics[width = 0.3 \textwidth]{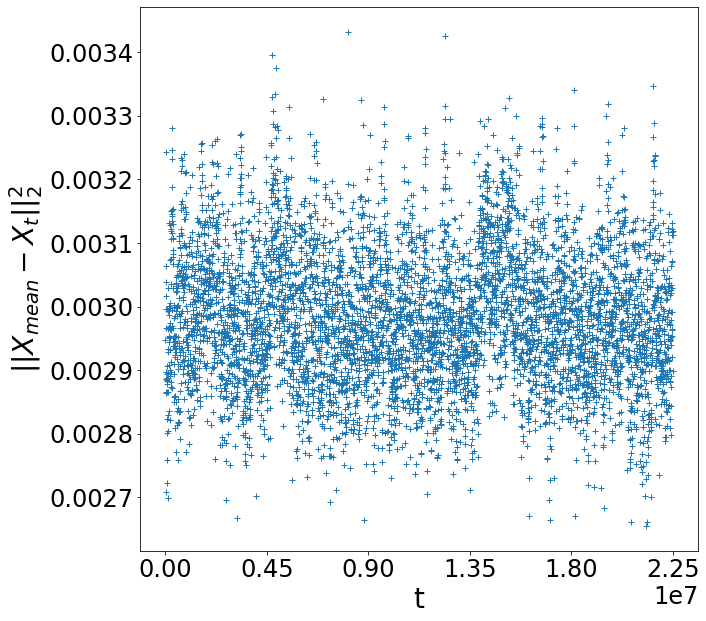} &
    \includegraphics[width = 0.3 \textwidth]{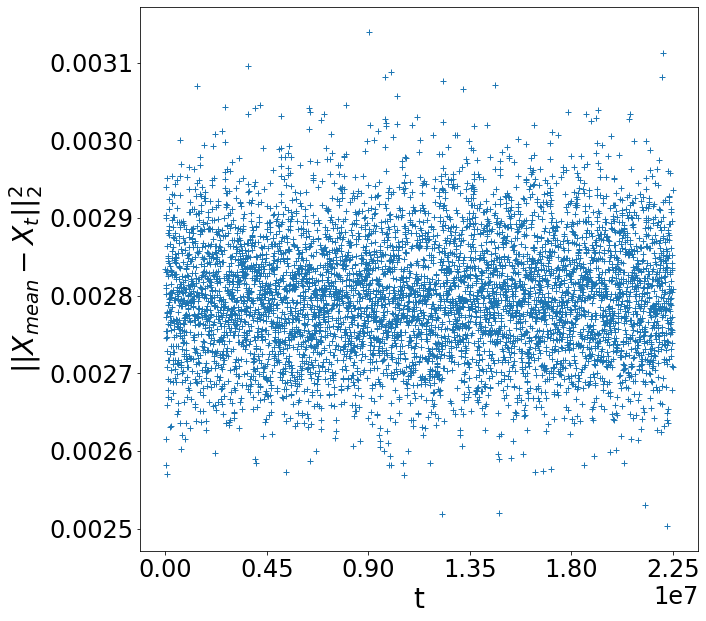}  
    \\
    \texttt{Alley.} & \texttt{Bridge.} & \texttt{Goldhill.}

    \end{tabular}

    \caption{Evolution of the $L_2$ distance between the final MMSE estimate and the samples generated by \pnpula \  for the inpainting problem after the burn-in phase. Samples randomly oscillate around the MMSE. It means that they are uncorrelated. For the images \texttt{Cameraman}, \texttt{Simpson} or \texttt{Bridge}, we note a change of range for the $L_2$ distance. It could be interpreted as a mode switching as our posterior is likely not log-concave.}
    \label{fig:inpainting_l2_distance}
\end{figure}

For the same experiments, \Cref{fig:inpainting_l2_distance} shows the Euclidean distance between the final MMSE estimate (computed using all samples) and the  samples of the chain, every 2500 samples (after the burn-in period, and hence in what is considered to be a close-to-stationary regime). Fluctuations around the posterior mean and the absence of temporal structure in the plots of \texttt{Alley} or \texttt{Goldhill} are a first indication that the chain explores the solution space with ease. However, in some other cases such as the \texttt{Simpson} image, we observe meta-stability, where the chain stays in a region of the space for millions of iterations and then jumps to a different region, again for millions of iterations. This is one of the drawbacks of operating with a posterior distribution that is not log-concave and that may exhibit  several modes.

\begin{figure}
    \centering
    \begin{tabular}{ccc}
        \includegraphics[width=0.3\textwidth]{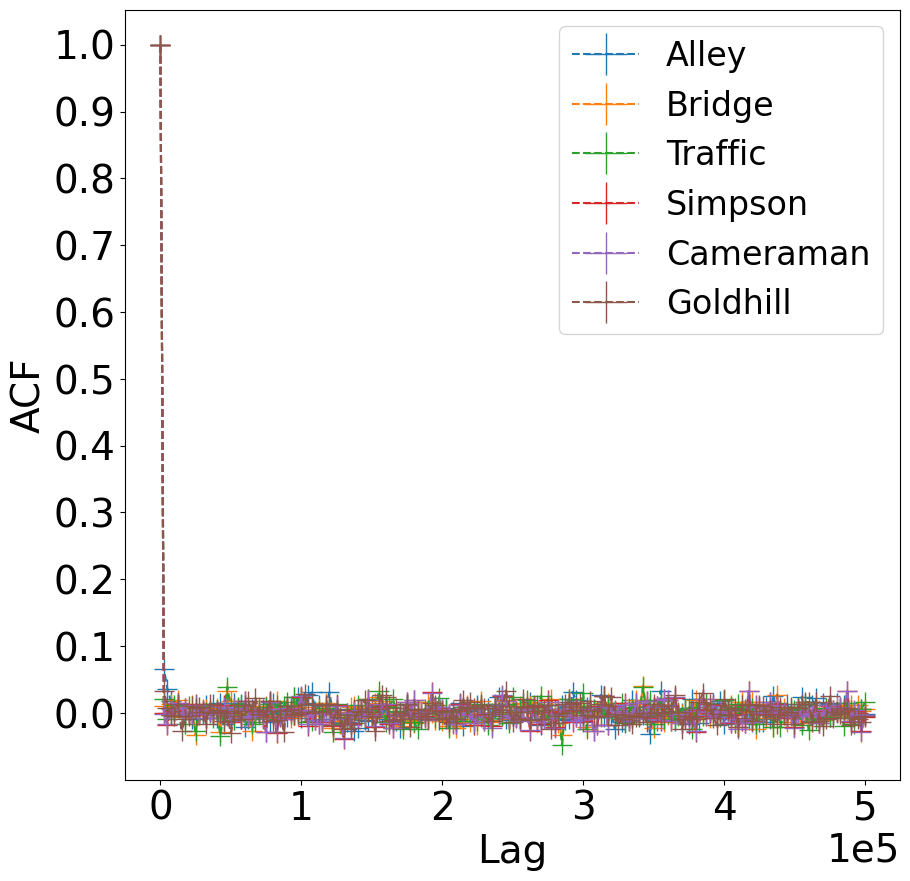}
        &
       \includegraphics[width=0.3\textwidth]{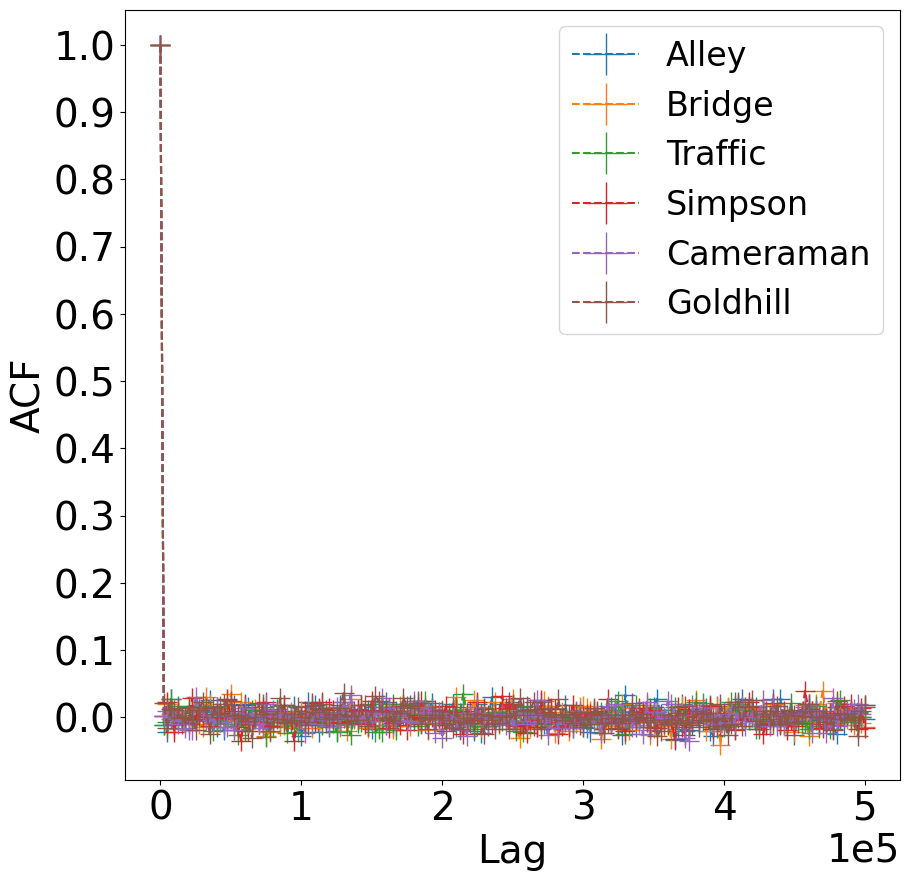}
        &
        \includegraphics[width=0.3\textwidth]{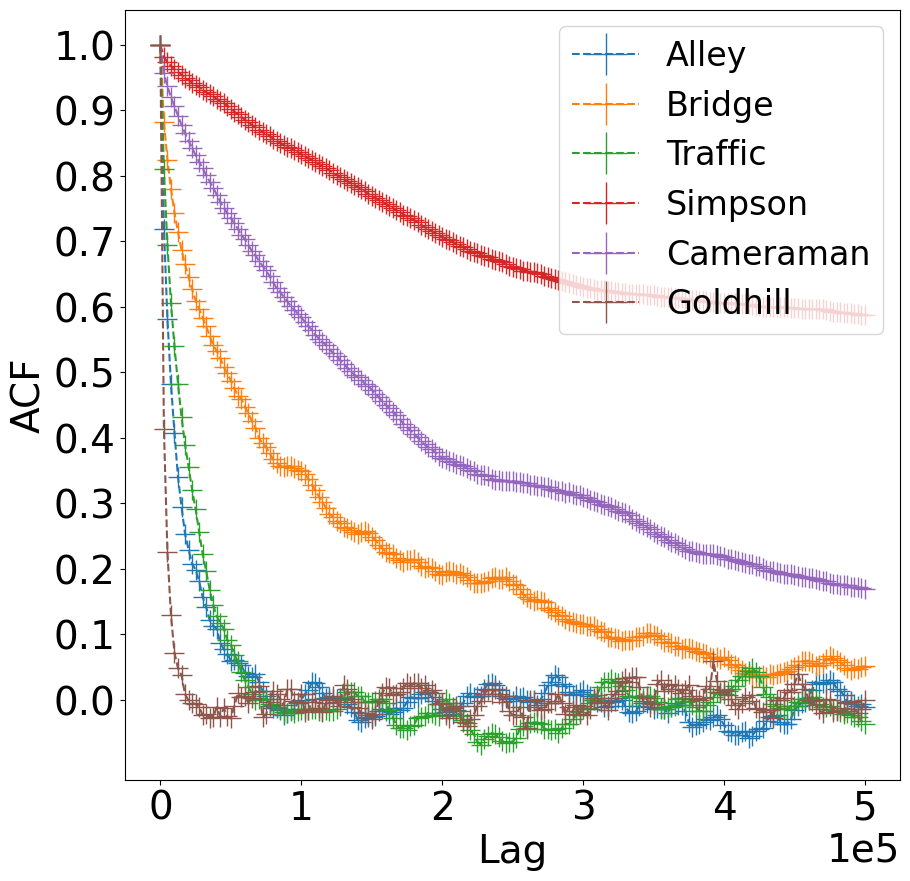}\\
        Fastest direction& Median direction& Slowest direction
    \end{tabular}
    \caption{ACF for the inpainting problem. The ACF are shown for lags up to 5e5 for all images in the pixel domain. After 5e5 iterations, sample pixels are nearly uncorrelated in all spatial directions for the images \texttt{Traffic}, \texttt{Alley}, \texttt{Bridge} and \texttt{Goldhill}. For the images \texttt{Cameraman} and \texttt{Simpson}, in the slowest direction, samples need more iterations to become uncorrelated.}
    \label{fig:inpainting_acf_pixels}
\end{figure}

Lastly, Figure~\ref{fig:inpainting_acf_pixels} displays the sample ACFs of the fastest and slowest converging statistics associated with the inpainting experiments (as estimated by identifying, for each image, the unknown pixels with lowest and highest uncertainty). These ACF plots measure how fast samples become uncorrelated.  A fast decay of the ACF is associated with good Markov chain mixing, which in turn implies accurate Monte Carlo estimates. On the contrary, a slow decay of the ACF indicates that the Markov chain is moving slowly, which leads to Monte Carlo estimates with high variance. As mentioned previously, because computing and visualising a multivariate ACF is difficult, here we show the ACF of the chain along the slowest and the fastest directions in the spatial domain (for completeness, we also show the ACF for a pixel with median uncertainty). We see that independence is reached very fast in the subspaces of low or median uncertainty, and is much slower for the few very uncertain pixels.

\begin{figure}
    \centering
   \begin{tabular}{cccc}
        \includegraphics[width=0.25\textwidth]{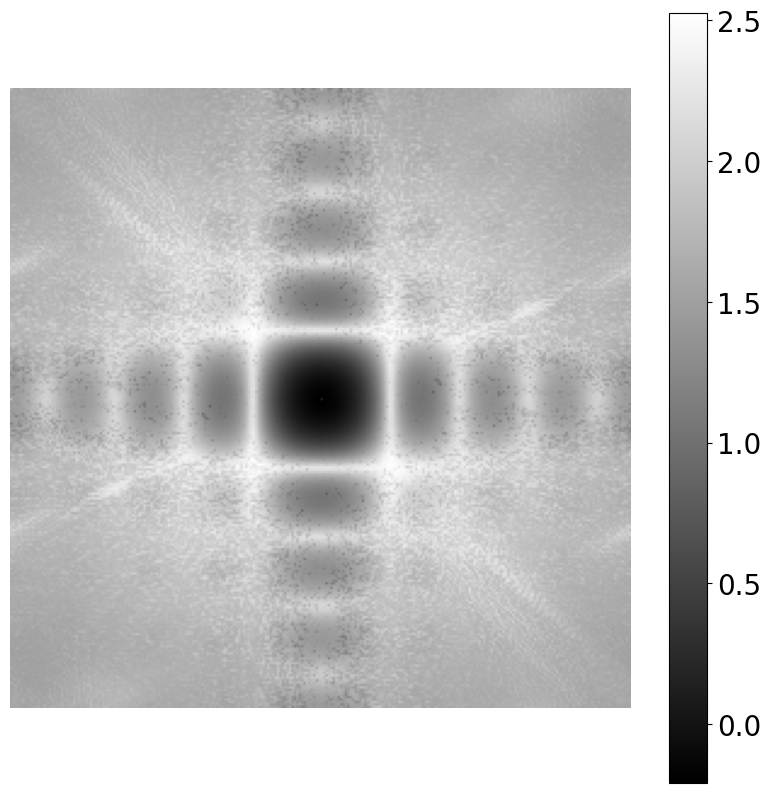} &
        \includegraphics[width=0.25\textwidth]{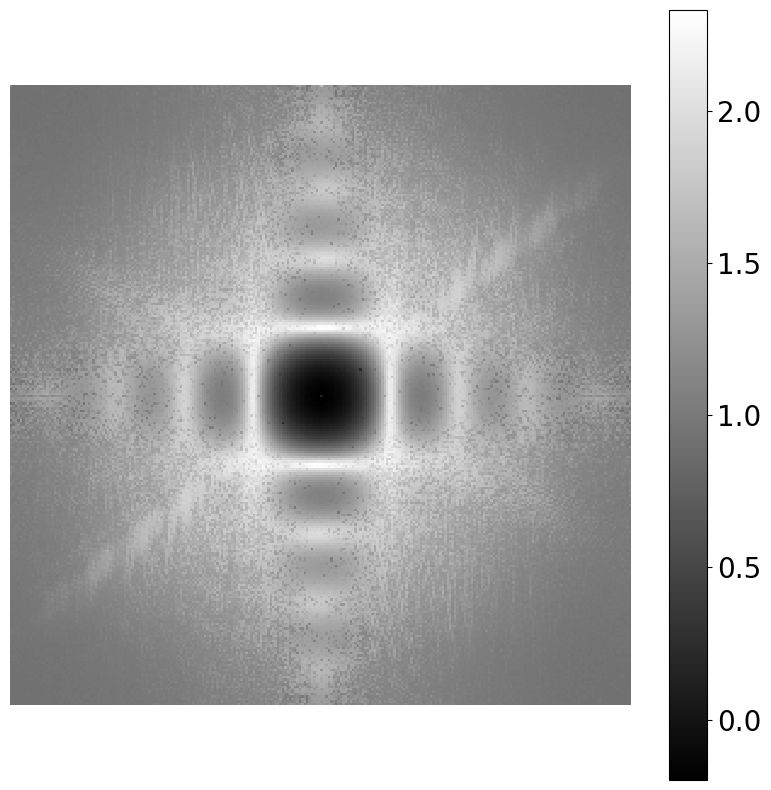} &
        \includegraphics[width=0.25\textwidth]{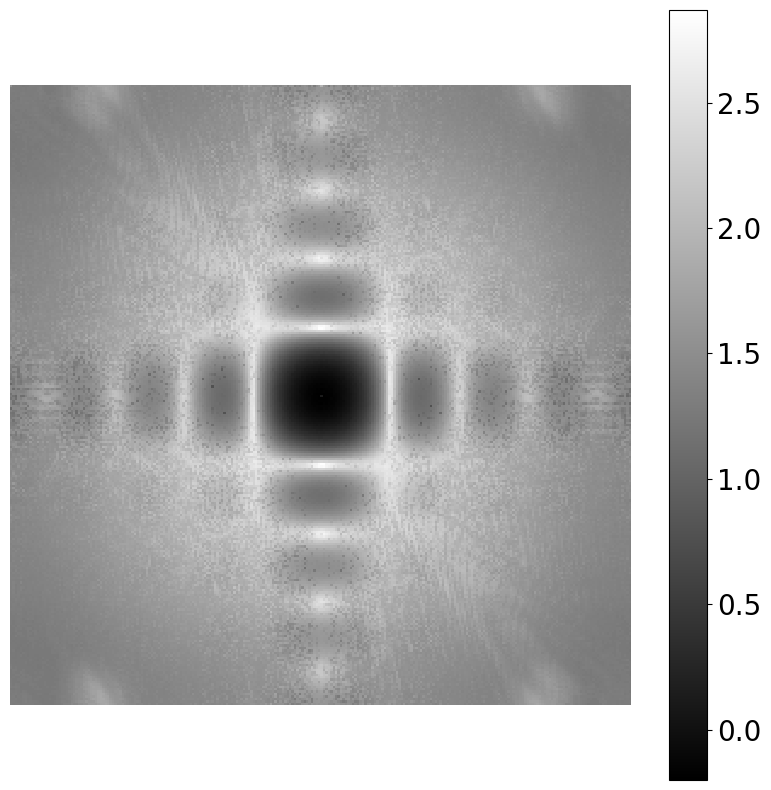}  &
        \multirow{2}{*}{\makecell{\includegraphics[width=0.13\textwidth]{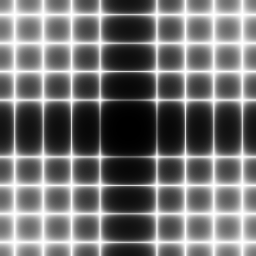} \\
        \scriptsize{Inverse Fourier} \\ \scriptsize{transform of}\\ \scriptsize{the blur kernel.}\\}}
        \\
        \includegraphics[width=0.25\textwidth]{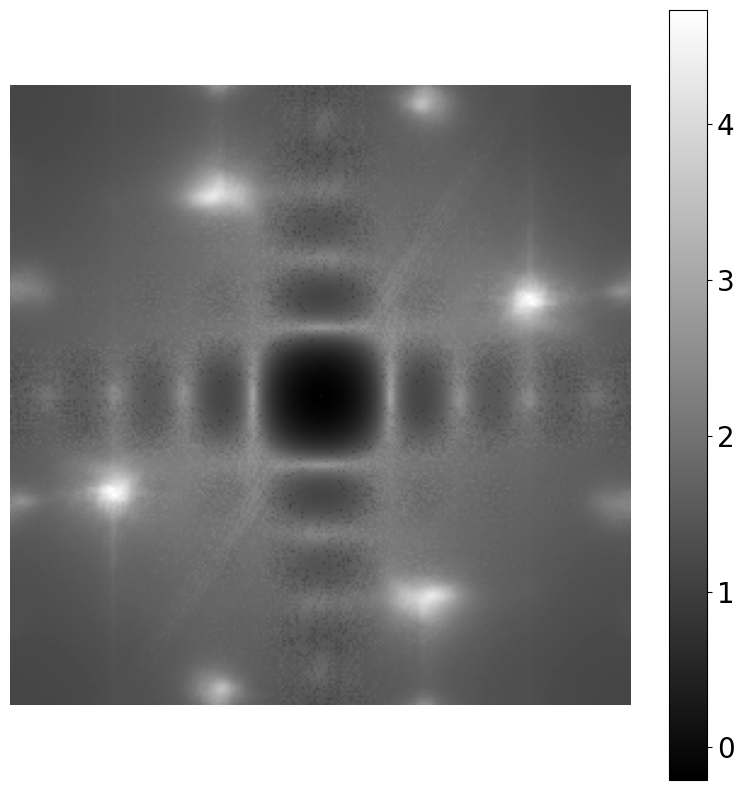} &
        \includegraphics[width=0.25\textwidth]{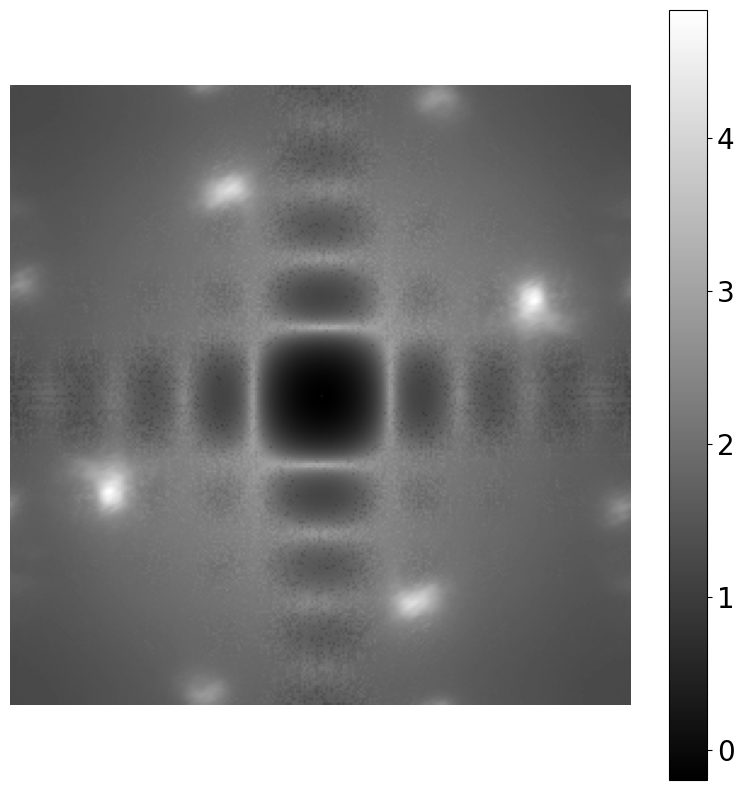} &
        \includegraphics[width=0.25\textwidth]{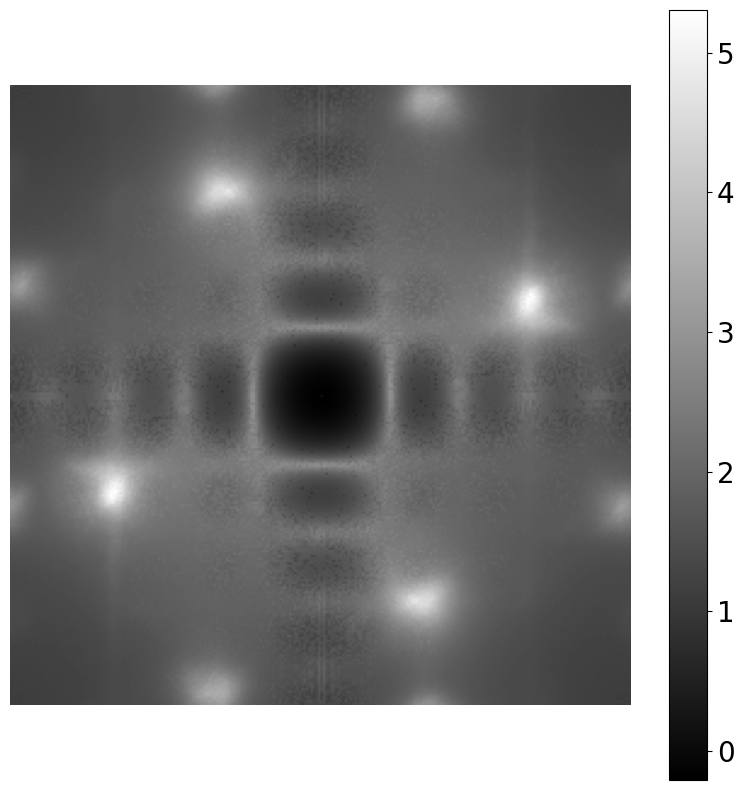} 

   \end{tabular}
     \caption{Log-standard deviation maps in the Fourier domain for the Markov chains defined by \pnpula \ for the deblurring problem. First line: images \texttt{Cameraman}, \texttt{Simpson}, \texttt{Traffic}. Second line: images \texttt{Alley}, \texttt{Bridge} and \texttt{Goldhill}. For the first three images, we clearly see that uncertainty is observed on frequencies that are near the kernel of the blur filter (shown on the right), and is also higher around high frequencies ({\em i.e}. around edges and textured areas in images). For the last three images, very high uncertainty is observed around some specific frequencies. In the direction of these frequencies, the Markov chain is moving very slowly and the mixing time of the chain is particularly slow, as shown on~\Cref{fig:deblurring_acf}. 
     }
    \label{fig:deblurring_std_map}
\end{figure}

\paragraph{Deblurring}
We now focus on the non-blind image deblurring experiments, where, as explained previously, we perform our convergence analysis by using statistics associated with the Fourier domain. \Cref{fig:deblurring_std_map} depicts the marginal standard deviation of the Fourier coefficients (in absolute value), for all images. For the three  images \texttt{Cameraman}, \texttt{Simpsons} and \texttt{Traffic}, all the standard deviations have a similar range of values, and the largest values are observed around frequencies in the kernel of the blur filter (shown on the right of the same figure) and for high frequencies. 
Conversely, for the three images \texttt{Alley}, \texttt{Bridge} and \texttt{Goldhill}, very high uncertainty is observed in the vicinity of four specific frequencies. This suggests that the denoiser used is struggling to regularise these specific frequencies, and consequently the posterior distribution is very spread along these directions and difficult to explore by Markov chain sampling as a result. Interestingly, this phenomenon is only observed in the images that are rich in texture content.%

\begin{figure}
    \centering
    \begin{tabular}{ccc}
    \includegraphics[width = 0.3 \textwidth]{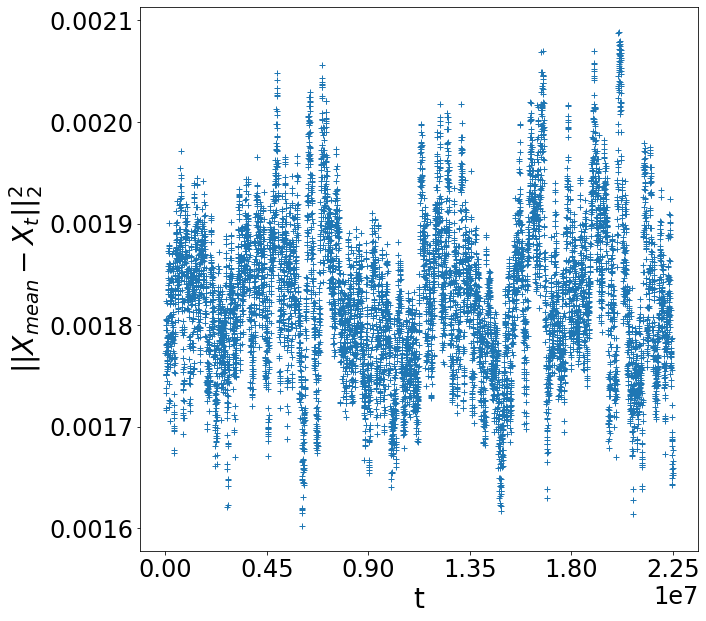} &
    \includegraphics[width = 0.3 \textwidth]{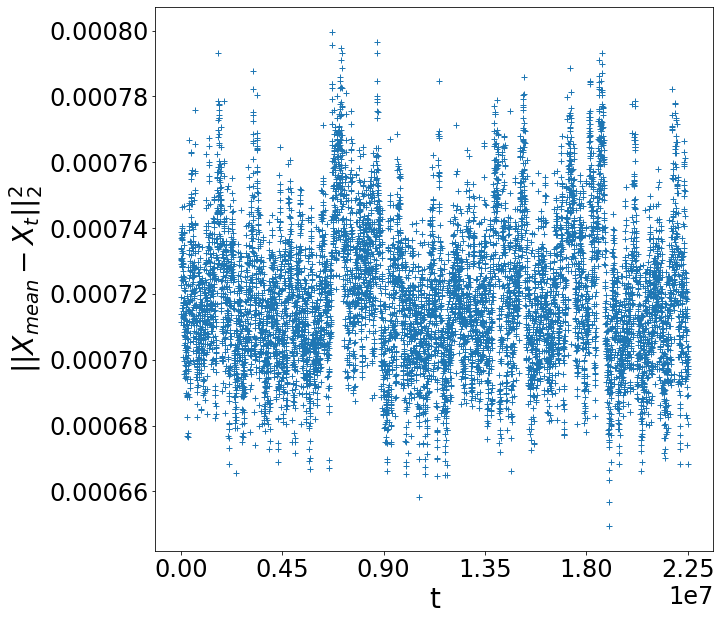} &
    \includegraphics[width = 0.3 \textwidth]{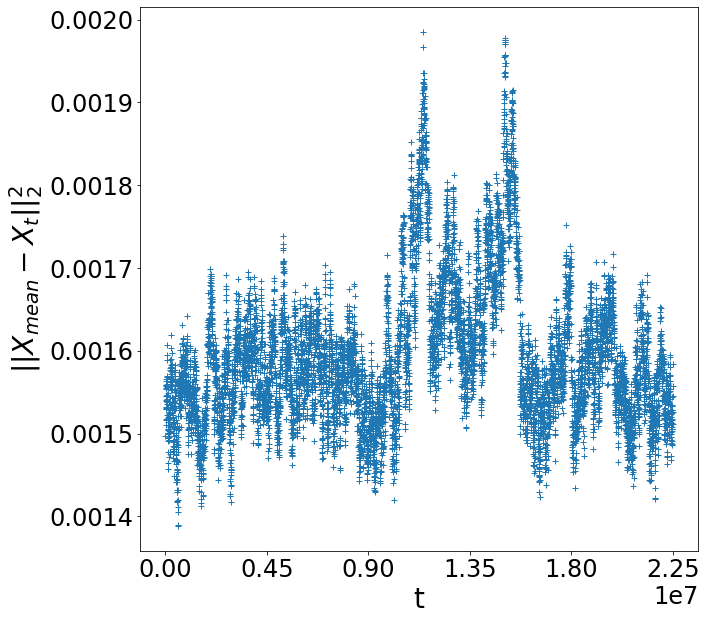}
    \\
    \texttt{Cameraman.} & \texttt{Simpson.} & \texttt{Traffic.}
    \\
    \includegraphics[width = 0.3 \textwidth]{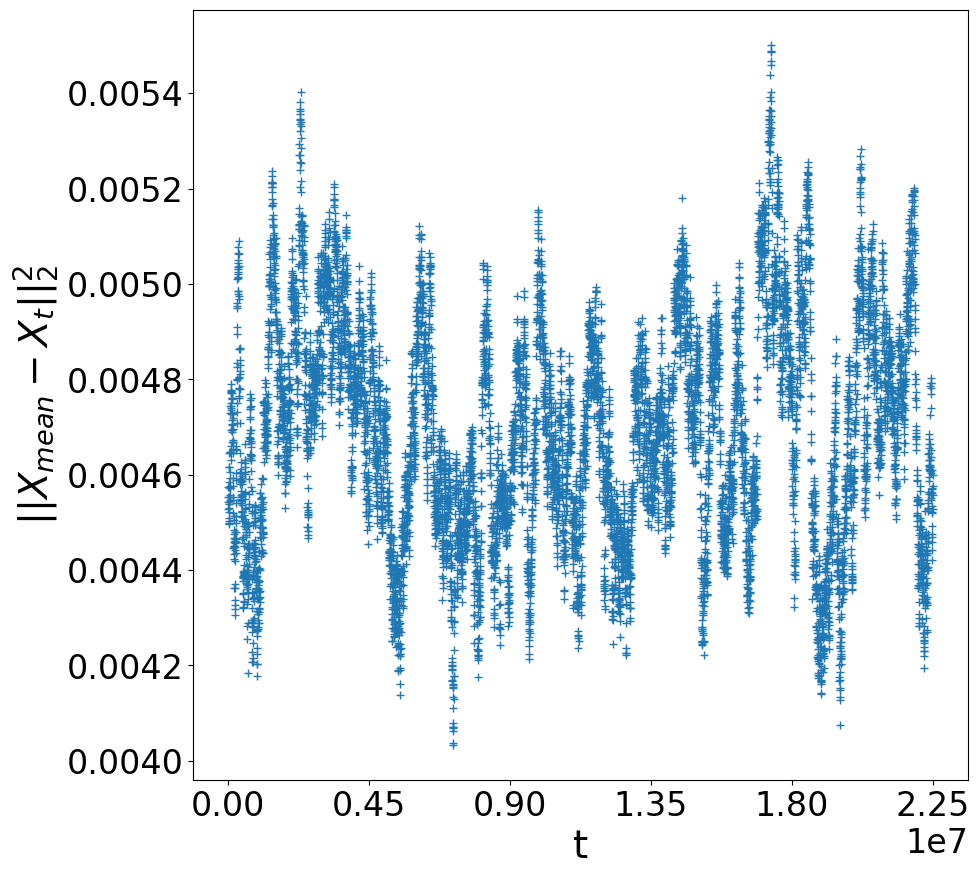} &  
    \includegraphics[width = 0.3 \textwidth]{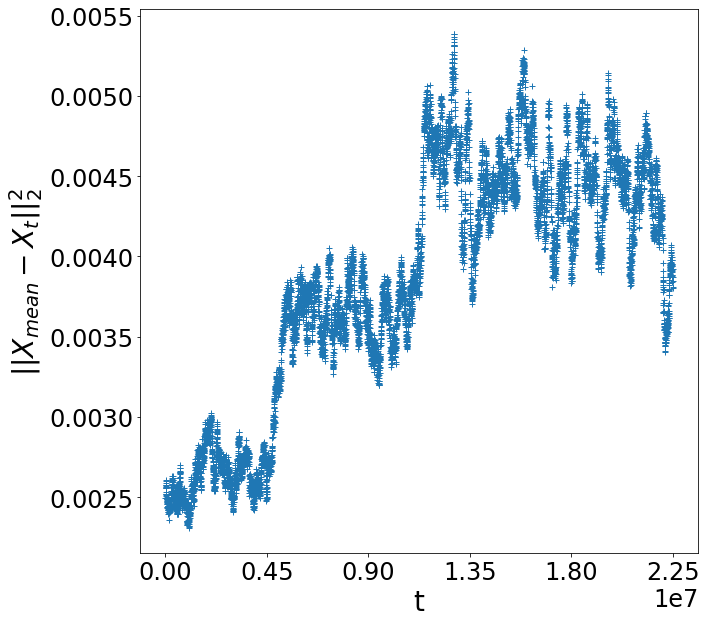} &
    \includegraphics[width = 0.3 \textwidth]{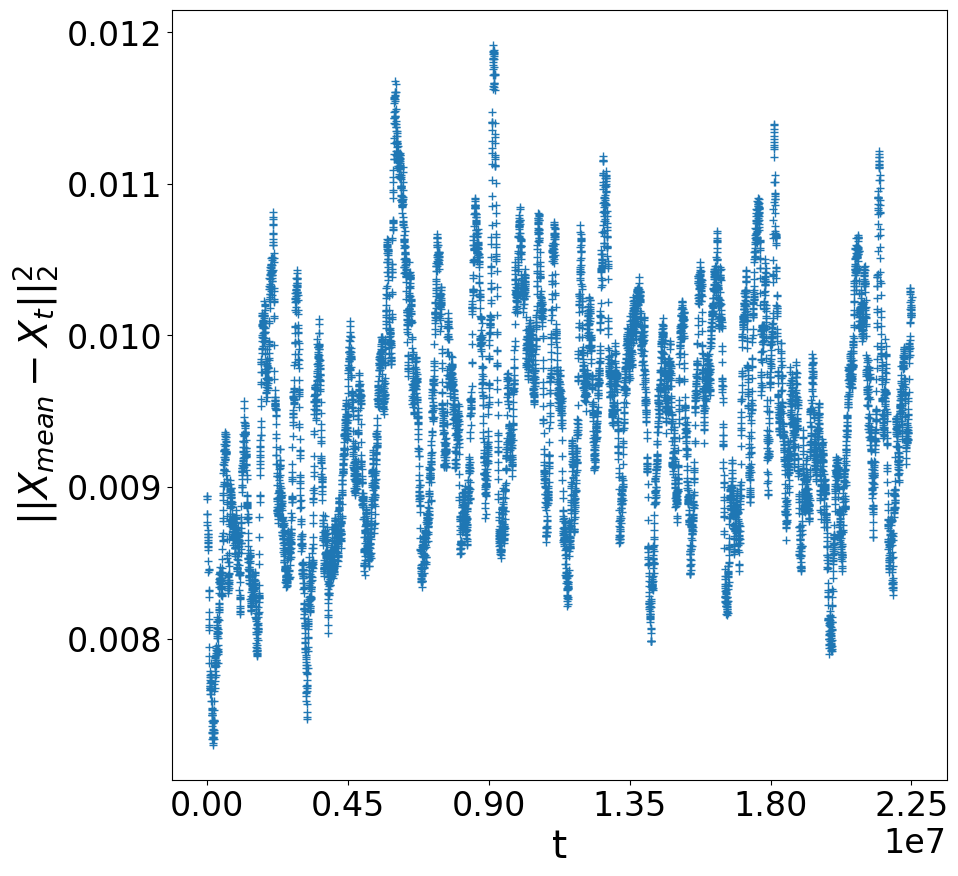} 
    \\
     \texttt{Alley.} & \texttt{Bridge.} & \texttt{Goldhill.}
    
    \end{tabular}

    \caption{Evolution of the $L_2$ distance between the final MMSE estimate and the samples generated by \pnpula \ for the deblurring problem after the burn-in phase. For images as \texttt{Cameraman} or \texttt{Simpson}, samples randomly oscillate around the MMSE. On the contrary, for images as \texttt{Bridge} or \texttt{Goldhill}, the plot is structured, meaning that samples are still correlated.}
    \label{fig:deblurring_l2_distance}
\end{figure}

Moreover, \Cref{fig:deblurring_l2_distance} depicts the Euclidean distance between the MMSE estimator computed from entire chain ({\em i.e.} all samples) and each sample (we show one point every 2500 samples). We notice that many of the images exhibit some degree of meta-stability or slow convergence because of the presence of directions in the solution space with very high uncertainty. Again, this is consistent with our convergence theory, which identifies posterior multimodality and anisotropy as key challenges that future work should seek to overcome.

\begin{figure}
    \centering
    \begin{tabular}{cccc}
        \includegraphics[width=0.22\textwidth]{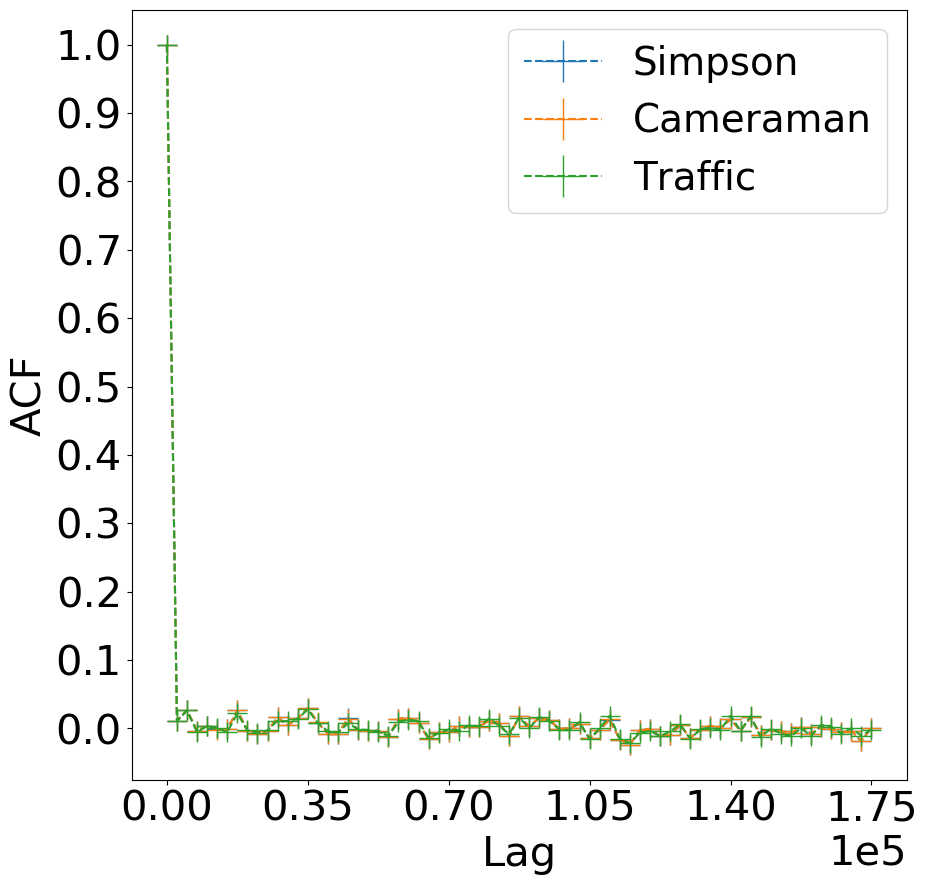} &
        \includegraphics[width=0.22\textwidth]{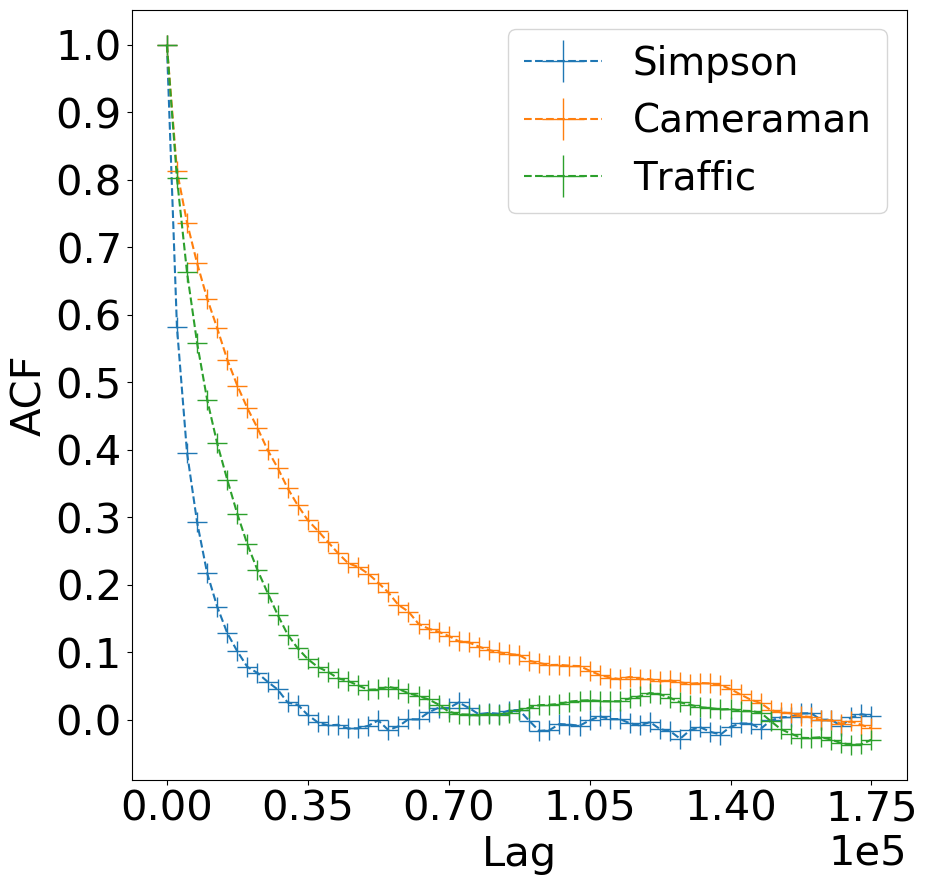} &
        \includegraphics[width=0.22\textwidth]{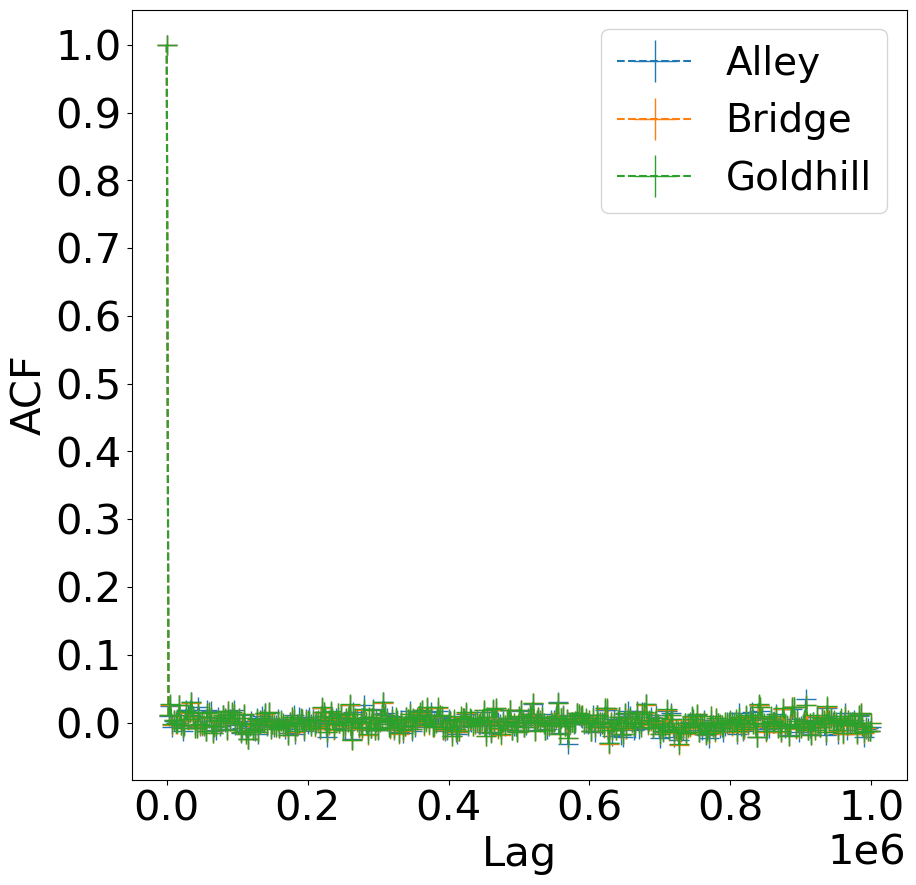} &
        \includegraphics[width=0.22\textwidth]{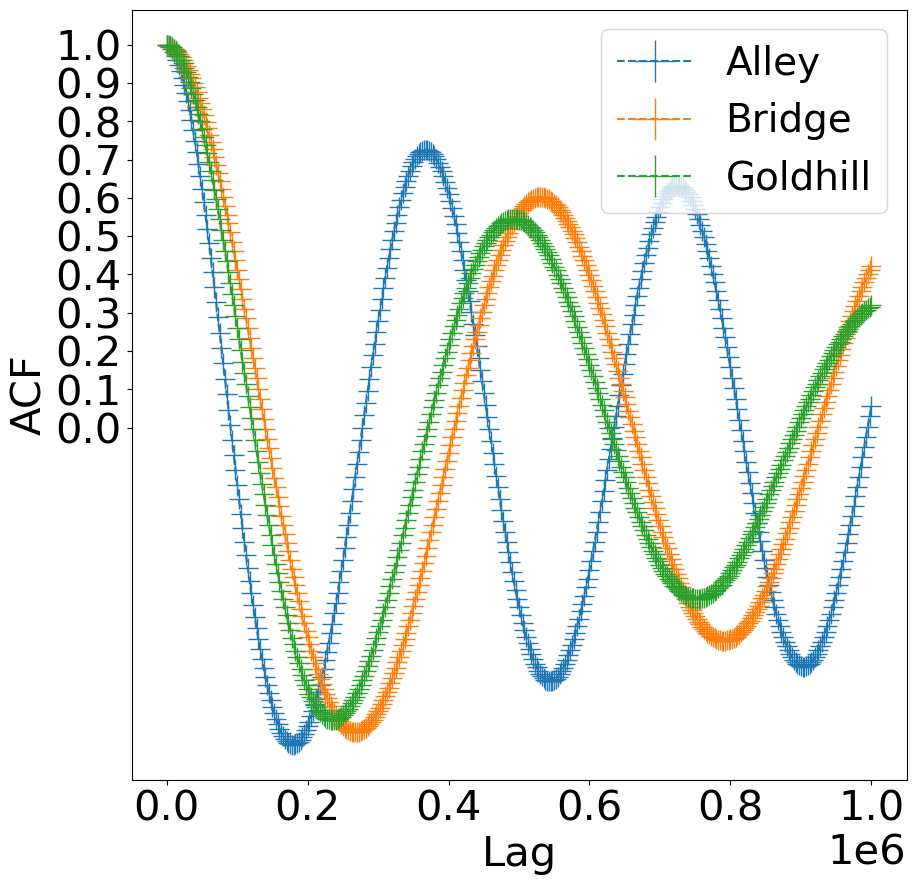}\\
        Fast direction& Slow direction& Fast direction& Slow direction
    \end{tabular}
    \caption{ACF for the deblurring problem. The ACF are shown for lags up to 1.75e5 for the three images \texttt{Cameraman}, \texttt{Simpson} and \texttt{Traffic} (see the two plots to the left) and independence seems to be achieved in all directions. For the three other images, independence is not achieved in the slowest direction (corresponding to the most uncertain frequency of the samples in the Fourier domain) even after 1e6 iterations.}
    \label{fig:deblurring_acf}
\end{figure}

Lastly, we show on Figure~\ref{fig:deblurring_acf} the sample ACFs for the slowest and the fastest directions in the Fourier domain\footnote{The slowest direction corresponds to the Fourier coefficient with the highest (real or imaginary) variance.}. Again, in all experiments, independence is achieved quickly in the fastest direction. The behaviours of the slowest direction for the three images \texttt{Alley}, \texttt{Bridge} and \texttt{Goldhill} suggest that the Markov chain is close to the stability limit and exhibits highly oscillatory behaviour as well as poor mixing.

\subsection{Point estimation for non-blind image deblurring and inpainting}
\label{sec:deblurring_inpainting}

\begin{figure}
    \centering
    \begin{tabular}{cccc}
        \CenteredVcell{0.3\textwidth}{\pnpula, $\delta =\delta_{th}$.}
        &
        \includegraphics[width = 0.28\textwidth]{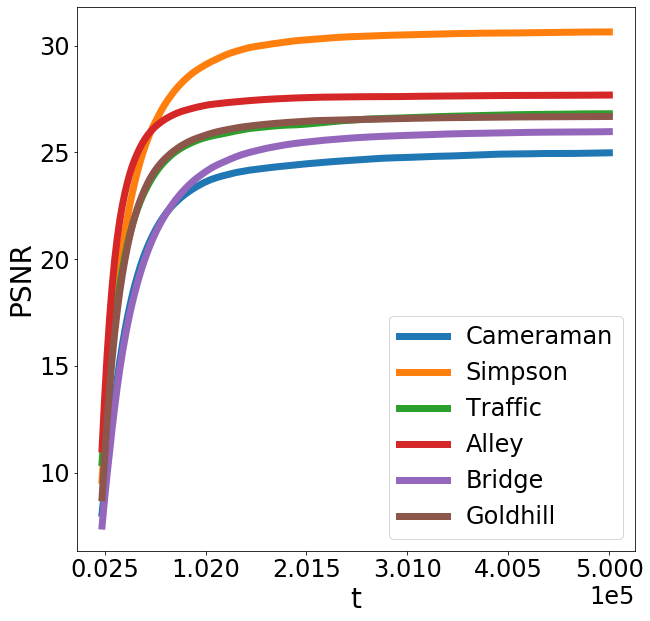}
        &
        \includegraphics[width = 0.28 \textwidth]{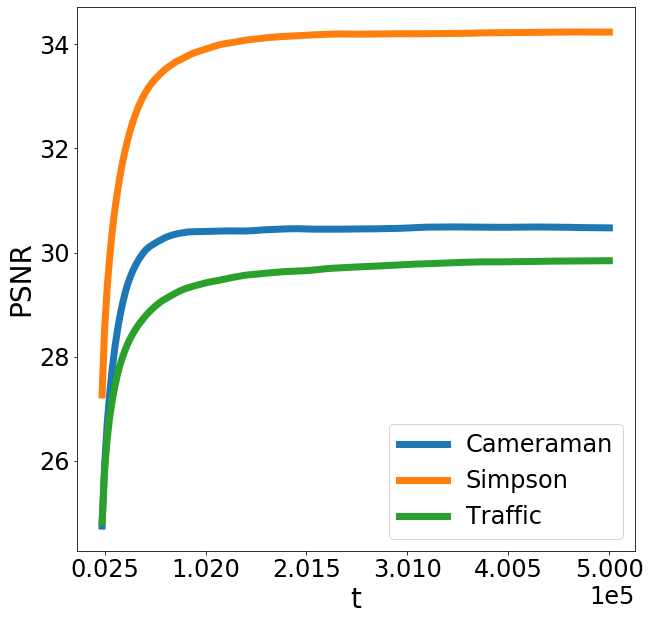} 
        &
        \includegraphics[width = 0.28 \textwidth]{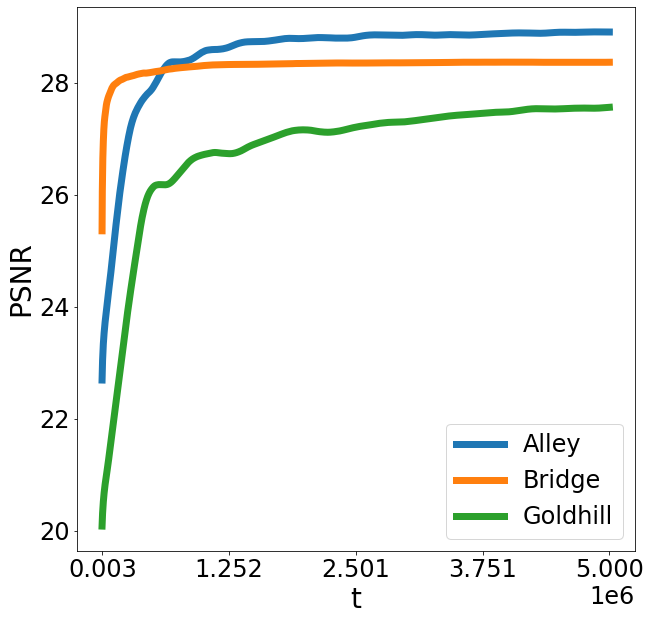} 
        \\
        \CenteredVcell{0.3\textwidth}{P\pnpula, $\delta =6\delta_{th}$.}
        &
        \includegraphics[width = 0.28\textwidth]{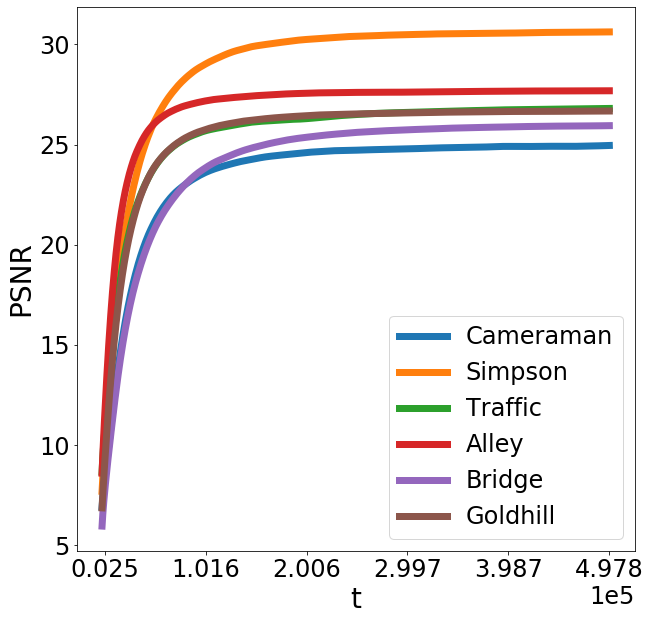}
        &
        \includegraphics[width = 0.28 \textwidth]{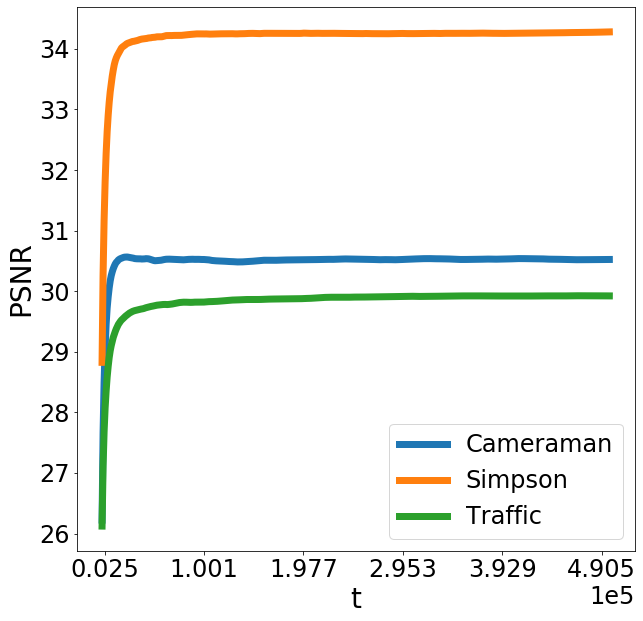} 
        &
        \includegraphics[width = 0.28 \textwidth]{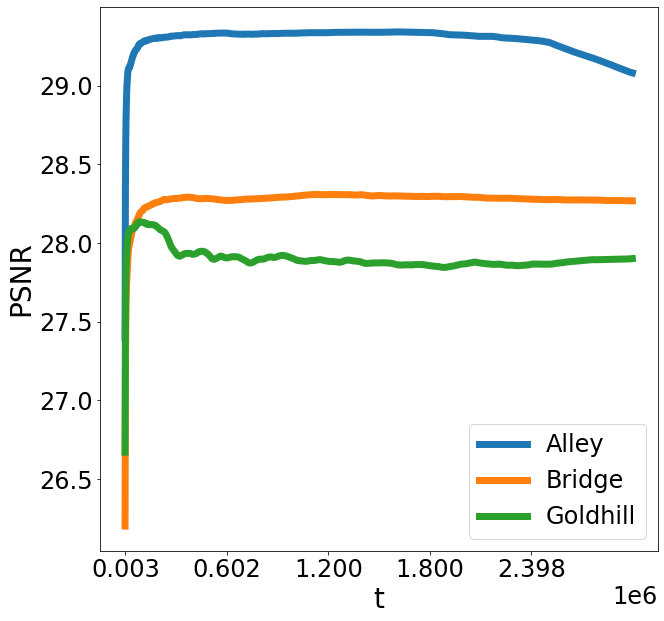} 

    \end{tabular}

    \caption{Left: PSNR evolution of the estimated MMSE for the inpainting problem. After $5e5$ iterations, the convergence of the first order moment of the posterior distribution seems to be achieved for all images. Middle and right: PSNR evolution of the estimated MMSE for the deblurring problem. The convergence for the posterior mean can be fast  for simple images such as \texttt{Cameraman}, \texttt{Simpson}, and \texttt{Traffic} (for these images the PSNR evolution is shown for the first 5e5 iterations). Increasing the $\delta$ increases the convergence speed for these images by a factor close to 2. For more complex images, such as \texttt{Alley} %
    or \texttt{Goldhill}, the convergence is much slower and is still not achieved after 3e6 iterations with P\pnpula\  for $\delta = 6\delta_{th}$.%
    }
    \label{fig:deblurring_evolution_psnr}
\end{figure}

We are now ready to study the quality of the MMSE estimators delivered by \pnpula\ and P\pnpula \ and report comparisons with MAP estimation by \pnpsgd\ \cite{laumont2021maximum}.

\paragraph{Quantitative results}
Figure~\ref{fig:deblurring_evolution_psnr} illustrates the evolution of the PSNR of the mean of the Markov chain (the Monte Carlo estimate of the MMSE solution), as a function of the number of iterations, for the six images of ~\Cref{fig:original_images}.  These plots have been computed by using a step-size $\delta = \delta_{th}$ that is just below the stability limit and a 1-in-2500 thinning. We observe that the PSNR between the MMSE solution as computed by the Markov chain and the truth stabilises in approximately $10^5$ iterations in the experiments where the chain exhibits fast convergence, whereas over $10^6$ are required in experiments that suffer from slow convergence (e.g., deblurring of \texttt{Alley}, \texttt{Bridge} and \texttt{Goldhill}). Moreover, we observe that using P\pnpula\ with a larger step-size can noticeably reduce the number of iterations required to obtain a stable estimate of the posterior mean, particularly in the image deblurring experiments.

\begin{figure}[h!]
	\centering
	\begin{tabular}{cccc}
	    \CenteredVcell{0.2\textwidth}{\pnpula} &
		\includegraphics[width=0.25\textwidth]{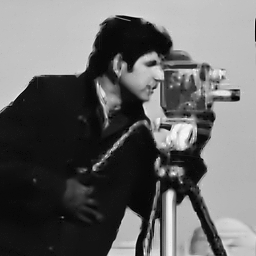} &
		
		\includegraphics[width=0.25\textwidth]{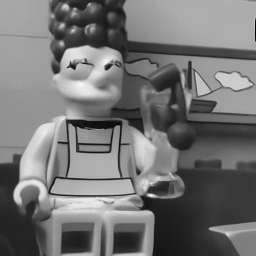} &
		
		\includegraphics[width=0.25\textwidth]{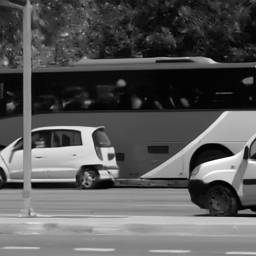} \\
		& \begin{small}PSNR=25.06/SSIM=0.89\end{small} & \begin{small}PSNR=30.62/SSIM=0.93\end{small} & \begin{small}PSNR=26.90/SSIM=0.85\end{small}
		\\
		\CenteredVcell{0.2\textwidth}{\pnpsgd} &
		\includegraphics[width=0.25\textwidth]{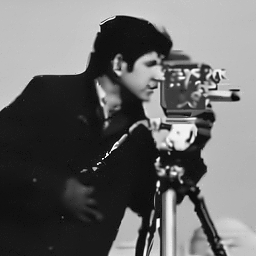} &
		\includegraphics[width=0.25\textwidth]{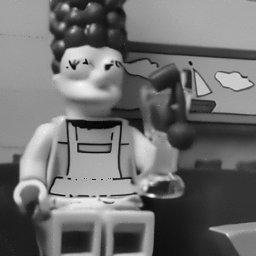} &
	
		\includegraphics[width=0.25\textwidth]{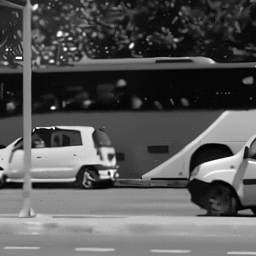}
		\\
		& \begin{small}PSNR=23.94/SSIM=0.88\end{small} & \begin{small}PSNR=28.90/SSIM=0.90\end{small} & \begin{small}PSNR=24.20/SSIM=0.81\end{small}
	\end{tabular}
	\caption{Results comparison for the inpainting task of the images presented in \Cref{fig:toipaint_images} using \pnpula \ (first row) and \pnpsgd \ initialized with a TVL2 restoration (second row).
	}
	\label{fig:inpainting_mmse1}
\end{figure}

\begin{figure}
    \centering
    \begin{tabular}{cccc}
    \CenteredVcell{0.2\textwidth}{\pnpula} &
    \includegraphics[width=0.25\textwidth]{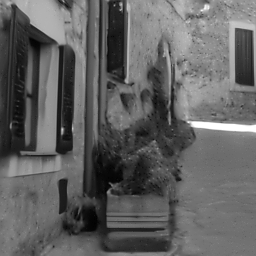} &
    
    \includegraphics[width=0.25\textwidth]{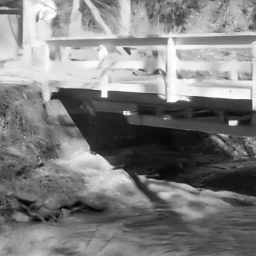} &
    
    \includegraphics[width=0.25\textwidth]{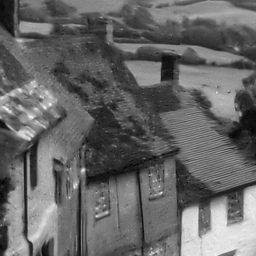} \\
    & \begin{small}PSNR=27.74/SSIM=0.79\end{small} & \begin{small}PSNR=26.16/SSIM=0.80\end{small} & \begin{small}PSNR=26.76/SSIM=0.74\end{small} \\
    \CenteredVcell{0.2\textwidth}{\pnpsgd} &
    \includegraphics[width=0.25\textwidth]{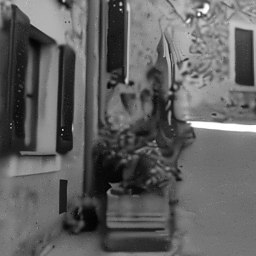} &
    
    \includegraphics[width=0.25\textwidth]{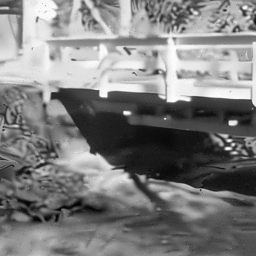} &
    
   \includegraphics[width=0.25\textwidth]{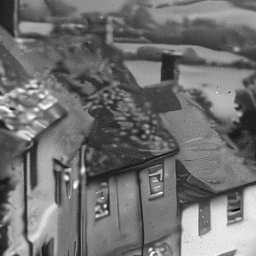} \\
    
    & \begin{small}PSNR=26.45/SSIM=0.75\end{small} & \begin{small}PSNR=24.71/SSIM=0.77\end{small} & \begin{small}PSNR=25.96/SSIM=0.72\end{small} 
    
   \end{tabular}
    \caption{Results comparison for the inpainting task of the images presented in \Cref{fig:toipaint_images} using \pnpula \ (first row) and \pnpsgd \ initialized with a TVL2 restoration (second row).}
    \label{fig:inpainting_mmse2}
\end{figure}

\begin{figure}[h!]
	\centering
	\begin{tabular}{cccc}
	    \CenteredVcell{0.2\textwidth}{\pnpula, $\alpha =1$.} &
	    \includegraphics[width=0.25\textwidth]{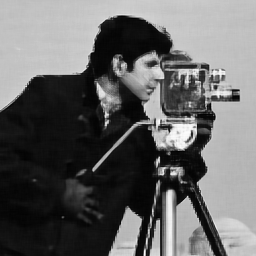} &
		\includegraphics[width=0.25\textwidth]{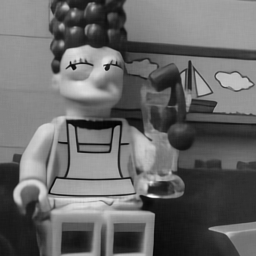} &
		\includegraphics[width=0.25\textwidth]{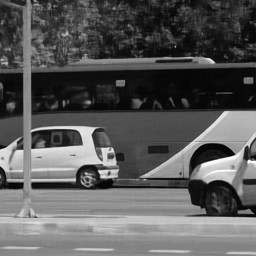} 
        \\
		& \begin{small}PSNR=30.50/SSIM=0.93\end{small} & \begin{small}PSNR=34.26/SSIM=0.94\end{small} & \begin{small}PSNR=29.90/SSIM=0.90\end{small}
		\\
		\CenteredVcell{0.22\textwidth}{\pnpsgd, $\alpha =0.3$.} &
		\includegraphics[width=0.25\textwidth]{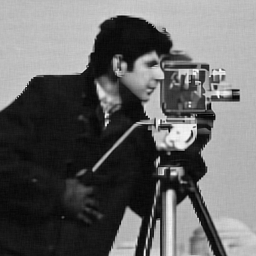} &
		\includegraphics[width=0.25\textwidth]{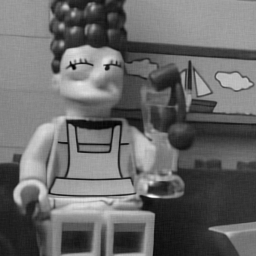} &
		\includegraphics[width=0.25\textwidth]{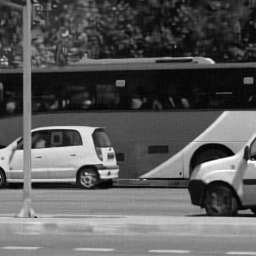} 
        \\
		& \begin{small}PSNR=30.73/SSIM=0.92\end{small} & \begin{small}PSNR=33.52/SSIM=0.92\end{small} & \begin{small}PSNR=29.42/SSIM=0.88\end{small} 
		\\
		\CenteredVcell{0.2\textwidth}{\pnpsgd, $\alpha =1$.} &
		\includegraphics[width=0.25\textwidth]{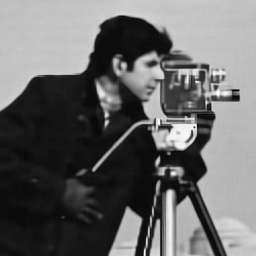}  &
		\includegraphics[width=0.25\textwidth]{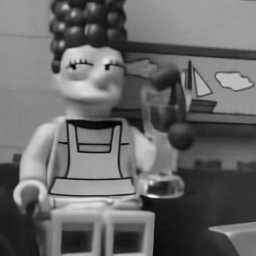} &
		\includegraphics[width=0.25\textwidth]{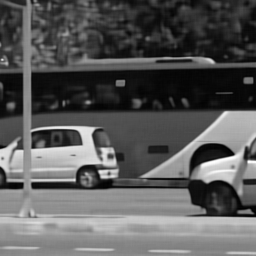}
		\\
		& \begin{small}PSNR=29.39/SSIM=0.93\end{small} & \begin{small}PSNR=33.33/SSIM=0.93\end{small} & \begin{small}PSNR=28.13/SSIM=0.85\end{small} 
	\end{tabular}
	\caption{Results comparison for the deblurring task of the images presented in \Cref{fig:blurred_images} using \pnpula \ with $\alpha=1$ (first row), \pnpsgd \ with $\alpha=0.3$ (second row) and $\alpha=1$  (third row). \pnpula \ was initialized with the observation $y$ (see \Cref{fig:blurred_images}) whereas \pnpsgd \  was initialised with a TVL2 restoration.
	}
	\label{fig:deblurring_res1}
\end{figure}

\begin{figure}[h!]
	\centering
	\begin{tabular}{cccc}
	    \CenteredVcell{0.2\textwidth}{\pnpula, $\alpha =1$.} &
	    \includegraphics[width=0.25\textwidth]{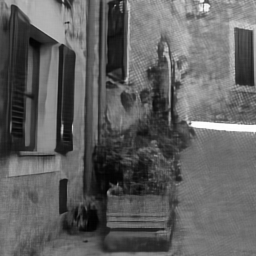} &
		\includegraphics[width=0.25\textwidth]{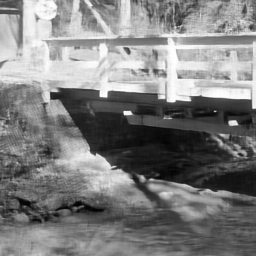} &
		\includegraphics[width=0.25\textwidth]{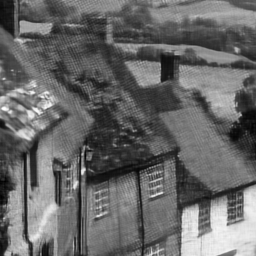}
        \\
		& \begin{small}PSNR=28.98/SSIM=0.80\end{small} & \begin{small}PSNR=28.28/SSIM=0.84\end{small} & \begin{small}PSNR=27.72/SSIM=0.73\end{small}
		\\
		\CenteredVcell{0.24\textwidth}{\pnpsgd, $\alpha =0.3$.} &
		\includegraphics[width=0.25\textwidth]{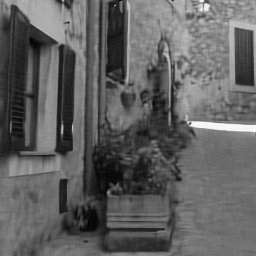} &
		\includegraphics[width=0.25\textwidth]{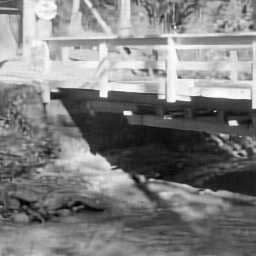} &
		\includegraphics[width=0.25\textwidth]{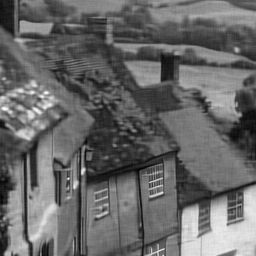} 
        \\
		& \begin{small}PSNR=29.26/SSIM=0.82\end{small} & \begin{small}PSNR=28.04/SSIM=0.84\end{small} & \begin{small}PSNR=28.27/SSIM=0.76\end{small} 
		\\
		\CenteredVcell{0.2\textwidth}{\pnpsgd, $\alpha =1$.} &
		\includegraphics[width=0.25\textwidth]{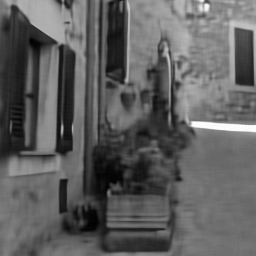}&
		\includegraphics[width=0.25\textwidth]{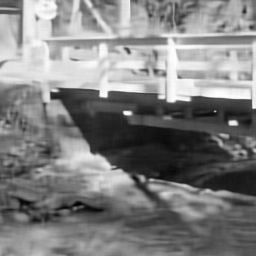} &
		\includegraphics[width=0.25\textwidth]{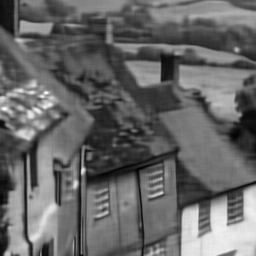}
		\\
		&\begin{small}PSNR=28.28/SSIM=0.76\end{small} & \begin{small}PSNR=27.14/SSIM=0.79\end{small} & \begin{small}PSNR=27.42/SSIM=0.70\end{small} 
	\end{tabular}
	\caption{Results comparison for the deblurring task of the images presented in \Cref{fig:blurred_images} using \pnpula \ with $\alpha=1$ (first row), \pnpsgd \ with $\alpha=0.3$ (second row) and $\alpha=1$  (third row). \pnpula \ was initialized with the observation $y$ (see \Cref{fig:blurred_images}) whereas \pnpsgd \  was initialised with a TVL2 restoration. 
	}
	\label{fig:deblurring_res2}
\end{figure}

\begin{figure}[h!]
	\centering
	\begin{tabular}{ccc}
	   
		\includegraphics[width=0.28\textwidth]{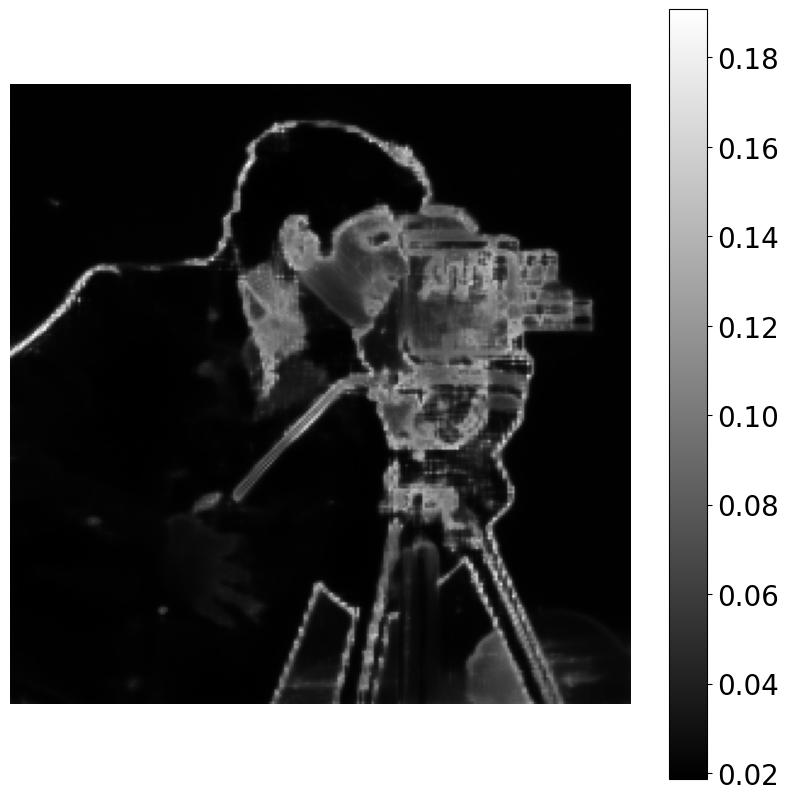} &
		
		\includegraphics[width=0.28\textwidth]{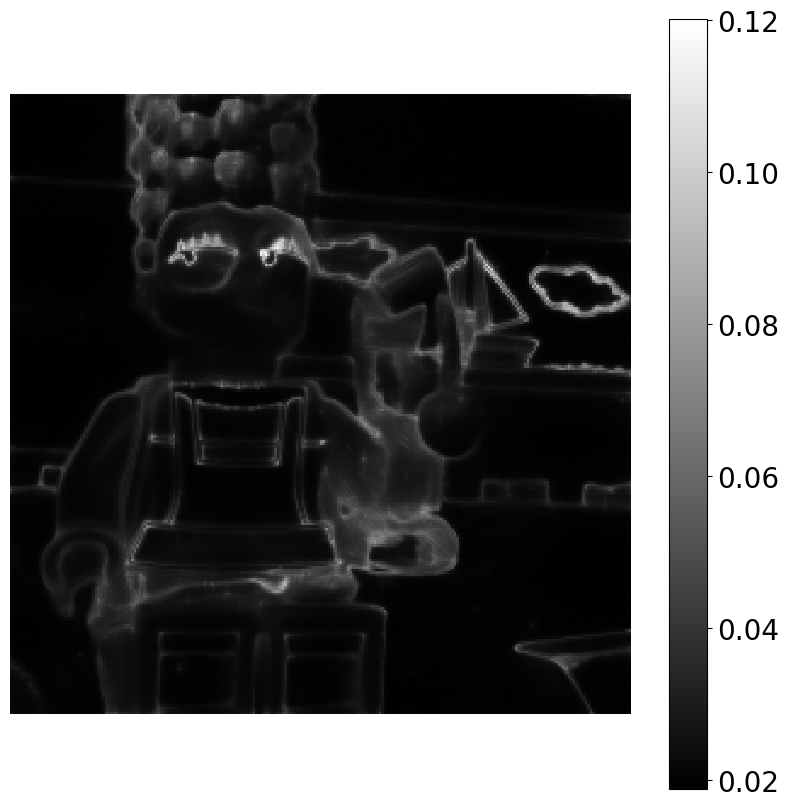} &
	
		\includegraphics[width=0.28\textwidth]{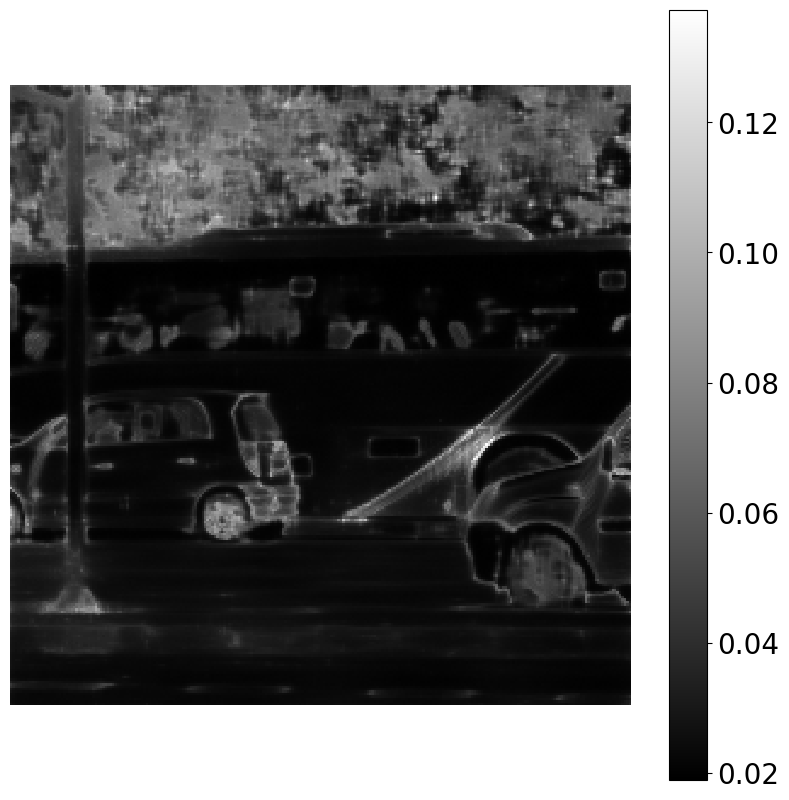} 
		
	    \\
	    \includegraphics[width=0.3\textwidth]{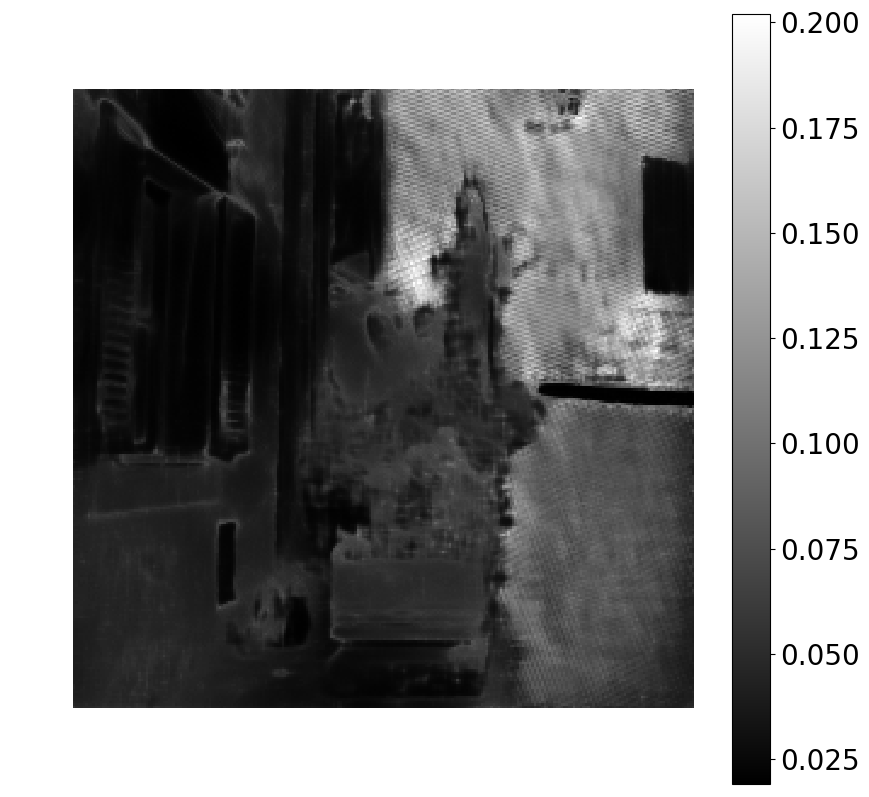} &
		
		\includegraphics[width=0.28\textwidth]{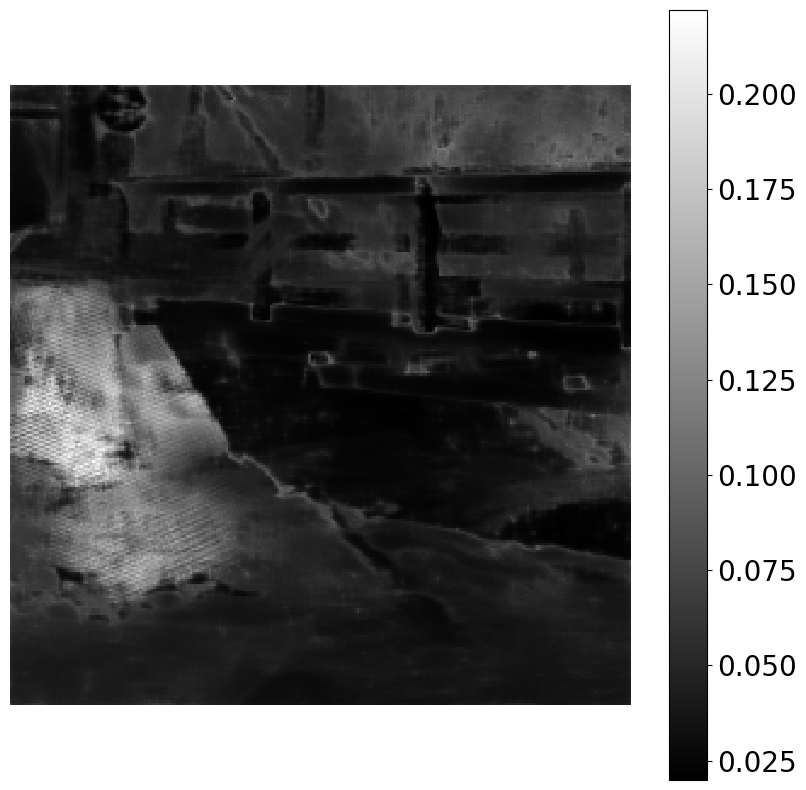} &

		\includegraphics[width=0.28\textwidth]{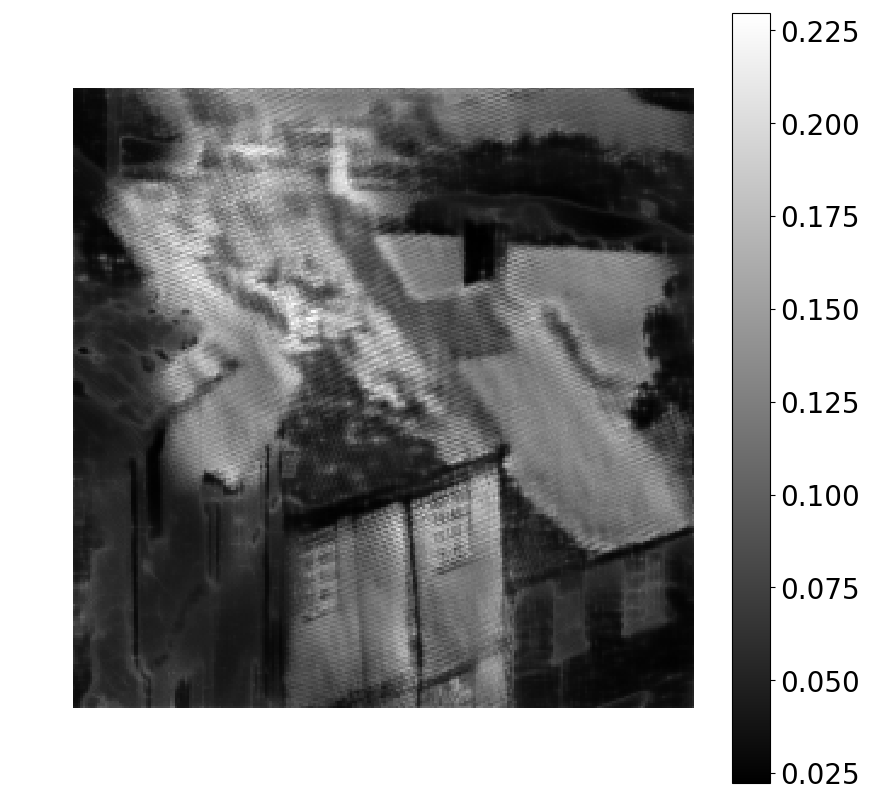}

	\end{tabular}
	\caption{Marginal posterior standard deviation for the deblurring problem. On simple images such as Simpson (see \cref{fig:original_images}), most of the uncertainty is located around the edges. For the images \texttt{Alley}, \texttt{Bridge} and \texttt{Goldhill}, associated with a highly correlated Markov chain in some directions, some areas are very uncertain. They correspond to the zones where the rotated rectangular pattern appears in the MMSE estimate.}
	\label{fig:deblurring_std}
\end{figure}

\paragraph{Visual results}
Figures~\ref{fig:inpainting_mmse1}, ~\ref{fig:inpainting_mmse2}, ~\ref{fig:deblurring_res1} and~\ref{fig:deblurring_res2}  show the MMSE estimate computed by \pnpula \ on the whole chain including the burn-in for the 6 images, for the inpainting and deblurring  experiments. We also provide the MAP estimation results computed by using \pnpsgd \ ~\cite{laumont2021maximum}, which targets the same posterior distributions. We report the \textit{Peak Signal-To Noise Ratio} (PSNR) and the \textit{Structural Similarity Index} (SSIM) \cite{psnr_default,ssim_definition} for all these experiments.

For the inpainting experiments, \pnpsgd \ struggles to converge when initialized with the observed image (see~\cite{laumont2021maximum}). For this reason, we warm start \pnpsgd \ by using an estimate of $x$ obtained by minimizing the Total Variation pseudo-norm under the constraint of the known pixels. For simplicity, \pnpula \ is initialized with the observation $y$. We observe in \Cref{fig:inpainting_mmse1} and ~\Cref{fig:inpainting_mmse2} that the results obtained by computing the MMSE Bayesian estimator with \pnpula \ are visually and quantitatively superior to the ones delivered by MAP estimation with \pnpsgd. In particular, the sampling approach seems to better recover the continuity of fine structures and lines in the different images. 

For the deblurring experiments, the results of \pnpsgd \ are provided by using a regularisation parameter $\alpha = 0.3$ (which was shown to yield optimal results on this set of images in~\cite{laumont2021maximum}) and for $\alpha = 1$, which recovers the model used by \pnpula. Observe that for the three first images (shown on~\Cref{fig:deblurring_res1}), the MMSE result is much sharper than the best MAP result, and the PSNR / SSIM results also show a clear advantage for the MMSE. For the other three images (results are shown on~\Cref{fig:deblurring_res2}), the quality of the MMSE solutions delivered is slightly deteriorated by the slow convergence of the Markov chain and the poor regularisation of some specific frequencies, which leads to a common visual artefact (a rotated rectangular pattern). Using a different denoiser more suitable for handling textures, or combining a learnt denoiser with an analytic regularisation term, might correct this behaviour and will be the topic of future work.

A partial conclusion from this set of comparisons is that the sampling approach of \pnpula, when it samples the space correctly, seems to provide much better results than the MAP estimator for the same posterior. Of course, this increase in quality comes at the cost of a much higher computation time.

\subsection{Deblurring and inpainting: uncertainty visualisation study}
\label{sec:statistical_study_deblurring}

\begin{figure}
    \centering
    \begin{tabular}{c c}
        \includegraphics[width = 0.3\textwidth]{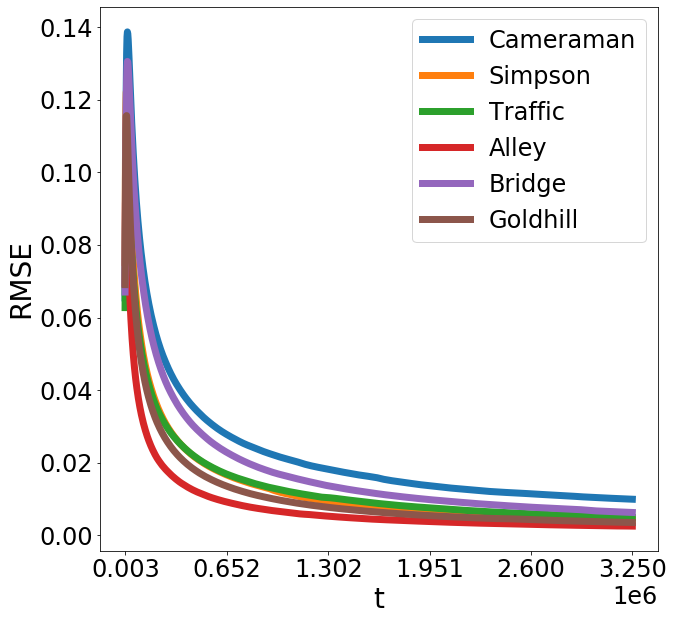} &  
        \includegraphics[width = 0.3\textwidth]{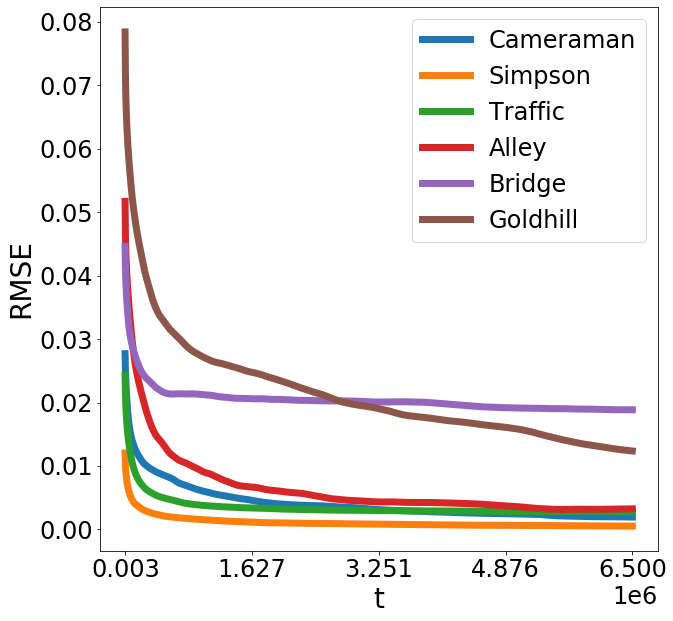}
        \\
        Inpainting. & Deblurring. 
    \end{tabular}
    \caption{Evolution of the Root Mean Squared Error (RMSE) between the final standard deviation and the estimated current standard deviation for the inpainting and deblurring problems.}
    \label{fig:evolution_rmse}
\end{figure}

One of the benefits of sampling from the posterior distribution with \pnpula \ is that we can probe the uncertainty in the delivered solutions. In the following, we present an uncertainty visualisation analysis that is useful for displaying the uncertainty related to image structures of different sizes and located in different regions of the image (see \cite{Cai2018} for more details). The analysis proceeds as follows. First, \Cref{fig:inpainting_std} and \Cref{fig:deblurring_std} show the marginal posterior standard deviation associated with each image pixel, as computed by \pnpula \ over all samples, for the inpainting and deblurring problems. As could be expected,  we observe for both problems that highly uncertain pixels are concentrated around the edges of the reconstructed images, but also on textured areas. The dynamic range of the pixel standard deviations is larger for the inpainting problem than for deblurring, which  suggests that the problem has a higher level of intrinsic uncertainty.

\Cref{fig:evolution_rmse} shows the evolution of the RMSE between the standard deviation computed along the samples and its asymptotic value, respectively for the inpainting and deblurring problems. Estimating these standard deviation maps necessitates to run the chain longer than to estimate the MMSE, as could be expected for second order statistical moment.

\begin{figure}
    \centering
    \begin{tabular}{c c c c}
    \includegraphics[width=0.22\textwidth]{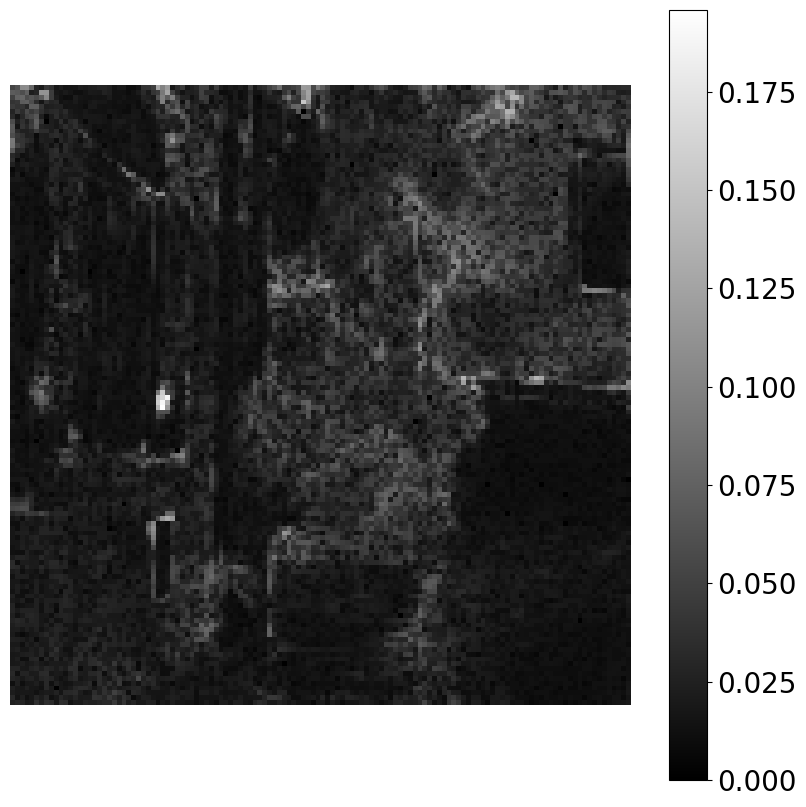} &
    \includegraphics[width=0.22\textwidth]{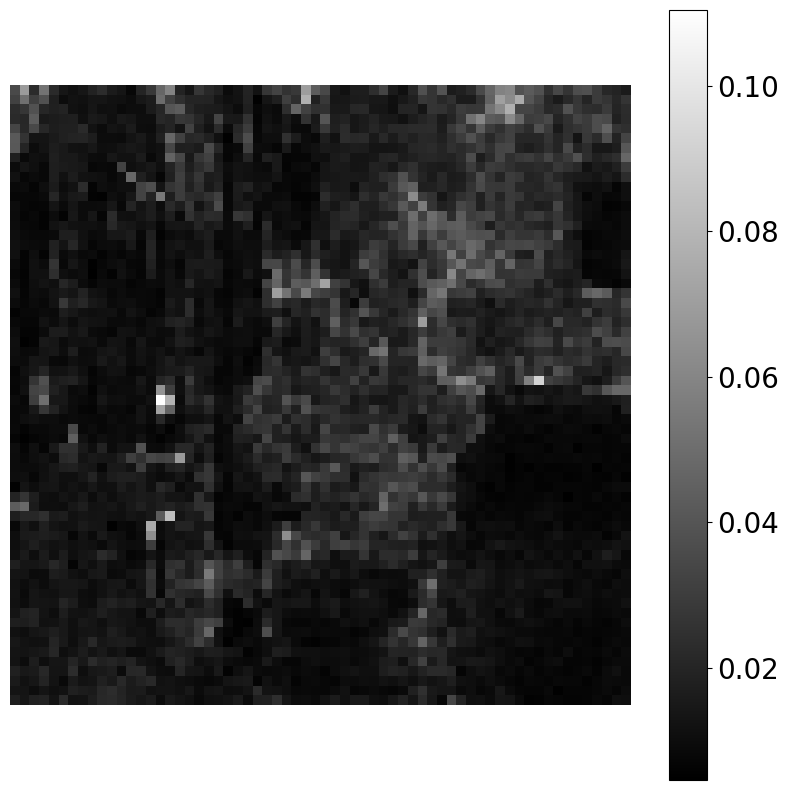} &
    \includegraphics[width=0.22\textwidth]{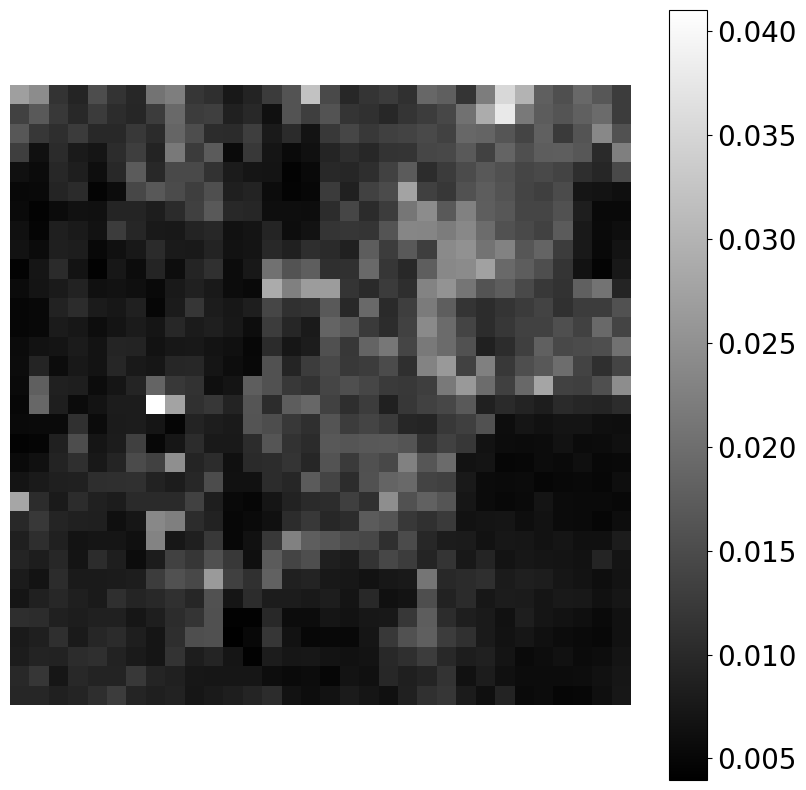} &
    \includegraphics[width=0.22\textwidth]{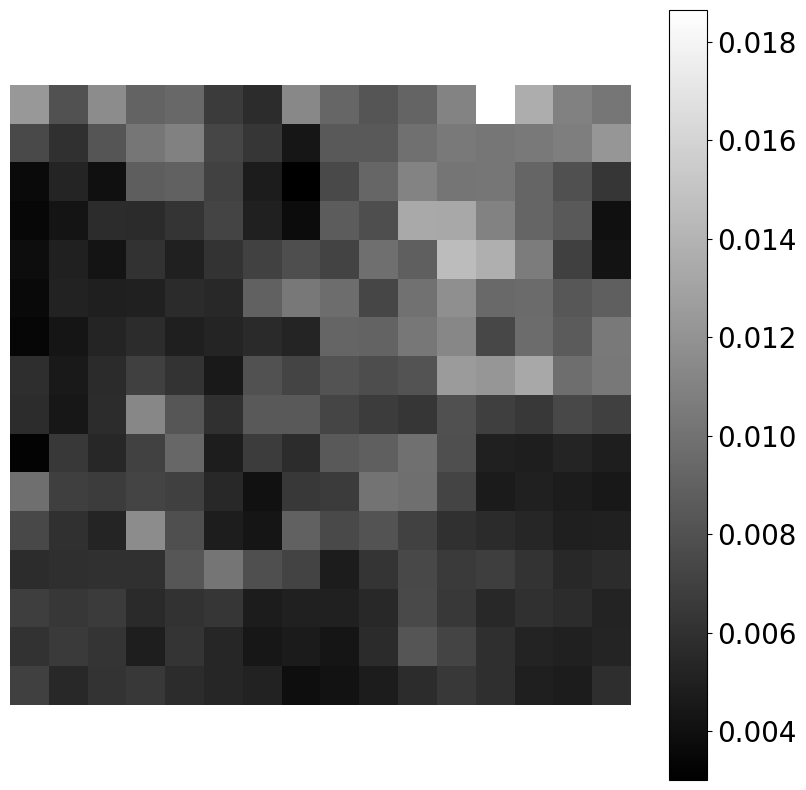}  
     \\
    \includegraphics[width=0.22\textwidth]{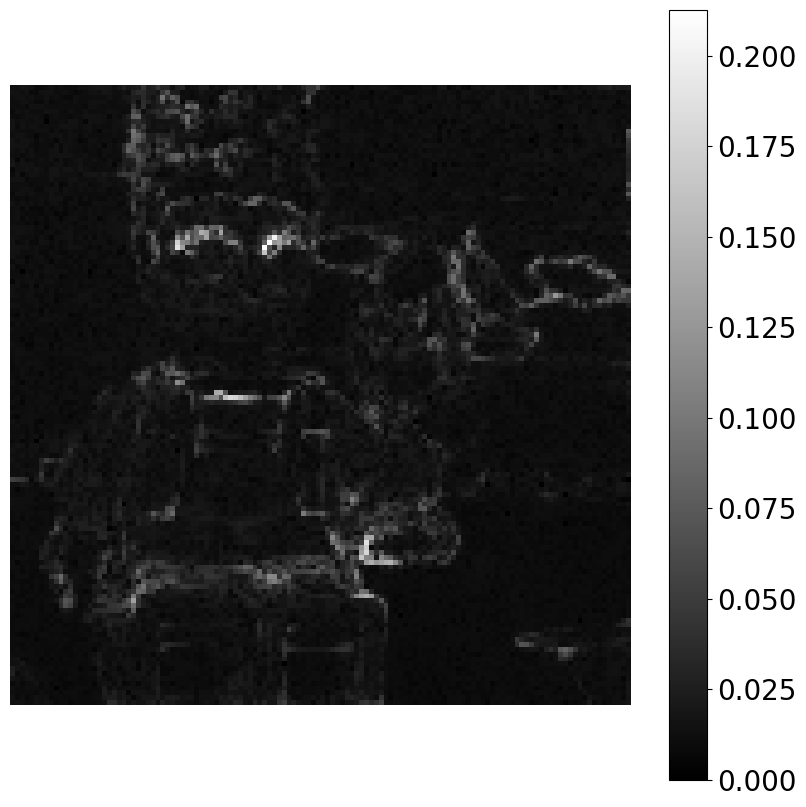} &
    \includegraphics[width=0.22\textwidth]{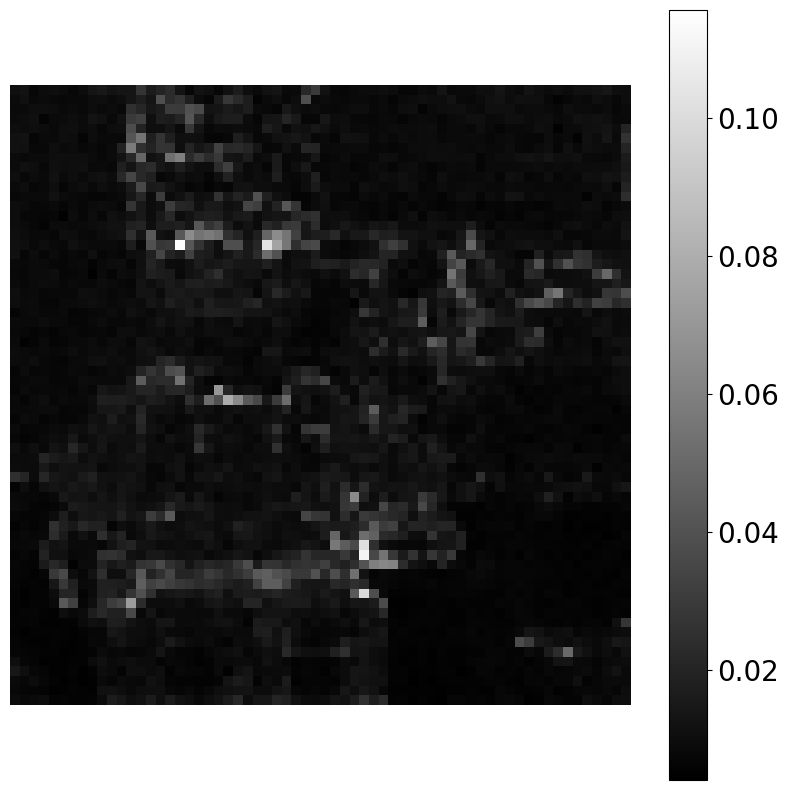} &
    \includegraphics[width=0.22\textwidth]{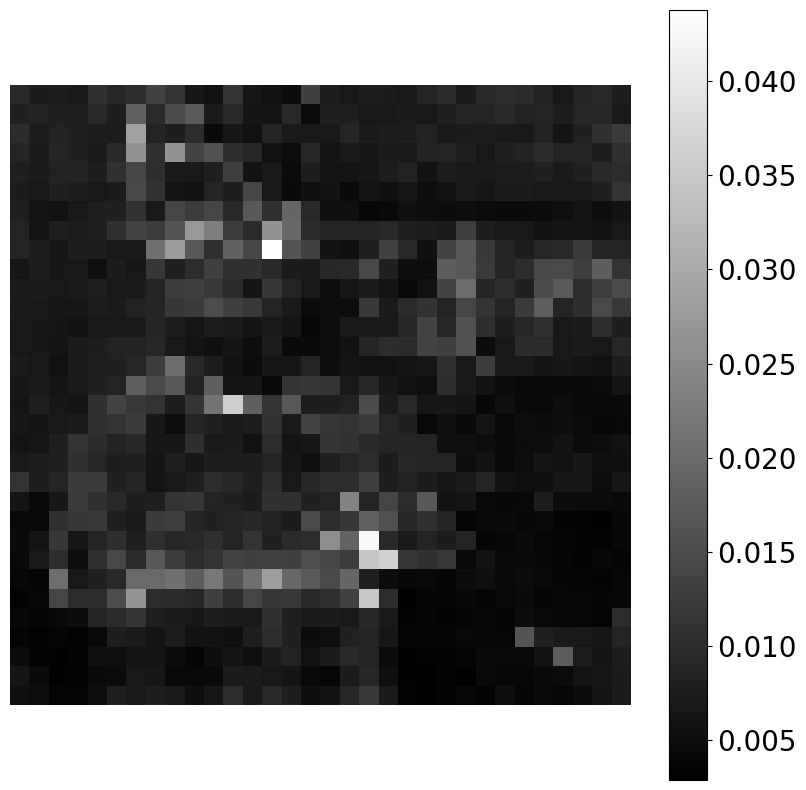} &
    \includegraphics[width=0.22\textwidth]{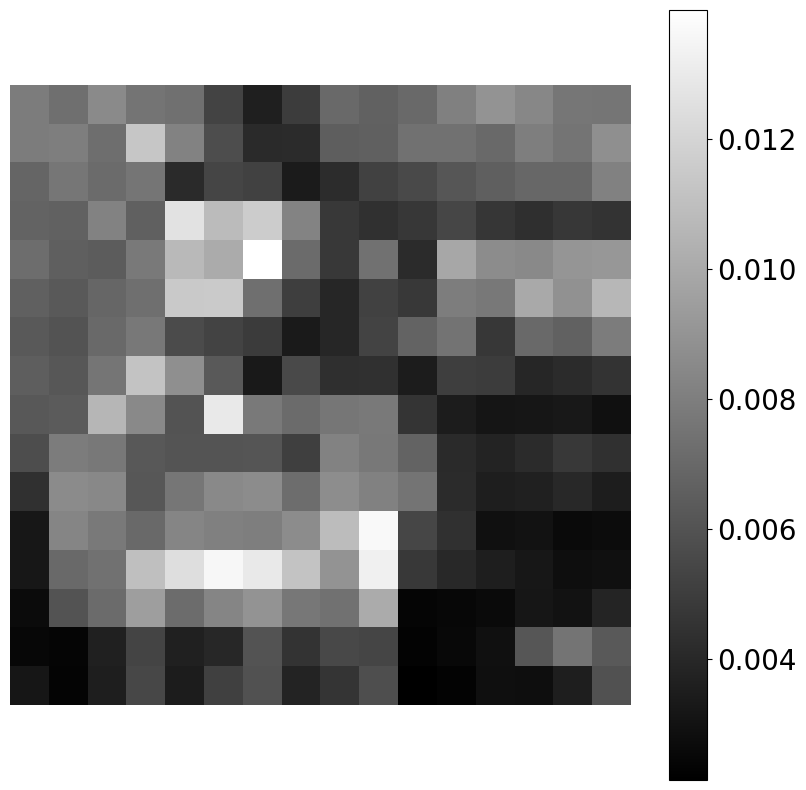} 
    \\
    Scale 1 & Scale 2 & Scale 3 & Scale 4
    \end{tabular}

    \caption{Marginal posterior standard deviation of the \texttt{Alley} and \texttt{Simpson} images for the inpainting problem at different scales. The scale $i$ corresponds to a downsampling by a factor $2 i$ of the original sample size.}
    \label{fig:inpainting_scaled_std}
\end{figure}

\begin{figure}
    \centering
    \begin{tabular}{c c c c}
    \includegraphics[width=0.22\textwidth]{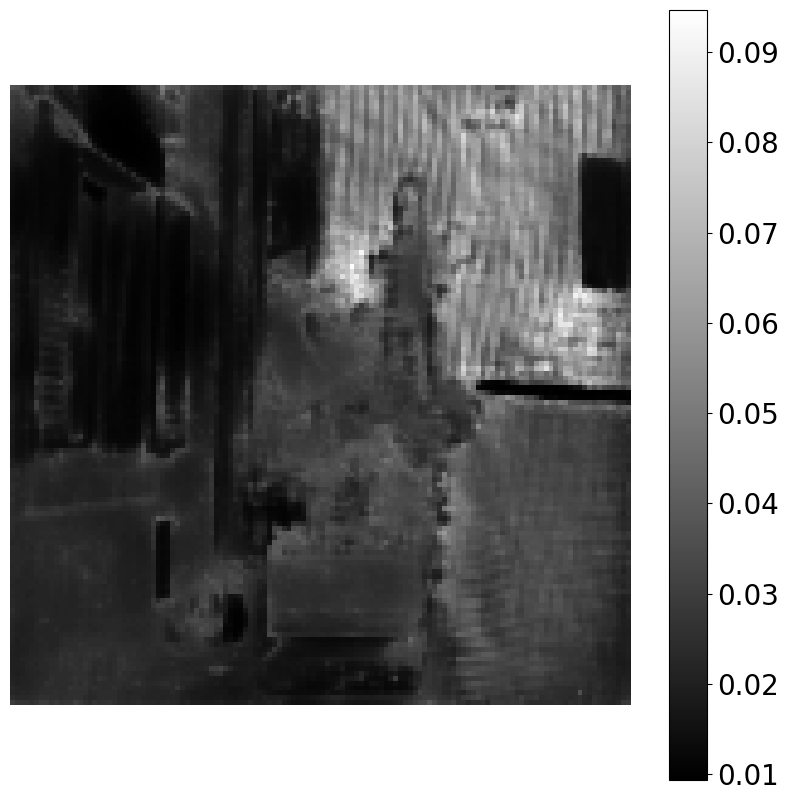} &
    \includegraphics[width=0.22\textwidth]{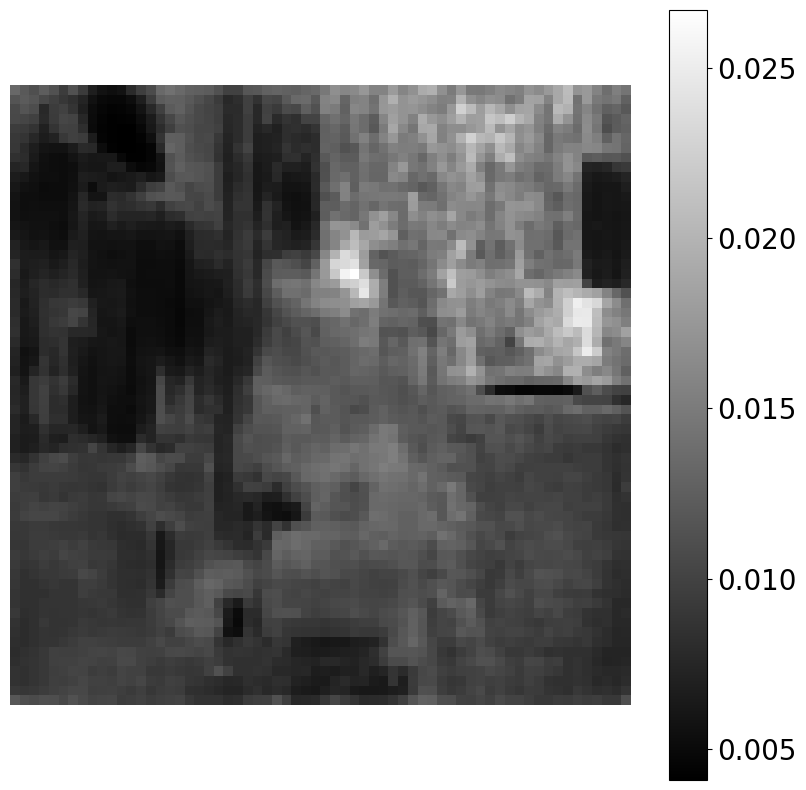} &
    \includegraphics[width=0.22\textwidth]{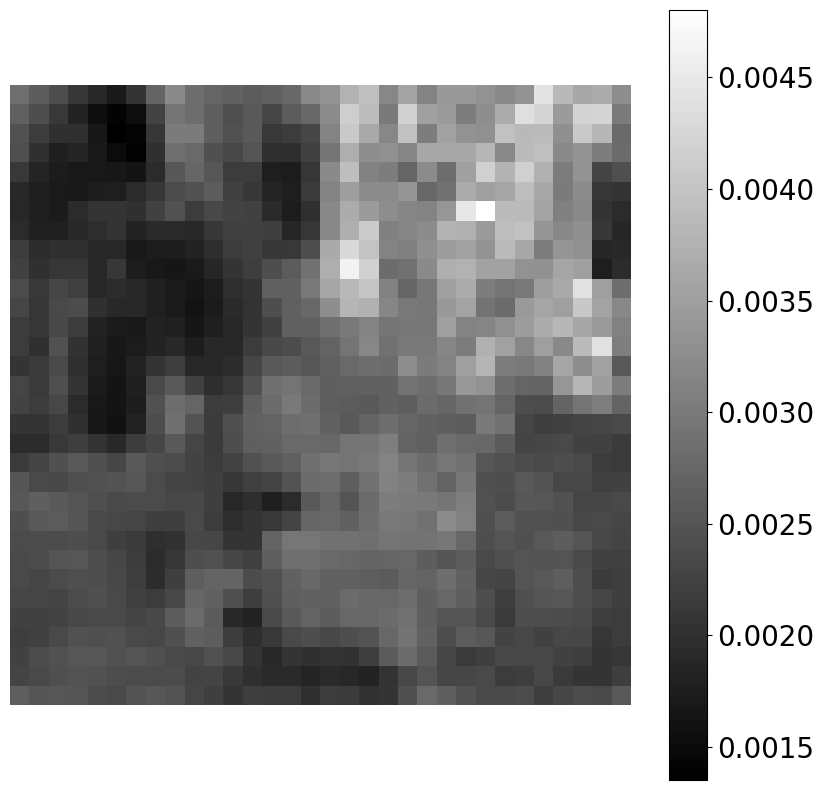} &
    \includegraphics[width=0.22\textwidth]{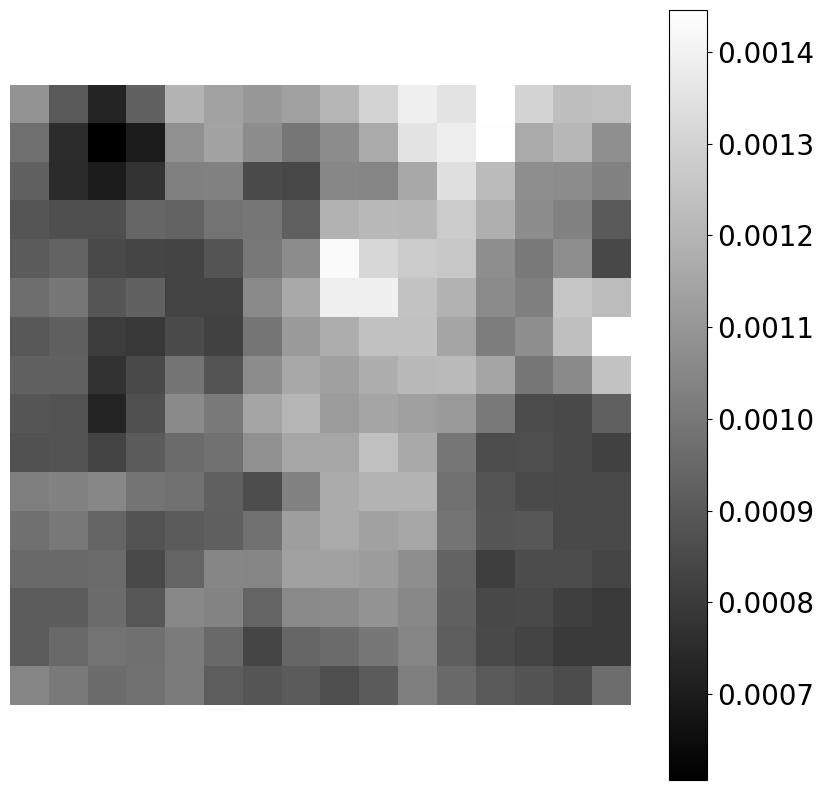}  
     \\
    \includegraphics[width=0.22\textwidth]{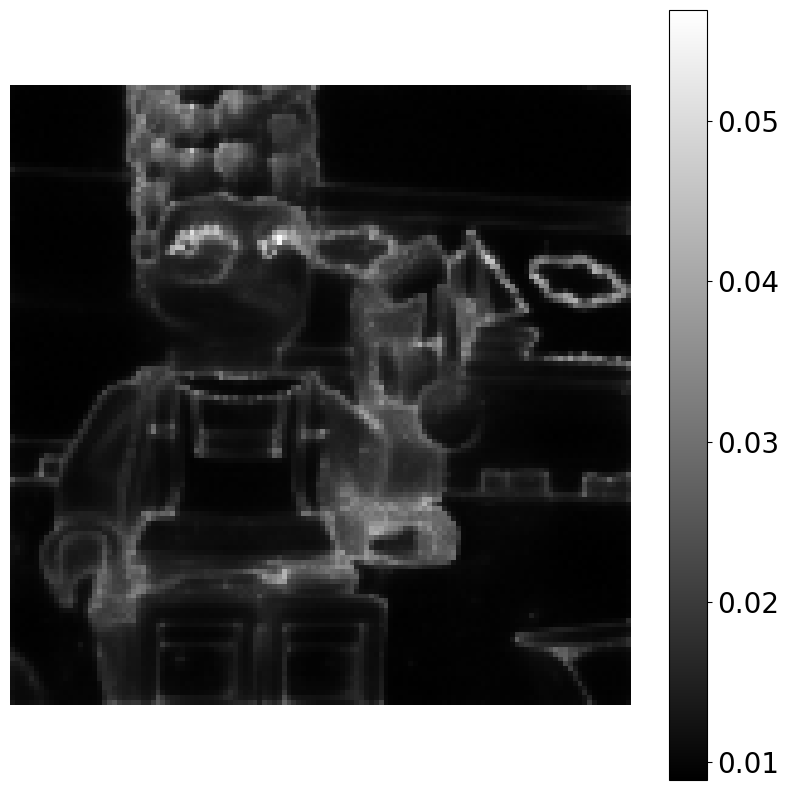} &
    \includegraphics[width=0.22\textwidth]{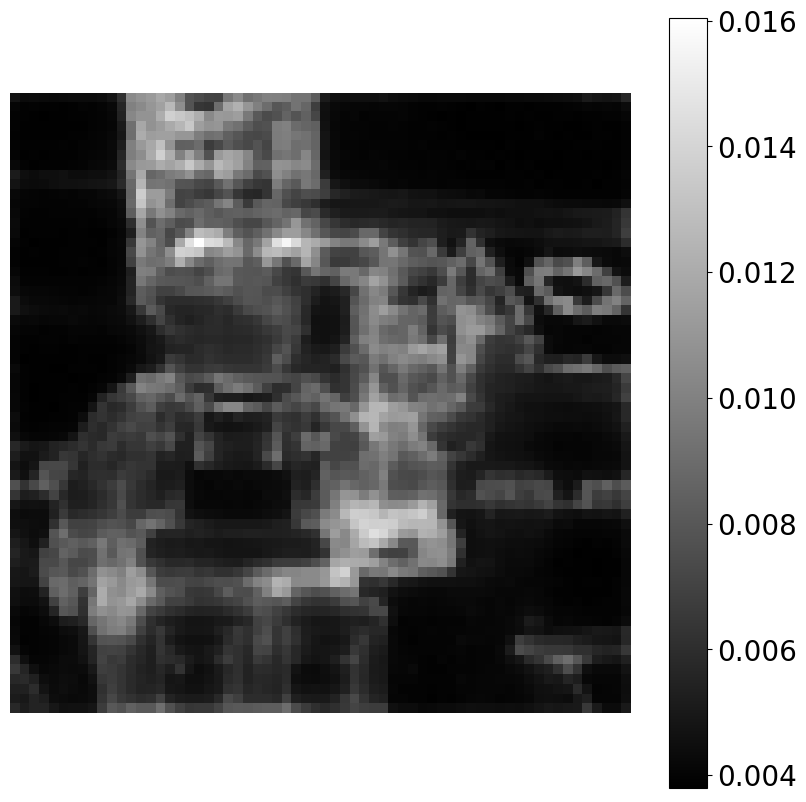} &
    \includegraphics[width=0.22\textwidth]{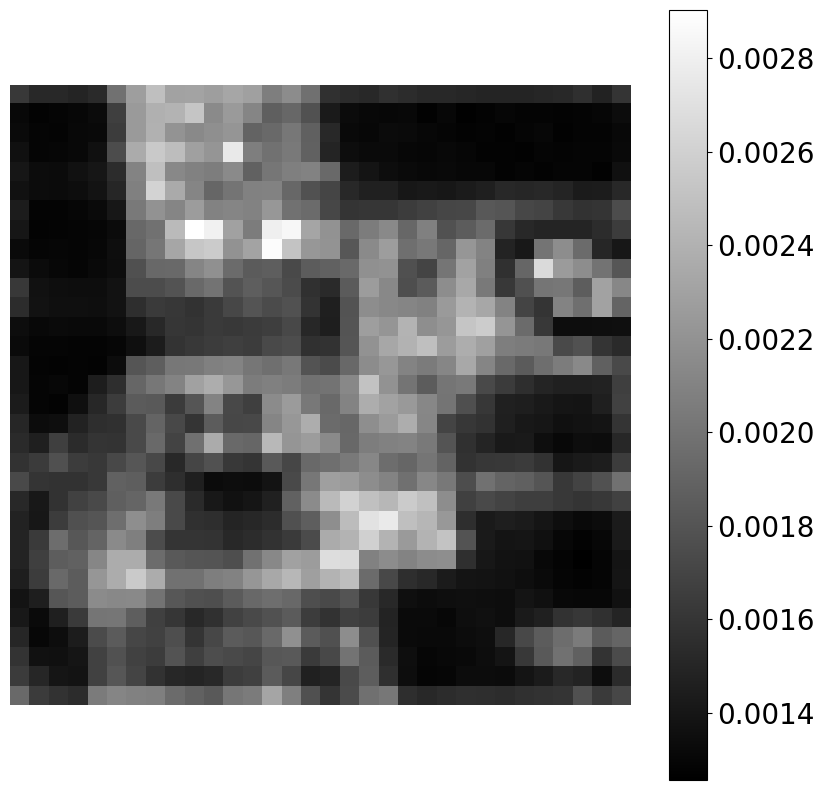} &
    \includegraphics[width=0.22\textwidth]{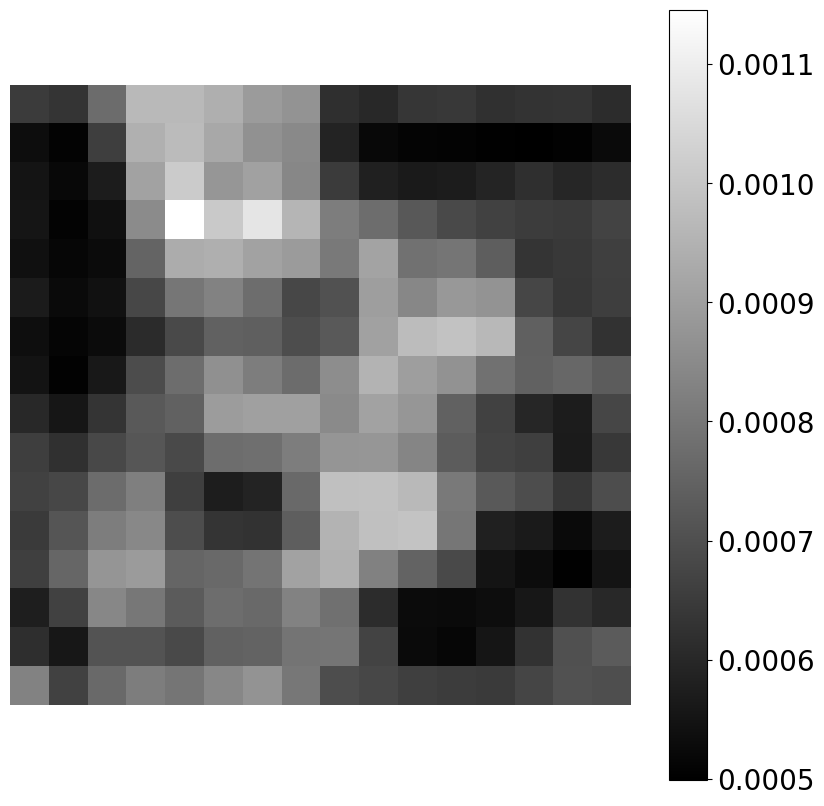} 
    \\
    Scale 1 & Scale 2 & Scale 3 & Scale 4
    \end{tabular}

    \caption{Marginal posterior standard deviation of the images \texttt{Alley} and \texttt{Simpson} for the deblurring problem at different scales. The scale $i$ corresponds to a downsampling by a factor $2 i$ of the original sample size.}
    \label{fig:deblurring_scaled_std}
\end{figure}

Following on from this, to explore the uncertainty for structures that are larger than one pixel, \Cref{fig:inpainting_scaled_std} and~\Cref{fig:deblurring_scaled_std} report the marginal standard deviation associated with higher scales. More precisely, for different values of the scale $i$, we downsample the stored samples by a factor $2i$ before computing the standard deviation. This downsampling step permits quantifying the uncertainty of larger or lower-frequency structures, such as the bottom of the glass in \texttt{Simpson} for the deblurring experiment.
At each scale, we see that the uncertainty of the estimate is much more localized for the inpainting problem (resulting in higher uncertainty values in some specific regions) and more spread out for deblurring, certainly because of the different nature of the degradations involves.

	\section{Conclusion}
\label{sec:conclusion}

This paper presented theory, methods, and computation algorithms for performing Bayesian inference with \emph{Plug \& Play} priors. This mathematical and computational framework is rooted in the Bayesian \emph{M-complete} paradigm and adopts the view that \emph{Plug \& Play} models approximate a regularised oracle model. We established clear conditions ensuring that the involved models and quantities of interest are well defined and well posed. Following on from this, we studied three Bayesian computation algorithms related to biased approximations of a Langevin diffusion process, for which we provide detailed convergence guarantees under easily verifiable and realistic conditions. For example, our theory does not require the denoising algorithms representing the prior to be gradient or proximal operators. We also studied the estimation error involved in using these algorithms and models instead of the oracle model, which is decision-theoretically optimal but intractable. To the best of our knowledge, this is the first Bayesian \emph{Plug \& Play} framework with this level of insight and guarantees on the delivered solutions. We illustrated the proposed framework with two Bayesian image restoration experiments  - deblurring and inpainting - where we computed point estimates as well as uncertainty visualisation and quantification analyses and highlighted how the limitations of the chosen denoiser manifest in the resulting Bayesian model and estimates. 

In future work, we would like to continue our theoretical and empirical
investigation of Bayesian \emph{Plug \& Play} models, methods and algorithms. From a modelling viewpoint, it would be interesting to consider priors that combine a denoiser with an analytic regularisation term, and 
other neural network based priors such as the generative ones used in \cite{bora2017compressed} or the autoencoder-based priors
in~\cite{Gonzalez2021}, as well as to generalise the Gaussian smoothing to other smoothings and investigate their properties in the context of Bayesian inverse problems. We are also very interested in strategies for training denoisers that automatically verify the conditions required for exponentially fast convergence of the Langevin SDE, for example by using the framework recently proposed in \cite{pesquet2020learning} to learn maximally monotone operators, or the data-driven regularisers described in \cite{Kobler_2020_CVPR,mukherjee2021learned}. In addition, we would like to understand when the projected RED estimator \cite{cohen2020regularization} - or its relaxed variant - are the MAP estimators for well-defined Bayesian models, as well as to study the interplay between the geometric aspects of the loss defining this estimator \cite{pereyra2019b} and the geometry of the set of fixed points of the denoiser defining the model. With regards to Bayesian analysis, it would be important to investigate the frequentist accuracy of \emph{Plug \& Play} models, as well as the adoption of robust Bayesian techniques in order to perform inference directly w.r.t. to the oracle model \cite{Watson2016}. From a Bayesian computation viewpoint, a priority is to develop accelerated algorithms similar to \cite{pereyra2020accelerating}. Lastly, with regards to experimental work, we intend to study the application of this framework to uncertainty quantification problems, e.g., in the context of medical
imaging. 


	\bibliographystyle{siamplain}
	\bibliography{pnp}
	
	\appendix
	\section{Organization of the supplementary}\label{sec:supplemantary}

In this supplementary document we present some extensions and gather the proofs
of this paper. We first introduce a more general framework in
\Cref{sec:general-framework}. Then in \Cref{sec:strongly-log-concave} we present
our improved convergence results in the case where the log-likelihood is
strongly log-concave. Posterior approximation bounds in our general setting are
gathered in \Cref{sec:posterior_approx}. Then we turn to the proof of these
results. We first derive technical results in
\Cref{sec:technical-results}. Proofs of
\Cref{sec:convergence_pnpula} and
\Cref{sec:proj-altern} are presented in
\Cref{sec:proofs-citecr} and \Cref{sec:proofs-citecr-altern} respectively.
Finally, proofs of \Cref{sec:posterior_approx} are given in
\Cref{sec:proofs-citecr-1}.

\section{A general framework}
\label{sec:general-framework}

We start by considering a slightly more general framework than the one
previously introduced. More precisely, instead of $p^\star$ we
consider a general distribution $p$ and instead of considering $p_\vareps$
as a prior we consider a tamed version of this density by introducing another
hyperparameter $\alpha > 0$. In what follows, we describe this setting in
details. We start by recalling a mild assumption on the
likelihood.

\begingroup
\def\theassumption{\ref{assum:post}}
\begin{assumption}
	For any $y \in \rset^\dimY$, $\sup_{x \in \rset^d} p(y|x) < +\infty$,
	$p(y|\cdot) \in \rmc^1(\rset^d, \ooint{0, +\infty})$ and there exists
	$\Ltt_y > 0$ such that $\nabla \log (p(y|\cdot))$ is $\Ltt_y$ Lipschitz
	continuous.
\end{assumption}
\addtocounter{assumption}{-1}
\endgroup

For any $\vareps > 0$ we recall that
$p_{\vareps}$ is given by the Gaussian smoothing of $p$ with level
$\vareps$, for any $x \in \rset^d$ by 
\begin{equation}
  \textstyle{
    p_{\vareps}(x) = (2 \uppi \vareps)^{-d/2} \int_{\rset^{\dim}} \exp[-\norm{x - \tilde{x}}^2 / (2 \vareps)] \ p(\tilde{x}) \rmd \tilde{x} \eqsp .
    }
\end{equation}
One typical example of likelihood function that we consider in our numerical
illustration, see \Cref{sec:experimental}, is
$p(y|x) \propto \exp[-\norm{\rmA x - y}^2/(2 \sigma^2)]$ for any $x \in \rset^d$ with
$\sigma > 0$ and $\rmA \in \rset^{m \times \dim}$. Before turning to the
analysis of the convergence of the introduced algorithms we state the following
proposition which ensures the regularity of the posterior model w.r.t to the
observation $y$.

We consider the following assumption on $x \mapsto p(y|x)$ and the prior $p$ for some
hyperparameter $\alpha > 0$ and an observation $y \in \rset^\dimY$.
\begin{assumption}
  \label{assum:finite}
  The following hold:
    \begin{enumerate}[label=(\alph*), leftmargin=4em]

    \item \label{item:properr} $\int_{\rset^{\dim}} p(y|\tilde{x})  
      p^{\alpha}(\tilde{x}) \rmd \tilde{x} < + \infty$ and for any $\vareps > 0$, $\int_{\rset^{\dim}} p(y|\tilde{x}) 
  p_{\vareps}^{\alpha}(\tilde{x}) \rmd \tilde{x} < + \infty$.
  \item \label{item:square} $\int_{\rset^d} \normLigne{\tilde{x}}^2 p(x) \rmd x < +\infty$.
 
  \end{enumerate}
\end{assumption}
Note that if $\alpha = 1$, \Cref{assum:finite}-\ref{item:properr} hold under \Cref{assum:post},
see \Cref{prop:prop_eps}.  Under \Cref{assum:finite}-\ref{item:properr}, define $\posteriorw$ the
target probability distribution for any $x \in \rset^d$ by
\begin{equation}
  \label{eq:posterior}
  (\rmd \posteriorw / \rmd \Leb) (x) = \left. p(y|x)  p^{\alpha}(x) \middle/ \int_{\rset^{\dim}} p(y|\tilde{x}) p^{\alpha}(\tilde{x}) \rmd \tilde{x} \right. \eqsp .
\end{equation}
Note that for ease of notation, we do not explicitly highlight the dependency of
the posterior distribution $\pi$ with respect to the hyperparameter
$\alpha > 0$, since it is fixed in the rest of this section. We also consider the family of probability distributions
$\ensembleLigne{\posteriorweps}{\vareps >0}$ given for any $\vareps > 0$ and
$x \in \rset^{\dim}$ by
\begin{equation}
  \label{eq:posterior_eps}
  (\rmd \posteriorweps / \rmd \Leb) (x) = \left. p(y|x) p_{\vareps}^{\alpha}(x) \middle/ \int_{\rset^{\dim}} p(y|\tilde{x}) p_{\vareps}^{\alpha}(\tilde{x}) \rmd \tilde{x} \right. \eqsp .
\end{equation}

We also recall the assumption on the denoiser $D_\vareps$, see \Cref{sec:convergence_pnpula} for details.

\begingroup
\def\theassumption{\ref{assum:neural_net}}
\begin{assumption}
	There exist $\bvareps> 0$, $\Mtt_R \geq 0$ and
	$\Ltt \geq 0$ such that for any $\vareps \in \ocint{0, \bvareps}$,
	$x_{1}, x_{2} \in \rset^d$ and $x \in \cball{0}{R}$ we have
	\begin{equation}
	\norm{(\Id - D_{\vareps})(x_{1}) - (\Id
		- D_{\vareps})(x_{2})} \leq \Ltt \norm{x_{1} - x_{2}} \eqsp , \qquad  
	\norm{D_{\vareps}(x) - D_{\vareps}^{\star}(x)} \leq \Mtt_R \eqsp ,
	\end{equation}
	where we recall that \begin{equation}
	\textstyle{
		D_{\vareps}^{\star}(x_{1}) =  \int_{\rset^{\dim}} \tilde{x} \ g_{\vareps}(\tilde{x} | x_{1}) \rmd \tilde{x} \eqsp .
	}
	\end{equation}
\end{assumption}
\addtocounter{assumption}{-1}
\endgroup

\section{Strongly log-concave case}
\label{sec:strongly-log-concave}

We now present an improvement on the results of \Cref{sec:convergence_pnpula} in
the case where the log-likelihood $x \mapsto \log p(y|x)$ is strongly
concave. We recall that the Markov chain is given by the following
      recursion: $X_0 \in \rset^d$ and for any $k \in \nset$
\begin{align}
  \label{eq:pnpula}
  X_{k+1} &= X_k + \delta b_{\vareps}(X_k)  + \sqrt{2 \delta} Z_{k+1} \eqsp , \\
  b_{\vareps}(x) &= \nabla \log p(y|x) + \alpha  P_{\vareps}(x) + (\prox_{\lambda}(\iota_{\msc})(x)-x)/\lambda  \eqsp , \quad   P_{\vareps}(x) = (D_{\vareps}(x) - x) / \vareps \eqsp ,
\end{align}
In the strongly concave setting we set $\msc = \rset^d$, \ie
$\forall x \in \msc,\ \prox_{\lambda}(\iota_{\msc})(x) = x$. We recall that in our image processing
applications, we have that for any $x \in \rset^d$,
$p(y|x) \propto \exp[-\norm{\rmA x - y}^2 / (2 \sigma^2)]$ and that
$x \mapsto p(y|x)$ is strongly log-concave if and only if $\rmA$ is
invertible. This is the case for denoising tasks where $\rmA = \Id$ and for
deblurring tasks with convolution kernels which have full Fourier support.

We start with the following result which
ensures that the Markov chain \eqref{eq:pnpula} is geometrically
ergodic under \rref{assum:neural_net} for the Wasserstein metric
$\wassersteinD[1]$ and in $V$-norm for $V: \ \rset^{\dim} \to \coint{1, +\infty}$ given for any $x \in \rset^{\dim}$ by
\begin{equation}
  \label{eq:V_def}
V(x) = 1 + \norm{x}^2 \eqsp .  
\end{equation}

The following proposition is the counterpart of \Cref{prop:ergo_C}.

\begin{proposition}
  \label{prop:ergo}
  Assume \tup{\Cref{assum:post}}, \tup{\Cref{assum:finite}} and \tup{\Cref{assum:neural_net}($R$)} for some
  $R > 0$. Let $\alpha >0$ and $\vareps \in \ocint{0, \vareps_0}$.  If there
  exists $\mtt >0$ such that $\log(p(y|\cdot))$ is $\mtt$-concave with
  $\mtt \geq 2 \alpha \Ltt / \vareps$ then there exist $A_1 \geq 0$ and
  $\rho_1 \in \coint{0,1}$ such that for any
  $\delta \in \ocintLigne{0, \bdelta}$, $x_1, x_2 \in \rset^{\dim}$ and
  $k \in \nset$ we have
  \begin{align}
    \Vnorm{\updelta_{x_1} \Rker_{\vareps, \delta}^k-  \updelta_{x_2} \Rker_{\vareps, \delta}^k} &\leq A_1 \rho_1^{k \delta} (V^2(x_1) + V^2(x_2))  \eqsp , \\
    \wassersteinD[1](\updelta_{x_1} \Rker_{\vareps, \delta}^k, \updelta_{x_2} \Rker_{\vareps, \delta}^k) &\leq A_1 \rho_1^{k \delta} \norm{x_1-x_2}  \eqsp ,
  \end{align}
  where $V$ is given in \eqref{eq:V_def} and $\bdelta = \mtt (\Ltt_y + \alpha \Ltt/\vareps)^{-2}/2$.
\end{proposition}

\begin{proof}
  The proof is postponed to \Cref{prop:ergo:proof}.
\end{proof}

We recall the assumption on $g_\vareps$ which ensures that
$x \ \mapsto \log(p_{\vareps}(x))$ has Lipschitz gradients.

\begingroup
\def\theassumption{\ref{assum:cov}}
\begin{assumption}
	For any $\vareps > 0$, there exists $\Ktt_{\vareps} \geq 0$ such that for any $x \in \rset^d$,
	\begin{align}
	&\int_{\rset^d} \norm{\tilde{x} - \int_{\rset^d} \tilde{x}' g_{\vareps}(\tilde{x}'| x) \rmd \tilde{x}'}^2 g_{\vareps}(\tilde{x}| x) \rmd \tilde{x} \leq \Ktt_{\vareps} \eqsp ,
	\end{align}
	with $g_{\vareps}$ given in \eqref{eq:cond_prior}.
\end{assumption}
\addtocounter{assumption}{-1}
\endgroup

The following proposition is the counterpart of \Cref{prop:bias_C}.
  
\begin{proposition}
  \label{prop:bias}
  Assume \rref{assum:post}, \tup{\Cref{assum:finite}},
  \tup{\Cref{assum:neural_net}($R$)} for some $R > 0$ and \rref{assum:cov}. Moreover, let $\alpha > 0$,
  $\vareps \in \ocint{0, \vareps_0}$ and assume that
  $\int_{\rset^{\dim}} (1+\norm{\tilde{x}}^4) p_{\vareps}^{\alpha}(\tilde{x})
  \rmd \tilde{x} < + \infty$.  In addition, if there exists $\mtt >0$ such that
  $\log(p(y|\cdot))$ is $\mtt$-concave with
  $\mtt \geq (2\alpha/\vareps) \max( \Ltt, 1 + \Ktt_{\vareps}/\vareps)$ and
  $\bdelta = \mtt (\Ltt_y + \alpha \Ltt/\vareps)^{-2}/2$, then for any
  $\delta \in \ocintLigne{0, \bdelta}$, $\Rker_{\vareps, \delta}$ admits an
  invariant probability measure $\posteriorwepsdelta$ and there exists
  $B_1 \geq 0$ such that for any $\delta \in \ocintLigne{0, \bdelta}$
  \begin{equation}
    \label{eq:control}
    \Vnorm{\posteriorwepsdelta -  \posteriorweps} \leq B_1 (\delta^{1/2} + \Mtt_R + \exp[-R]) \eqsp ,
  \end{equation}
  where $V$ is given in \eqref{eq:V_def} and $B_1$ does not depend on $R$.
\end{proposition}

\begin{proof}
  The proof is postponed to \Cref{prop:bias:proof}.
\end{proof}

The bound appearing in \eqref{eq:control} depends on an extra hyperparameter
$R > 0$ which may be optimized if \tup{\Cref{assum:neural_net}($R$)} holds for
any $R > 0$ and $\ensembleLigne{\Mtt_R}{R > 0}$ can be expressed in a closed
form. In particular if there exists $\Mtt \in \ooint{0,1}$ such that for any
$R > 0$, $\Mtt_R = \Mtt \times R$ then there exists $B_1 \geq 0$ such that for
any $\delta \in \ocintLigne{0, \bdelta}$ and $R > 0$
  \begin{equation}
    \Vnorm{\posteriorwepsdelta -  \posteriorweps} \leq B_1 (\delta^{1/2} + \Mtt \log(1/\Mtt)) \eqsp ,
  \end{equation}
  by setting $R = \log(1/\Mtt)$. Similarly if there exists $\Mtt > 0$ such that for any $R > 0$, $\Mtt_R = \Mtt$ then  there exists $B_1 \geq 0$ such that for any
  $\delta \in \ocintLigne{0, \bdelta}$ and $R > 0$
  \begin{equation}
    \Vnorm{\posteriorwepsdelta -  \posteriorweps} \leq B_1 (\delta^{1/2} + \Mtt) \eqsp ,
  \end{equation}
  by letting $R \to +\infty$.

  We now combine \Cref{prop:ergo} and \Cref{prop:bias} in order to control the
  bias of the Monte Carlo estimator obtained using \pnpula . This proposition is
  the counterpart of \Cref{prop:bias_control_final_C}.

\begin{proposition}
  \label{prop:bias_control_final}
  Assume \rref{assum:post}, \tup{\Cref{assum:finite}}, \tup{\Cref{assum:neural_net}($R$)} for some
  $R > 0$ and \rref{assum:cov}.
  Moreover, let $\alpha >0$, $\vareps \in \ocint{0, \vareps_0}$ and assume that
  $\int_{\rset^{\dim}} (1+\norm{\tilde{x}}^4) p_{\vareps}^{\alpha}(\tilde{x})
  \rmd \tilde{x} < + \infty$.  In addition, if there exists $\mtt >0$ such that
  $\log(p(y|\cdot))$ is $\mtt$-concave with
  $\mtt \geq (2 \alpha / \vareps) \max( \Ltt , 1 + \Ktt_{\vareps}/\vareps)$ and
  $\bdelta = \mtt (\Ltt_y + \alpha \Ltt/\vareps)^{-2}/2$, then there exists
  $C_{1, \vareps} \geq 0$ such that for any $h : \ \rset^{\dim} \to \rset$
  measurable with $\sup_{x \in \rset^d} \{\abs{h(x)}(1 + \norm{x}^{2})^{-1}\} \leq 1$, $n \in \nsets$,
  $\delta \in \ocintLigne{0, \bdelta}$ we have 
    \begin{equation}
    \abs{n^{-1}\sum_{k=1}^n \expe{h(X_k)} - \int_{\rset^{\dim}} h(\tilde{x}) \rmd \posteriorweps(\tilde{x}) } \leq C_{1, \vareps} (\delta^{1/2} + \Mtt_R + \exp[-R]  + (n \delta)^{-1})  (1 + \norm{x}^4)  \eqsp .
  \end{equation}

\end{proposition}

\begin{proof}
  The proof is straightforward upon combining \Cref{prop:ergo} and \Cref{prop:bias}.
\end{proof}

In particular, applying \Cref{prop:bias_control_final} to the family
$\{h_i\}_{i=1}^d$ where for any $i \in \{1, \dots, d\}$, $h_i(x) = x_i$ we get
that
    \begin{equation}
    \norm{n^{-1}\sum_{k=1}^n \expe{X_k} - \int_{\rset^{\dim}} \tilde{x} \rmd \posteriorweps(\tilde{x}) } \leq C_{1, \vareps} (\delta^{1/2} + \Mtt_R + \exp[-R]  + (n \delta)^{-1})  (1 + \norm{x}^4)  \eqsp ,
  \end{equation}
and $n^{-1}\sum_{k=1}^n X_k$ is an approximation of the MMSE given by $\int_{\rset^{\dim}} \tilde{x} \rmd \posteriorweps(\tilde{x})$.

\section{Posterior approximation}
\label{sec:posterior_approx}
We consider the following general regularity assumption.

\begin{assumption}[$\alpha$]
  \label{assum:posterior}
  There
  exist $\constanteM \geq 0$, $\upbeta >0$ and $q : \ \rset^{\dim} \to \ooint{0,+\infty}$ such that $\int_{\rset^{\dim}} q(\tilde{x}) \rmd \tilde{x} = 1$, $\|q\|_{\infty} < +\infty$ and for  almost every
  $x \in \rset^{\dim}$,
  $\int_{\rset^{\dim}} \abs{p(\tilde{x}) - p(x- \tilde{x})} q^{\min(1 -1/\alpha, 0)}(\tilde{x}) \rmd \tilde{x}  \leq \rme^{\constanteM(1
    + \norm{x}^2)} \norm{x}^{\upbeta}$.
\end{assumption}

In the case where $\alpha \geq 1$, \Cref{assum:posterior}($\alpha$) is
equivalent to the following assumption: there exist $\constanteM \geq 0$ and
$\upbeta >0$ such that for almost every $x \in \rset^{\dim}$,
$\tvnorm{\prior -(\tau_x)_{\#} \prior} \leq \rme^{\constanteM(1 + \norm{x}^2)}
\norm{x}^{\upbeta}$, where we recall that $\prior$ is the probability distribution
with density with respect to the Lebesgue measure proportional to $p$ and that
for any $\tilde{x} \in \rset^{\dim}$, $\tau_x(\tilde{x}) = \tilde{x} - x$.  Note
that since $p \in \mathrm{L}^1(\rset^{\dim})$ we have
$\lim_{x \to 0} \tvnorm{\prior - (\tau_{x})_{\#} \prior} = 0$. In
\Cref{assum:posterior}($\alpha$) for $\alpha < 1$ we assume more regularity
for $x \mapsto (\tau_x)_{\#} \prior$ in total variation in order to obtain explicit
bounds between $\posteriorweps$ and $\posteriorw$.

In the following proposition we provide easy-to-check conditions on the density
of the prior distribution $\prior$ so that \Cref{assum:posterior}($\alpha$)
holds.

\begin{proposition}
  \label{lemma:H1_check}
  Assume that there exists $U : \ \rset^{\dim} \to \rset$ such that
  for any $x \in \rset^{\dim}$,
  $p(x) = \rme^{-U(x)} / \int_{\rset^{\dim}} \rme^{-U(\tilde{x})} \rmd
  \tilde{x}$. Assume that $U$ is $\upgamma$-Hölder, \ie \ there exists
  $C_{\upgamma} > 0$ such that for any $x_1, x_2 \in \rset^{\dim}$,
  \ie
  $\norm{U(x_1) - U(x_2)} \leq C_{\upgamma} \norm{x_1 -
    x_2}^{\upgamma}$. Then \tup{\Cref{assum:posterior}($\alpha$)} is
  satisfied for $\alpha \geq 1$. In addition, assume that
  $\upgamma \leq 2$ and that there exist $c_1, \varpi > 0$ and $c_2 \in \rset$
  such that for any $x \in \rset^{\dim}$,
  $U(x) \geq c_1 \norm{x}^{\varpi} + c_2$ then
  \tup{\Cref{assum:posterior}($\alpha$)} holds for any $\alpha > 0$.
\end{proposition}

Under \Cref{assum:posterior}($\alpha$) we establish the following result which
ensures that $\posteriorweps$ is close to $\posteriorw$ in total variation for
small values of $\vareps$.

\begin{proposition}
  \label{prop:posterior_eps}
  Assume \rref{assum:post}, then the following hold:
  \begin{enumerate}[label=(\alph*), leftmargin=4em]
  \item \label{item:alpha1} If $\alpha = 1$, then $\lim_{\vareps \to 0} \tvnorm{\posteriorweps - \posteriorw} = 0$ .
  \item \label{item:alpha2} Assume that $\normLigne{p}_{\infty} < +\infty$ then
  for any $\alpha \geq 1$,
  $\lim_{\vareps \to 0} \tvnorm{\posteriorweps - \posteriorw} = 0$.
\item \label{item:alpha3} Assume that $\normLigne{p}_{\infty} < +\infty$ and
  \tup{\Cref{assum:posterior}($\alpha$)} then there exist $\vareps_1 > 0$ and
  $A_0 \geq 0$ such that for any $\vareps \in \ocint{0, \vareps_1}$ we have
  $\tvnorm{\posteriorweps - \posteriorw} \leq A_0 \vareps^{\upbeta\min(\alpha,
    1)/2}$.
  \end{enumerate}
\end{proposition}

Note that a related result in the case where
$p(x) = \rme^{-U(x)} / \int_{\rset^{\dim}} \rme^{-U(\tilde{x})} \rmd \tilde{x}$
with $U$ Lipschitz continuous and $\alpha = 1$ can be found in \cite[Corollary
1]{vono2019asymptotically} with explicit dependency with respect to the
dimension $d$. However, note that \Cref{prop:posterior_eps} differs from
\cite[Corollary 1]{vono2019asymptotically} since the Gaussian smoothing
approximation is applied to the prior distribution and the estimate is given on
the posterior distribution in \Cref{prop:posterior_eps}, whereas in
\cite[Corollary 1]{vono2019asymptotically} the Gaussian smoothing approximation
is applied to the posterior distribution and the estimate is given on the
posterior distribution as well.

The following proposition is an extension of \Cref{prop:bias_control_final} and \Cref{prop:bias_control_final_C}. The main difference
is that the approximation is expressed with respect to the true posterior
$\posteriorw$ and not $\posteriorweps$ for some value $\vareps > 0$. Let
$\vareps_1 > 0$ be given by \Cref{prop:posterior_eps}. In order to state this
proposition, we recall the following assumption which is a relaxation of the
strongly log-concave condition.

\begingroup
\def\theassumption{\ref{assum:one_sided}}
\begin{assumption}
	There exists $\mtt \in \rset$ such that for any $x_1, x_2 \in \rset^d$ we have
	\begin{equation}
	\langle \nabla \log p(y|x_2) - \nabla \log p(y|x_1) , x_2 - x_1 \rangle \leq -\mtt \norm{x_2 - x_1}^2 \eqsp . 
	\end{equation}
\end{assumption}
\addtocounter{assumption}{-1}
\endgroup

Note that the posterior is strongly log-concave if and only if $\mtt > 0$.

\begin{proposition}
  \label{prop:bias_control_final_final}
  Assume \rref{assum:post}, \tup{\Cref{assum:finite}}, \rref{assum:neural_net},
  \rref{assum:cov} and \rref{assum:one_sided}. Let $\alpha >0$ and assume that
  for any $\vareps \in \ocint{0, \min(\vareps_0, \vareps_1)}$,
  $\int_{\rset^{\dim}} (1+\norm{\tilde{x}}^4) (p_{\vareps}^{\alpha} +
  p^{\alpha})(\tilde{x}) \rmd \tilde{x} < + \infty$ and
  \tup{\Cref{assum:posterior}($\alpha$)}. Then there exists $C_0 \geq 0$ such
  that for any $\vareps > 0$ and $\lambda >0$ such that
  $2\lambda (\Ltt_y + (\alpha /\vareps) \max(\Ltt, 1 +\Ktt_{\vareps}/\vareps)  - \min(\mtt,0)) \leq 1$
  and $\bdelta = (1/3)(\Ltt_y + \alpha \Ltt /\vareps + 1 / \lambda)^{-1}$, there
  exists $C_{1, \vareps} \geq 0$ such that for any $\msc$ convex compact with
  $\cball{0}{R_{\msc}} \subset \msc$ and $R_{\msc} > 0$, there exists
  $C_{2, \vareps, \msc} \geq 0$ such that for any $h : \ \rset^{\dim} \to \rset$
  measurable with $\sup_{x \in \rset^d} \{\abs{h(x)}(1 + \norm{x}^{2})^{-1}\} \leq 1$, $n \in \nsets$,
  $\delta \in \ocintLigne{0, \bdelta}$ and $R > 0$ we have
  \begin{multline}
    \abs{n^{-1}\sum_{k=1}^n \expe{h(X_k)} - \int_{\rset^{\dim}} h(\tilde{x}) \rmd \posteriorw(\tilde{x}) } \\\leq \defEns{C_0 \vareps^{\upbeta \min(\alpha, 1) /4} + C_{1, \vareps} R_{\msc}^{-1} + C_{2, \vareps, \msc} (\delta^{1/2} + \Mtt_R + \exp[-R]  + (n \delta)^{-1})  } (1 + \norm{x}^4)  \eqsp .
  \end{multline}
  In addition, if there exists $\mtt >0$ such that $\log(p(y|\cdot))$ is $\mtt$-concave
  with $\mtt \geq 2 (\alpha /\vareps)\max(\Ltt, 1 + \Ktt_{\vareps} / \vareps)$ and
  $\bdelta = \mtt (\Ltt_y + \alpha \Ltt/\vareps)^{-2}/2$, then there exists
  $C_{1, \vareps} \geq 0$ such that for any $h : \ \rset^{\dim} \to \rset$
  measurable with $\sup_{x \in \rset^d} \{\abs{h(x)}(1 + \norm{x}^{2})^{-1}\} \leq 1$, $n \in \nsets$,
  $\delta \in \ocintLigne{0, \bdelta}$ and $R >0$ we have
    \begin{multline}
    \abs{n^{-1}\sum_{k=1}^n \expe{h(X_k)} - \int_{\rset^{\dim}} h(\tilde{x}) \rmd \posteriorw(\tilde{x}) } \\ \leq C_0 \vareps^{\upbeta \min(\alpha, 1) /4} + C_{1, \vareps} (\delta^{1/2} + \Mtt_R + \exp[-R] + (n \delta)^{-1})  (1 + \norm{x}^4)  \eqsp .
  \end{multline}
\end{proposition}

\begin{proof}
  In the general case where $\log(p(y|\cdot))$ is not assumed to be $\mtt$-concave with
  $\mtt > 0$, the proof is completed upon combining
  \Cref{prop:bias_control_final_C},
  \Cref{prop:posterior_eps} and the fact that for any probability distribution
  $\nu_1, \nu_2$,
  $\Vnorm{\nu_1 - \nu_2} \leq \tvnormsq{\nu_1 - \nu_2} (\nu_1[V^2] +
  \nu_2[V^2])^{1/2}$.  The proof is similar in the case where $\log(p(y|\cdot))$ is
  $\mtt$-concave upon replacing
  \Cref{prop:bias_control_final_C} by
  \Cref{prop:bias_control_final}.
\end{proof}

\section{Technical results}
\label{sec:technical-results}

In this section, we gather technical results which will be used throughout our
analysis. Let $b \in \rmc(\rset^d, \rset^d)$ such that for any $x \in \rset^d$,
the following Stochastic Differential Equation admits a unique strong solution
\begin{equation}
  \label{eq:langevin}
  \rmd \bfX_t = b(\bfX_t) \rmd t + \sqrt{2} \rmd \bfB_t \eqsp ,
\end{equation}
where $(\bfB_t)_{t \geq 0}$ is a $d$-dimensional Brownian motion and
$\bfX_0 = x$.  In this case, \eqref{eq:langevin} defines a Markov semi-group
$(\Pker_{t})_{t \geq 0}$ for any $x \in \rset^d$ and $\msa \in \mcb{\rset^d}$ by
$\Pker_{t}(x,\msa) = \probaLigne{\bfX_{t} \in \msa}$ where
$(\bfX_{t})_{t \geq 0}$ is the solution of \eqref{eq:langevin} with
$\bfX_{0} = x$. Consider now the generator of $(\Pker_{t})_{t \geq 0}$,
defined for any $f \in \rmc^2(\rset^d, \rset)$ by
\begin{equation}
\label{eq:def_generator}
  \generator f = \ps{\nabla f}{b(x)} + \Delta f \eqsp.
\end{equation}
We say that a Markov semi-group $(\Pker_t)_{t \geq 0}$ on $\rset^d \times \mcb{\rset^d}$ with extended infinitesimal generator $(\mathcal{A},\domain(\generator))$ (see \eg~\cite{meyn1993criteria_iii} for the definition of $(\generator,\domain(\generator))$) satisfies a continuous drift condition \hypertarget{assum:drift_continuous}{$\bfDc(W,\zeta,\beta)$} if there exist $\zeta >0$, $\beta \geq 0$ and a measurable function $W: \rset^d \to \coint{1,+\infty}$ with $W \in \domain(\generator)$ such that for all $x \in \rset^d$
\begin{equation}
  \label{eq:continuous_drift}
  \mathcal{A}W(x) \leq - \zeta W(x) + \beta \eqsp .
\end{equation}

Similarly, we consider the Markov chain $(X_k)_{k \in \nset}$ given by the
following recursion for any $k \in \nset$ and $x \in \rset^d$
\begin{equation}
  X_{k+1} = X_k + \gamma b(X_k) + \sqrt{2 \gamma} Z_k \eqsp ,
\end{equation}
with $X_0 = x$, $\gamma > 0$ and $\ensemble{Z_k}{k \in \nset}$ a family of i.i.d
Gaussian random variables with zero mean and identity covariance matrix. We
define its associated Markov kernel
$\Rker_{\gamma}: \rset^d \times \mcb{\rset^d} \to \ccint{0,1}$ as follows for any
$x \in \rset^d$ and $\msa \in \mcb{\rset^d}$

\begin{equation}
  \Rker_{\gamma}(x, \msa) = \int_{\rset^d} \1_{\msa}(x + \gamma b(x) + \sqrt{2 \gamma} z) \exp[-\norm{z}^2/2] \rmd z \eqsp .
\end{equation}
We say that $\Rker_{\gamma}$ satisfies a discrete drift condition
\hypertarget{assum:drift_discrete}{$\bfDd(W,\lambda,c)$} if there exist
$\lambda \in \coint{0, 1}$, $c \geq 0$ and a measurable function
$W: \rset^d \to \coint{1,+\infty}$ such that for all $x \in \rset^d$
\begin{equation}
  \Rker_{\gamma} W(x) \leq \lambda W(x) + c \eqsp .
\end{equation}

The following two lemmas are classical, see for instance
\cite[Lemma 18, Lemma 19]{debortoli2020maximum}. We recall these results and their proofs for the
sake of completeness.

\begin{lemma}
  \label{lemma:drift_plus}
  Assume that there exist $\Ltt, \cc \geq 0$ and $\mtt > 0$ such that for any
  $x_1, x_2 \in \rset^d$ we have
  \begin{equation}
    \label{eq:ineq_curv}
    \langle b(x_1), x_1 \rangle \leq - \mtt \norm{x_1}^2 + \cc \eqsp, \qquad \norm{b(x_1) - b(x_2)} \leq \Ltt \norm{x_1 - x_2} \eqsp .
  \end{equation}
  Let $\bgamma =  \mtt / \Ltt^2$. Then the following results hold:
  \begin{enumerate}[wide, labelwidth=!, labelindent=0pt, label=(\alph*)]
  \item For any $\varpi \in \nsets$ there exist $\lambda \in \ocint{0,1}$,
  $c, \beta \geq 0$ and $\zeta > 0$ such that for any
  $\gamma \in \ocint{0, \bgamma}$, $\Rker_{\gamma}$ satisfies
  \hyperlink{assum:drift_discrete}{$\bfDd(W,\lambda^{\gamma},c \gamma)$} and
  $(\Pker_t)_{t \geq 0}$ satisfies \hyperlink{assum:drift_continuous}{$\bfDc(W,\zeta,\beta)$} with
  $W(x) = 1 + \norm{x}^{2 \varpi}$.
\item For any $\varpi > 0$, there exist $\lambda \in \ocint{0,1}$,
  $c, \beta \geq 0$ and $\zeta > 0$ such that for any
  $\gamma \in \ocint{0, \bgamma}$, $\Rker_{\gamma}$ satisfies
  \hyperlink{assum:drift_discrete}{$\bfDd(W,\lambda^{\gamma},c \gamma)$} and
  $(\Pker_t)_{t \geq 0}$ satisfies \hyperlink{assum:drift_continuous}{$\bfDc(W,\zeta,\beta)$} with
  $W(x) = \exp[\varpi \sqrt{1+ \normLigne{x}^2}]$.
  \end{enumerate}
\end{lemma}

\begin{proof}
  We divide the proof into two parts.
  \begin{enumerate}[wide, labelwidth=!, labelindent=0pt, label=(\alph*)]
  \item Let $\varpi \in \nsets$ and $\gamma \in \ocint{0, \bgamma}$ with
    $\bgamma = \mtt / (4 \Ltt^2)$.  Let $\Tg(x) = x- \gamma b(x)$. In the
    sequel, for any $k \in \{1, \dots, \varpi \}$, $c , \tilde{c}_k \geq 0$ and
    $\lambda, \tilde{\lambda}_k \in \coint{0,1}$ are constants independent of
    $\gamma$ which may take different values at each appearance.  Let
    $\vareps \in \ooint{0,1/2}$. Using \eqref{eq:ineq_curv}, the fact that for
    any $a, b\geq 0$, $(a+b)^2 \leq (1 + \vareps)a^2 + (1 + \vareps^{-1})b^2$
    and the fact that for any $a, b \geq 0$ we have
    $(a+b)^{1/2} \leq a^{1/2} + b^{1/2}$, we get that for any $x \in \rset^d$
    with $\norm{x} \geq (2 \cc / (\vareps \mtt))^{1/2}$
    \begin{align}
      \label{eq:ineq_un_tgamma}
        \norm{\Tg(x)} &= \parenthese{\norm{x}^2 + 2\gamma \langle b(x), x \rangle + \gamma^2 \norm{b(x)}^2}^{1/2} \\
                        &\leq \parenthese{(1 - 2\gamma \mtt+ (1 + \vareps) \gamma^2 \Ltt^2) \norm{x}^2 + 2 \gamma \cc  + (1 + \vareps^{-1}) \gamma^2 \norm{b(0)}^2}^{1/2} \\
                      &\leq \parenthese{(1 - \gamma \mtt+ (1 + \vareps)\gamma^2 \Ltt^2) \norm{x}^2 +  (1 + \vareps^{-1}) \gamma^2 \norm{b(0)}^2}^{1/2} \\
        &\leq \exp[-\gamma ((2-\vareps)\mtt - (1 + \vareps)\Ltt^2\bgamma)/2] \norm{x} + (1 + \vareps^{-1/2}) \gamma \norm{b(0)} \eqsp .
    \end{align}
    Note that $(2 - \vareps)\mtt - (1 +\vareps) \Ltt^2 \bgamma < 0$ since
    $\vareps \in \ooint{0,1/2}$ and $\bgamma = \mtt / \Ltt^2$.  On the other
    hand using \eqref{eq:ineq_curv} and the fact that for any $a,b \geq 0$ with
    $a \geq b$ and $\rme^a - \rme^b \leq \rme^a (a - b)$, we have for any
    $x \in \rset^d$ with $\norm{x} \leq (2 \cc / (\vareps\mtt))^{1/2}$
      \begin{align}
        \norm{\Tg(x)} &\leq (1 + \gamma \Ltt) \norm{x} + \gamma \norm{b(0)} \\
                      &\leq \exp[-\gamma ((2-\vareps)\mtt - (1 + \vareps)\Ltt^2\bgamma)/2] \norm{x} \\
        & \qquad + (2 \cc / (\vareps\mtt))^{1/2}\defEns{\exp[\gamma \Ltt] - \exp[-\gamma ((2-\vareps)\mtt - (1 + \vareps)\Ltt^2\bgamma)/2]} + \gamma\norm{b(0)} \\
                      &\leq \exp[-\gamma ((2-\vareps)\mtt - (1 + \vareps)\Ltt^2\bgamma)/2] \norm{x} + \gamma (2 \cc / (\vareps\mtt))^{1/2} \exp[\bgamma \Ltt](\Ltt + 2\mtt) + \gamma \norm{b(0)} \eqsp .
                                \label{eq:ineq_deux_tgamma}
      \end{align}
      Combining \eqref{eq:ineq_un_tgamma} and \eqref{eq:ineq_deux_tgamma}, there
      exist $\lambda \in \coint{0,1}$ and $c \geq 0$ such that for any
      $\gamma \in \ocint{0, \bgamma}$ and $x \in \rset^d$,
      \begin{equation}
        \label{eq:ineq_fin_tgamma}
        \norm{\Tg(x)} \leq \lambda^{\gamma} \norm{x} + \gamma c \eqsp .
      \end{equation}
      Note that  using \eqref{eq:ineq_fin_tgamma}, for any $k \in \{1, \dots, 2\varpi\}$ there exist $\tilde{\lambda}_k \in \ooint{0,1}$ and $\tilde{c}_k \geq 0$ such that
  \begin{align}
    \label{eqII:1}
    \norm[k]{\Tg(x)} &\leq \{\tilde{\lambda}_k^{\gamma} \norm{x} + \gamma \tilde{c}_k\}^{k} \\ & \leq \tilde{\lambda}_k^{\gamma k} \norm{x}^{k} + \gamma 2^{k} \max(\tilde{c}_k, 1)^{k} \max(\bgamma, 1)^{k - 1} \defEnsLigne{1  + \norm{x}^{k-1}} \\ & \leq \tilde{\lambda}_k^{\gamma} \norm[k]{x} + \tilde{c}_k \gamma  \defEnsLigne{1+\norm[k-1]{x}} \leq (1 + \norm{x}^{k}) (1 + \tilde{c}_k \gamma) \eqsp. 
  \end{align}
  Therefore, combining \eqref{eqII:1} and the Cauchy-Schwarz inequality we
  obtain that for any $\gamma \in \ocint{0, \bgamma}$ and $x \in \rset^d$
  \begin{align}
    &\int_{\rset^{\dim}} (1 + \norm{y}^{2\varpi}) \Rker_{\gamma}(x , \rmd y)  = 1 + \expeLigne{(\norm{\Tg(x)}^2 + 2 \sqrt{2\gamma}  \langle \Tg(x), Z \rangle + 2 \gamma \norm{Z}^2)^{\varpi}} \\
                         & \quad  =1+ \sum_{k=0}^\varpi \sum_{\ell =0}^{k} {\varpi \choose k} {k \choose \ell}  \norm{\Tg(x)}^{2(\varpi-k)} 2^{(3k-\ell)/2}  \gamma^{(k+\ell)/2} \expeLigne{\langle \Tg(x), Z \rangle^{k -\ell}\norm{Z}^{2\ell}} \\
    &\quad \leq  1 + \norm{\Tg(x)}^{2\varpi} \\
    & \quad \quad 
      + 2^{3\varpi/2} \sum_{k=1}^\varpi \sum_{\ell =0}^{k} {\varpi \choose k} {k \choose \ell}  \norm{\Tg(x)}^{2(\varpi-k)}  \gamma^{(k+\ell)/2} \expeLigne{\langle \Tg(x), Z \rangle^{k -\ell}\norm{Z}^{2\ell}} \1_{\{(1,0)\}^{\complementary}}(k,\ell) \\
    &\quad \leq  1 + \norm{\Tg(x)}^{2\varpi} \\
    & \quad \quad 
      + \gamma 2^{3\varpi/2} \sum_{k=1}^\varpi \sum_{\ell =0}^{k} {\varpi \choose k} {k \choose \ell}  \norm{\Tg(x)}^{2\varpi - k-\ell}  \bgamma^{(k+\ell)/2 - 1} \expeLigne{\norm{Z}^{k + \ell}} \1_{\{(1,0)\}^{\complementary}}(k,\ell) \\
    &\quad \leq 1 + \tilde{\lambda}_{2\varpi}^{\gamma} \norm{x}^{2 \varpi} + \tilde{c}_{2\varpi} \gamma \defEnsLigne{1 + \norm{x}^{2 \varpi - 1}} \\ & \quad \quad + \gamma 2^{3\varpi/2}  2^{2\varpi} \max(\bgamma, 1)^{2\varpi}  \sup_{k \in \{1, \dots, \varpi\}} \defEnsLigne{(1 + \tilde{c}_k \bgamma) \expeLigne{\norm{Z}^k}}  (1 + \norm{x}^{2\varpi - 1}) \\ 
                         & \quad \leq 1 + \lambda^{\gamma} \norm{x}^{2\varpi} + \gamma c (1 + \norm{x}^{2\varpi-1}) \\ & \quad \leq \lambda^{\gamma/2} (1 +  \norm{x}^{2\varpi}) + \gamma c(1 + \norm{x}^{2 \varpi -1}) + \lambda^{\gamma} (1 +  \norm{x}^{2\varpi}) - \lambda^{\gamma/2} (1 +  \norm{x}^{2\varpi}) \eqsp .
  \end{align}
  Using that
  $\lambda^{\gamma} - \lambda^{\gamma /2} \leq -\log(1/\lambda) \gamma
  \lambda^{\gamma/2} / 2$, we get that for any $\gamma \in \ocint{0, \bgamma}$,
  $\Rker_{\gamma}$ satisfies
  \hyperlink{assum:drift_discrete}{$\bfDd(W,\lambda^{\gamma},c \gamma)$}.  We
  now show that there exist $\zeta > 0$ and $\beta \geq 0$ such that
  $(\Pker_t)_{t \geq 0}$ satisfies
  \hyperlink{assum:drift_continuous}{$\bfDc(W,\zeta,\beta)$}.  First, for any
  $x \in \rset^d$ we have
  \begin{equation}
    \nabla W(x) = 2 \varpi \norm{x}^{2(\varpi - 1)} x \eqsp , \qquad \Delta W(x) = 2 \varpi (2 \varpi - 1) \norm{x}^{2(\varpi - 1)}
  \end{equation}
  Combining this result, the Cauchy-Schwarz inequality and \eqref{eq:ineq_curv},
  we obtain that for any $x \in \rset^d$
  \begin{align}
    \generator W(x) &= \langle \nabla W(x), b(x) \rangle + \Delta W(x) \\
                    &\leq-2\mtt \varpi \norm{x}^{2 \varpi} + 2 \varpi c \norm{x}^{2 \varpi - 1} + 2 \varpi (2 \varpi - 1) \norm{x}^{2(\varpi - 1)} \\
                    &\leq-\mtt \varpi \norm{x}^{2 \varpi} + \sup_{x \in \rset^d} \defEnsLigne{2 \varpi (c + 2 \varpi - 1) \norm{x}^{2 \varpi - 1}  -\mtt \varpi \norm{x}^{2 \varpi}} \\
        &\leq-\mtt \varpi W(x) + \sup_{x \in \rset^d} \defEnsLigne{2 \varpi (c + 2 \varpi - 1) \norm{x}^{2 \varpi - 1}  -\mtt \varpi \norm{x}^{2 \varpi}} + \mtt \varpi \eqsp .
  \end{align}
  Hence letting $\zeta = \mtt \varpi$ and $\beta = \sup_{x \in \rset^d} \defEnsLigne{2 \varpi (c + 2 \varpi - 1) \norm{x}^{2 \varpi - 1}  -\mtt \varpi \norm{x}^{2 \varpi}} + \mtt \varpi $, we obtain that $(\Pker_t)_{t \geq 0}$ satisfies
  \hyperlink{assum:drift_continuous}{$\bfDc(W,\zeta,\beta)$}.
\item First, we show that for any
  $\gamma \in \ocint{0, \bgamma}$, $\Rker_{\gamma}$ satisfies
  \hyperlink{assum:drift_discrete}{$\bfDd(\Phi,\lambda^{\gamma},c)$}, where
  $\Phi(x) = \parentheseLigne{1 + \norm{x}^2}^{1/2} = W_2^{1/2}(x)$ and
  $W_2(x) = 1+ \norm{x}^2$. Using the first part of the proof, there exist
  $\lambda_0 \in \coint{0,1}$ and $c_0 \geq 0$ such that for any
  $\gamma \in \ocint{0, \bgamma}$ with $\bgamma = \mtt / (4 \Ltt^2)$ we have that
  $\Rker_{\gamma}$ satisfies
  \hyperlink{assum:drift_discrete}{$\bfDd(W_2,\lambda_0^{\gamma},c_0 \gamma)$}.
  Using Jensen's inequality we obtain that for any
  $\gamma \in \ocint{0, \bgamma}$ and $x \in \rset^d$ with $\norm{x} \geq R$
  and $R = \max(1, ((2 c_0 \lambda_0^{-\bgamma})/\log(1/\lambda_0))^{1/2})$ we
  have
  \begin{equation}
    \Rker_{\gamma} \Phi(x) \leq \parenthese{\Rker_{\gamma} W_2(x)}^{1/2} \leq \exp[(\gamma/2) \defEnsLigne{\log(\lambda_0) + \lambda_0^{-\bgamma}c_0 R^{-2}}] \Phi(x) \leq \lambda_0^{\gamma/4} \Phi(x) \eqsp .
  \end{equation}
  In addition, using that for any $a, b \geq 0$ with $a \geq b$ we have
  $\rme^a - \rme^b \leq \rme^a (b-a)$, we get for any $x \in \rset^d$ with
  $\norm{x} \leq R$
  \begin{align}
    \Rker_{\gamma} \Phi(x) \leq \parenthese{\Rker_{\gamma} W_2(x)}^{1/2} &\leq \exp[(\gamma/2) \defEnsLigne{\log(\lambda_0) + \lambda_0^{-\bgamma}c_0 }] \Phi(x) \\
    &\leq \exp[(\gamma/2) \defEnsLigne{\log(\lambda_0) + \lambda_0^{-\bgamma}c_0 R^{-2}}] \Phi(x) \\
    &\qquad + \lambda_0^{-\bgamma}c_0 (1 - R^{-2}) \exp[(\gamma/2) \defEnsLigne{\log(\lambda_0) + \lambda_0^{-\bgamma}c_0 R^{-2}}]\Phi(R) \eqsp .
  \end{align}
  Hence, there exist $\lambda_1 \in \coint{0,1}$ and $c_1 \geq 0$ such that for
  any $\gamma \in \ocint{0, \bgamma}$ we have that
  $\Rker_{\gamma}$ satisfies
  \hyperlink{assum:drift_discrete}{$\bfDd(\varpi \Phi,\lambda_1^{\gamma},c_1 \gamma)$}.
  Now let $W(x) = \exp[\Phi(x)]$. Using the logarithmic Sobolev inequality
  \cite[Theorem 5.5]{boucheron2013concentration} we get for any
  $\gamma \in \ocint{0, \bgamma}$ and $x \in \rset^d$ with $\norm{x} \geq R$
  and $R = 1 + (\varpi^2 + c_1)^{-1}\log(1/\lambda_1)$ 
  \begin{align}
    \Rker_{\gamma} W(x) \leq \exp[\Rker_{\gamma} \varpi \Phi(x) + \gamma \varpi^2 ] &\leq  \exp[-(1 - \lambda_1^{\gamma})\Phi(x)  + \gamma(\varpi^2 + c_1)] W(x) \\
    &\leq \exp[-\gamma  \log(1/\lambda_1) R  + \gamma(\varpi^2 + c_1)] W(x) \leq \lambda_1^{\gamma} W(x) \eqsp .
  \end{align}
  In addition, using that for any $a, b \geq 0$ with $a \geq b$ we have
  $\rme^a - \rme^b \leq \rme^a (b-a)$, we get for any $x \in \rset^d$ with
  $\norm{x} \leq R$
  \begin{align}
    \Rker_{\gamma} W(x) &\leq \exp[\Rker_{\gamma} \varpi \Phi(x) + \gamma] \leq  \exp[\gamma(\varpi^2 + c_1)] W(x) \\ &\leq \lambda_1^{\gamma} W(x) + \gamma \exp[\bgamma(\varpi^2 + c_1)] ((1 + c_1) + \log(1/\lambda_1)) W(R) \eqsp .
  \end{align}
  Therefore, there exist $\lambda \in \coint{0,1}$ and $c \geq 0$ such that for
  any $\gamma \in \ocint{0, \bgamma}$ we have that
  $\Rker_{\gamma}$ satisfies
  \hyperlink{assum:drift_discrete}{$\bfDd(W,\lambda^{\gamma},c \gamma)$}.
    We now show that there exist $\zeta > 0$ and $\beta \geq 0$ such that
  $(\Pker_t)_{t \geq 0}$ satisfies
  \hyperlink{assum:drift_continuous}{$\bfDc(W,\zeta,\beta)$}.
  First, for any $x \in \rset^d$ we have
  \begin{equation}
    \nabla W(x) = \varpi x\Phi^{-1}(x)W(x) \eqsp , \quad \Delta W(x) = \defEnsLigne{\varpi \Phi^{-1}(x)(1 - \norm{x}^2/\Phi^{2}(x)) + \varpi^2 \norm{x}^2/\Phi^{2}(x)} W(x) \eqsp .
  \end{equation}
  Therefore using \eqref{eq:ineq_curv} we obtain that for any $x \in \rset^d$ with
  $\norm{x} \geq \sqrt{2}(1 + (c+1 + \varpi)/\mtt)$
  \begin{equation}
    \generator W(x) \leq \varpi (-\mtt \Phi^{-1}(x)\norm{x}^2 + c + 1 + \varpi)W(x) \leq -(\mtt/2)W(x) \eqsp ,
  \end{equation}
  which concludes the proof.
  \end{enumerate}
\end{proof}

\begin{lemma}
  \label{lemma:bound_plus}
  Assume that there exist $\lambda \in \ocint{0,1}$, $c, \beta \geq 0$,
  $\zeta, \bgamma > 0$ such that for any $\gamma \in \ocint{0, \bgamma}$,
  $\Rker_{\gamma}$ satisfies
  \hyperlink{assum:drift_discrete}{$\bfDd(W,\lambda^{\gamma},c \gamma)$} and
  $(\Pker_t)_{t \geq 0}$ satisfies
  \hyperlink{assum:drift_continuous}{$\bfDc(W,\zeta,\beta)$}. Then, there exists
  $C \geq 0$ such that for any $x \in \rset^d$, $t \geq 0$ and $k \in \nsets$ we have
  \begin{equation}
    \Rker_{\gamma}^k W(x) + \Pker_tW(x) \leq C W(x) \eqsp .
  \end{equation}
\end{lemma}

\begin{proof}
  There exists $C_c \geq 0$ such that for any $x \in \rset^d$ and $t \geq 0$,
  $\Pker_tW(x) \leq C_c W(x)$ using \cite[Lemma 25-(b)]{debortoli2020convergence}.
  Using that for any $t \geq 0$, $(1 - \rme^{-t})^{-1} \leq 1 + 1/t$ we get that
  for any $\gamma \in \ocint{0, \bgamma}$, $x \in \rset^d$ and $k \in \nsets$
  \begin{equation}
    \Rker_{\gamma}^k W(x) \leq W(x) + c \gamma \sum_{k \in \nset} \lambda^{k \gamma} \leq (1 + c(\bgamma + \log(1/\lambda))) W(x) \eqsp ,
  \end{equation}
  which concludes the proof upon letting $C = C_c + 1 + c(\bgamma + \log(1/\lambda)$.
\end{proof}

\begin{proposition}
  \label{prop:stability_observation}
  Assume that there exist $\Phi_1: \ \rset^d \to \coint{0, +\infty}$ and $\Phi_2: \ \rset^\dimY \to \coint{0,+\infty}$
  such that for any $x \in \rset^d$ and $y_1, y_2 \in \rset^\dimY$
  \begin{equation}
    \label{eq:condition_reg}
    \norm{\log(q_{y_1}(x)) - \log(q_{y_2}(x))} \leq (\Phi_1(x) + \Phi_2(y_1) + \Phi_2(y_2)) \norm{y_1 - y_2}\eqsp ,
  \end{equation}
  and for any $c>0$,
  $\int_{\rset^d} (1+\Phi_1(\tilde{x})) \exp[c\Phi_1(\tilde{x})] p(x) \rmd x <
  +\infty$. Then $y \mapsto \pi_y$ is locally Lipschitz w.r.t the total
  variation $\tvnorm{\cdot}$, where for any $x \in \rset^d, y \in \rset^\dimY$
  we have
  \begin{equation}
    \label{eq:def_pi_y}
    (\rmd \pi_y / \rmd \Leb)(x)  = \left. q_y(x) p(x) \middle/ \int_{\rset^d} q_y(\tilde{x}) p(\tilde{x}) \rmd \tilde{x} \right. \eqsp .
      \end{equation}
\end{proposition}

\begin{proof}
  Let $y_1, y_2 \in \msk$ with $\msk$ a compact set. Let $y_0 \in \msk$ and
  $D_{\msk}$ be the diameter of $\msk$. Using \Cref{lemma:tv_norm} we get that
  \begin{equation}
    \tvnorm{\pi_{y_1} - \pi_{y_2}} \leq 2 c_{y_1} \int_{\rset^d} \abs{q_{y_1}(x) - q_{y_2}(x)} p(x) \rmd x \eqsp ,
  \end{equation}
  with $c_{y_1} = \int_{\rset^d} q_{y_1}(x) p(x) \rmd x$. 
  Combining this result with the fact that for any $a, b \in \rset$ we have
  $\abs{\rme^a - \rme^b} \leq \abs{a - b} \max(\rme^a, \rme^b)$ we get that
  \begin{align}
    \tvnorm{\pi_{y_1} - \pi_{y_2}} &\leq 2 c_{y_1} \int_{\rset^d} \abs{q_{y_1}(x) - q_{y_2}(x)} p(x) \rmd x\\
                                   &\leq 2 c_{y_1} \int_{\rset^d} (\Phi_1(x) + \Phi_2(y_1) + \Phi_2(y_2)) \norm{y_1 - y_2} \\
                                   & \qquad \qquad \times \exp[(2\Phi_1(x) + \Phi_2(y_1) + \Phi_2(y_0) + \Phi_2(y_2)) D_{\msk}] p(x) \rmd x \\
                                   &\leq 2 c_{y_1} (\Phi_2(y_1) + \Phi_2(y_2)) \exp[\Phi_2(y_1) + \Phi_2(y_0) + \Phi_2(y_2)] \\
    & \qquad \qquad \times \int_{\rset^d} (1 + \Phi_1(x)) \exp[2D_{\msk}\Phi_1(x)] p(x) \rmd x \times \norm{y_1 - y_2} \eqsp ,
  \end{align}
  which concludes the proof.
\end{proof}

\section{Proofs of \Cref{sec:convergence_pnpula}}
\label{sec:proofs-citecr}

We recall that the Markov chain $(X_k)_{k \in \nset}$, defined in \eqref{eq:pnpula}, is given by
\begin{align}
X_{k+1} &= X_k + \delta b_{\vareps}(X_k)  + \sqrt{2 \delta} Z_{k+1} \eqsp , \\
b_{\vareps}(x) &= \nabla \log(p(y|x)) + \alpha (D_{\vareps}(x) - x) / \vareps + (x - \Pi_{\msc}(x))/ \lambda \eqsp ,
\end{align}
where $\delta > 0$ is a stepsize, $\alpha, \vareps, \lambda >0$ are
hyperparameters of the algorithm, $\msc \subset \rset^d$ is a closed
convex set with $0 \in \msc$, $\Pi_{\msc}$ is the projection on $\msc$ and
$\ensembleLigne{Z_k}{k \in \nset}$ a family of i.i.d. Gaussian random
variables with zero mean and identity covariance matrix.

In this section, we prove the convergence of \pnpula \ and control the bias of
its invariant measure in the general framework introduced in
\Cref{sec:general-framework} (\ie \ $\alpha \neq 1$) under two different
assumptions on the posterior: either the posterior is log-concave as in
\Cref{sec:strongly-log-concave} or the posterior satisfies a more general
one-sided Lipschitz condition as in \Cref{sec:convergence_pnpula}. Note that in \Cref{sec:convergence_pnpula} the results are only
stated for $\alpha =1$. The statements of the propositions can be generalized to
$\alpha > 0$ by replacing
$2\lambda (\Ltt_y + \Ltt /\vareps - \min(\mtt, 0)) \leq 1$ and
$\bdelta = (1/3)(\Ltt_y + \Ltt /\vareps + 1 / \lambda)^{-1}$ by
$2\lambda (\Ltt_y + \alpha \Ltt /\vareps - \min(\mtt, 0)) \leq 1$ and
$\bdelta = (1/3)(\Ltt_y + \alpha \Ltt /\vareps + 1 / \lambda)^{-1}$ in \Cref{prop:ergo_C} and
$2\lambda (\Ltt_y + (a/\vareps)\max(\Ltt, 1 + \Ktt_{\vareps}/ \vareps)
-\min(\mtt,0)) \leq 1$ and
$\bdelta = (1/3)(\Ltt_y + \Ltt /\vareps + 1 / \lambda)^{-1}$ by
$2\lambda (\Ltt_y + (\alpha /\vareps)\max(\Ltt, 1 + \Ktt_{\vareps}/ \vareps)
-\min(\mtt,0)) \leq 1$ and
$\bdelta = (1/3)(\Ltt_y + \alpha \Ltt /\vareps + 1 / \lambda)^{-1}$ in \Cref{prop:bias_C} and \Cref{prop:bias_control_final_C}.

\subsection{Proof of \Cref{prop:fun_res}}
\label{sec:fun_res_proof}

Let $R > 0$. Let $X$ and $Z$ be random variables with distribution $\mu$ and
zero mean Gaussian with identity covariance matrix. Let
$X_{\vareps} = X + \vareps^{1/2} Z$. We recall that the distributions of $X$ and
$X_{\vareps}$ have density with respect to the Lebesgue measure given by $p$ and
$p_{\vareps}$ respectively. In addition, the conditional density of $X$ given
$X_{\vareps}$ is given by $g_{\vareps}$. By definition
$D_{\vareps}^{\star}(X_{\vareps}) = \mathbb{E}[X|X_{\vareps}]$ and therefore we have
\begin{align}
&\ell_{\vareps}(w^{\dagger})  = \expe{\norm{X - f_{w^{\dagger}}(X_{\vareps})}^2} \\
& = \expe{\norm{X - D_{\vareps}^{\star}(X_{\vareps})}^2}  + 2 \expe{\langle X - D_{\vareps}^{\star}(X_{\vareps}) ,  D_{\vareps}^{\star}(X_{\vareps}) - f_{w^{\dagger}}(X_{\vareps}) \rangle} +  \expe{\norm{f_{w^{\dagger}}(X_{\vareps}) - D_{\vareps}^{\star}(X_{\vareps})}^2} \\
& = \expe{\norm{X - D_{\vareps}^{\star}(X_{\vareps})}^2}  +   \expe{\norm{f_{w^{\dagger}}(X_{\vareps}) - D_{\vareps}^{\star}(X_{\vareps})}^2}  = \ell_{\vareps}^{\star}  +   \expe{\norm{f_{w^{\dagger}}(X_{\vareps}) - D_{\vareps}^{\star}(X_{\vareps})}^2} \eqsp . 
\end{align}
Combining this result, the condition that
$\ell_{\vareps}(w^{\dagger}) \leq \ell_{\vareps}^{\star} + \eta$ and the
Cauchy-Schwarz inequality we get that
\begin{equation}
\label{eq:expe_bound}
\expeLigne{\norm{f_{w^{\dagger}}(X_{\vareps}) - D_{\vareps}^{\star}(X_{\vareps})}} \leq \sqrt{\eta} \eqsp .
\end{equation}
Since $f_{w^{\dagger}}$ and $D_{\vareps}^\star$ are locally
Lipschitz, there exists $C_R \geq 0$ such that for any
$x_1, x_2 \in \cball{0}{2R}$ we have
\begin{equation}
\label{eq:loc_lip}
\abs{\norm{f_{w^{\dagger}}(x_2) - D_{\vareps}^\star(x_2)}-  \norm{f_{w^{\dagger}}(x_1) - D_{\vareps}^\star(x_1)}} \leq C_R \norm{x_2 - x_1} \eqsp .
\end{equation}
Assume that
$\sup_{\tilde{x} \in \cball{0}{R}}\normLigne{f_{w^{\dagger}}(\tilde{x}) - D_{\vareps}^\star(\tilde{x})} >
\eta^{\varpi}$ with $\varpi = (2d + 2)^{-1}$ and denote $x_R \in \cball{0}{R}$ such that we have 
$\sup_{\tilde{x} \in \cball{0}{R}} \normLigne{f_{w^{\dagger}}(\tilde{x}) - D_{\vareps}^\star(x)} = \normLigne{f_{w^{\star}}(x_R) - D_{\vareps}^\star(x_R)}$. Using \eqref{eq:loc_lip} we have 
\begin{align}
\expeLigne{\norm{f_{w^{\dagger}}(X_{\vareps}) - D_{\vareps}^{\star}(X_{\vareps})}} &\geq \int_{\cball{0}{2R} \cap \cball{x_R}{C_R^{-1} \eta^{\varpi}}} \normLigne{f_{w^{\dagger}}(\tilde{x}) - D_{\vareps}^{\star}(\tilde{x})} p_{\vareps}(\tilde{x})  \rmd \tilde{x} \\
&\geq 
(\normLigne{f_{w^{\dagger}}(x_R) - D_{\vareps}^{\star}(x_R)} - \eta^{\varpi}) \int_{\cball{0}{2R} \cap \cball{x_R}{C_R^{-1} \eta^{\varpi}}} p_{\vareps}(\tilde{x})  \rmd \tilde{x} \eqsp .
\end{align}
Combining this result and \eqref{eq:expe_bound} we obtain that
\begin{equation}
\normLigne{f_{w^{\dagger}}(x_R) - D_{\vareps}^{\star}(x_R)} \leq \eta^{1/2} \parenthese{\int_{\cball{0}{2R} \cap \cball{x_R}{C_R^{-1} \eta^{\varpi}}} p_{\vareps}(\tilde{x})  \rmd \tilde{x}}^{-1} + \eta^{\varpi} \eqsp ,
\end{equation}
Setting
$\Mtt_R = \eta^{1/2} \parentheseLigne{\int_{\cball{0}{2R} \cap \cball{x_R}{C_R^{-1} \eta^{\varpi}}} p_{\vareps}(\tilde{x})  \rmd \tilde{x}}^{-1} + \eta^{\varpi}$ concludes the
first part of the proof. 
Denote $v_d$ the volume of the unit $d$-dimensional ball. We have that
$\Leb ( \cball{x_R}{C_R^{-1} \eta^{\varpi}})  = C_R^{-d} \eta^{\varpi
	d} v_d$. Using the Fubini theorem, the Lebesgue
differentiation theorem \cite[Theorem 5.6.2]{bogachev2007measure}, the dominated convergence theorem and the fact
that for $\eta \in \ocintLigne{0, (C_R R)^{1/\varpi}}$, 
$\cball{0}{2R} \cap \cball{x_R}{C_R^{-1}
	\eta^{\varpi}} = \cball{x_R}{C_R^{-1} \eta^{\varpi}}$ we get that
\begin{align}
&\lim_{\eta \to 0} \Leb(\cball{x_R}{C_R^{-1} \eta^{\varpi}})^{-1} \int_{\rset^d} \1_{\cball{x_R}{C_R^{-1} \eta^{\varpi}} \cap \cball{x_R}{C_R^{-1} \eta^{\varpi}}}(x) p_{\vareps}(x) \rmd x \\
& \qquad = \lim_{\eta \to 0} \int_{\rset^d} \vert \cball{x_R}{C_R^{-1} \eta^{\varpi}} \vert^{-1}  (2\uppi\vareps)^{-d/2}\int_{\rset^d}  \1_{\cball{x_R}{C_R^{-1} \eta^{\varpi}}}(x) \exp[-\normLigne{x - \tilde{x}}^2/ (2 \vareps)]p(\tilde{x}) \rmd x \rmd \tilde{x}  \\
& \qquad =  \int_{\rset^d}  (2\uppi\vareps)^{-d/2} \exp[-\normLigne{x_R - \tilde{x}}^2/ (2 \vareps)]p(\tilde{x}) \rmd x \rmd \tilde{x} = p_{\vareps}(x_R) > 0 \eqsp .
\end{align}
Using this result we have, 
\begin{align}
\limsup_{\eta \to 0} \eta^{-\varpi} \Mtt_R &= 1 + \limsup_{\eta \to 0} \eta^{1/2 - \varpi(d+1)} \eta^{\varpi d} \parenthese{\int_{\cball{0}{2R} \cap \cball{x_R}{C_R^{-1} \eta^{\varpi}}} p_{\vareps}(\tilde{x})  \rmd \tilde{x}}^{-1} \\
&= 1 + C_R^d v_d p_{\vareps}^{-1}(x_R) <+\infty \eqsp ,
\end{align}
which concludes the proof.

\subsection{Proof of \Cref{prop:ergo_C} and \Cref{prop:ergo}}
\label{prop:ergo:proof}
We divide this section into two parts.  First, we prove the general case where
$\log(p(y|\cdot))$ is not assumed to be strongly concave but only satisfying a
one-sided Lipschitz condition, \ie \ \Cref{prop:ergo_C}. Then we turn to the
proof of \Cref{prop:ergo}.
\begin{enumerate}[wide, labelwidth=!, labelindent=0pt, label=(\alph*)]    
	\item 
	Let $\lambda >0$ such that
	$2\lambda (\Ltt_y + \alpha \Ltt
	/\vareps) \leq 1$ and
	$\bdelta = (1/3)(\Ltt_y + \alpha \Ltt
	/\vareps + 1/\lambda)^{-1}$. Let $\msc$ be a compact convex set
	with $0 \in \msc$. Using \rref{assum:neural_net}, \eqref{eq:pnpula}
	and that $\Id -\Pi_{\msc}$ is non-expansive we have for any
	$x_1,x_2 \in \rset^{\dim}$
	\begin{align}
	\norm{b_{\vareps}(x_1) - b_{\vareps}(x_2)} \leq (\Ltt_y + \alpha \Ltt/\vareps + 1 / \lambda) \norm{x_1-x_2} \eqsp .
	\end{align}
	Denote
	$R_{\msc} = \sup \ensembleLigne{\normLigne{x_1-x_2}}{x_1, x_2 \in \msc}$.
	Using \eqref{eq:pnpula}, the Cauchy-Schwarz inequality and that
	$ 2\lambda (\alpha\Ltt/\vareps - \mtt) \leq 1$ we have for any
	$x_1, x_2 \in \rset^{\dim}$
	\begin{align}
	\langle b_{\vareps}(x_1) - b_{\vareps}(x_2), x_1 - x_2 \rangle &\leq (-\mtt  + \alpha \Ltt/\vareps )\norm{x_1-x_2}^2 - \norm{x_1-x_2}^2/\lambda + R_{\msc}\norm{x_1-x_2}/\lambda \\
	&\leq - \norm{x_1-x_2}^2/(2\lambda) + R_{\msc}\norm{x_1-x_2}/\lambda \eqsp .
	\end{align}
	Hence, for any $x_1, x_2 \in \rset^{\dim}$ with
	$\normLigne{x_1 - x_2} \geq 4 R_{\msc}$ we obtain that
	$\langle b_{\vareps}(x_1) - b_{\vareps}(x_2), x_1 - x_2 \rangle \leq
	- \normLigne{x_1 - x_2}^2/(4\lambda)$. We also have that for any $x \in \rset^d$
	\begin{equation}
	\langle b_{\vareps}(x), x \rangle \leq -  \norm{x}^2/ (4 \lambda) + \sup_{\tilde{x} \in \rset^{d}} \defEns{(R_{\msc}/\lambda + \norm{b(0)})\norm{\tilde{x}} - \norm{\tilde{x}}^2/(4 \lambda)} \eqsp .
	\end{equation}
	We conclude the proof of \Cref{prop:ergo_C} upon using \Cref{lemma:drift_plus},
	\Cref{lemma:bound_plus}, \cite[Corollary 2]{debortoli2020convergence} with
	$\bgamma \leftarrow (4\lambda)^{-1} (\Ltt_y + \alpha \Ltt/\vareps + 1 /
	\lambda)^{-2} \geq \bdelta$ and the fact that for any probability distribution
	$\nu_1, \nu_2$,
	\begin{equation}
	\label{eq:cauchy_tv}
	\Vnorm{\nu_1 - \nu_2} \leq \tvnormsq{\nu_1 - \nu_2} (\nu_1[V^2] + \nu_2[V^2])^{1/2} \eqsp .
	\end{equation}
	\item 
	Using that $\log(p(y|\cdot))$ is
	$\mtt$-concave with $2 \alpha \Ltt /(\mtt \vareps) \leq 1$, we obtain that for
	any $x_1, x_2 \in \rset^{\dim}$
	\begin{align}
	\langle b_{\vareps}(x_1) - b_{\vareps}(x_2) , x_1 - x_2 \rangle &\leq - \mtt \norm{x_1 - x_2}^2 / 2 \eqsp , \\
	\norm{b_{\vareps}(x_1) - b_{\vareps}(x_2)} &\leq (\Ltt_y + \alpha\Ltt/\vareps) \norm{x_1 - x_2} \eqsp .    
	\end{align}
	This concludes the proof of \Cref{prop:ergo} upon using \cite[Corollary
	2]{debortoli2020convergence} with
	$\bgamma \leftarrow \mtt (\Ltt_y +
	\alpha \Ltt/\vareps)^{-2} \geq \bdelta$ and \eqref{eq:cauchy_tv}.
\end{enumerate}

\subsection{Proof of \Cref{prop:bias_C} and \Cref{prop:bias}}
\label{prop:bias:proof}

Before proving \Cref{prop:bias_C} and
\Cref{prop:bias}, we show the following lemma which is a straightforward
consequence of Girsanov's theorem \cite[Theorem 7.7]{lipster2001statistics}. A
similar version of this lemma can be found in the proof of \cite[Proposition
2]{durmus2017nonasymp}.
\begin{lemma}
	\label{lemma:girsanov}
	Let $T > 0$,
	$b_1, b_2: \ \coint{0, +\infty} \times \rset^{\dim} \to \rset^{\dim}$
	measurable such that for any $i \in \{1, 2\}$ and $x \in \rset^{\dim}$,
	$\rmd \bfX_t^{(i)} = b_i(t, \bfX_t^{(i)}) \rmd t + \sqrt{2} \rmd \bfB_t$
	admits a unique strong solution with $\bfX_0^{(i)} = x$ with Markov semigroup
	$(\Pker_t^{(i)})_{t \geq 0}$ and where $(\bfB_t)_{t \geq 0}$ is a
	$d$-dimensional Brownian motion. In addition, assume that for any
	$x \in \rset^{\dim}$ and
	$\probaLigne{\int_0^T \defEnsLigne{ \normLigne{b_i(t, \bfX_t^{(i)})}^2 +
			\normLigne{b_i(t, \bfB_t)}^2 } \rmd t < + \infty} = 1$.  Let
	$V: \ \rset^{\dim} \to \coint{0,+\infty}$ measurable, then for any
	$x \in \rset^{\dim}$ we have
	\begin{multline}
	\Vnorm{\updelta_x \Pker_T^{(1)} - \updelta_x \Pker_T^{(2)}} \\ \leq  \parenthese{\updelta_x \Pker_t^{(1)}[V^2] + \updelta_x \Pker_t^{(2)}[V^2]}^{1/2} \parenthese{\int_0^T \expe{\normLigne{b_1(t, \bfX_t^{(1)}) - b_2(t, \bfX_t^{(1)})}^2} \rmd t}^{1/2} \eqsp .
	\end{multline}
\end{lemma}

\begin{proof}
	Let $T > 0$ and $x \in \rset^{\dim}$. For any $i \in \{1, 2\}$,
	denote $\mu_{(i)}^x$ the distribution of
	$(\bfX_t^{(i)})_{t \in \ccint{0,T}}$ on the Wiener space
	$(\rmc(\ccint{0,T}, \rset), \mcb{\rmc(\ccint{0,T}, \rset)})$ with
	$\bfX_0^{(i)} = x$. Similarly denote $\mu_B^x$ the distribution of
	$(\bfB_t)_{t \in \ccint{0,T}}$ witgh $\bfB_0 = x$. Using the
	generalized Pinsker inequality \cite[Lemma 24]{durmus2017nonasymp}
	and the transfer theorem \cite[Theorem 4.1]{kullback1997information}
	we get that
	\begin{equation}
	\Vnorm{\updelta_x \Pker_T^{(1)} - \updelta_x \Pker_T^{(2)}} \leq \sqrt{2} \parenthese{\updelta_x \Pker_t^{(1)}[V^2] + \updelta_x \Pker_t^{(2)}[V^2]}^{1/2} \KLsqrt(\mu_{(1)}|\mu_{(2)}) \eqsp .
	\end{equation}
	Since for any $i \in \{1, 2\}$ we have
	$\probaLigne{\int_0^T \defEnsLigne{ \normLigne{b_i(\bfX_t^{(i)})}^2 +
			\normLigne{b_i(\bfB_t)}^2 } \rmd t < + \infty} = 1$, we can apply
	Girsanov's theorem \cite[Theorem 7.7]{lipster2001statistics} and
	$\mu_B$-almost surely for any $w \in \rmc(\ccint{0,T}, \rset)$ we get 
	\begin{align}
	(\rmd \mu_{(1)}^x / \rmd \mu_B^x)((w_t)_{t \in \ccint{0,T}}) &= \exp \parentheseDeux{(1/2) \int_0^T \langle b_1(w_t), \rmd w_t \rangle - (1/4) \int_0^T  \norm{b_1(w_t)}^2 \rmd t} \eqsp , \\
	(\rmd \mu_B^x / \rmd \mu_{(2)}^x)((w_t)_{t \in \ccint{0,T}}) &= \exp \parentheseDeux{-(1/2) \int_0^T \langle b_2(w_t), \rmd w_t \rangle + (1/4) \int_0^T  \norm{b_2(w_t)}^2 \rmd t} \eqsp .
	\end{align}
	Hence, we obtain that
	\begin{equation}
	\KL(\mu_{(1)}^x | \mu_{(2)}^x) = \expe{\log((\rmd \mu_{(1)}^x / \rmd \mu_{(2)}^x)(\bfX_t^{(1)}))} = (1/4) \int_0^T \expe{\norm{b_1(\bfX_t^{(1)}) - b_2(\bfX_t^{(2)})}^2} \rmd t \eqsp ,
	\end{equation}
	which concludes the proof.
\end{proof}

In the following lemma, we show that under \Cref{assum:cov},
$\nabla \log(p_{\vareps})$ is Lipschitz continuous.

\begin{lemma}
	\label{lemma:lip_peps}
	Assume \rref{assum:cov}. Then for any $x_1, x_2 \in \rset^d$ we have
	\begin{equation}
	\norm{\nabla \log(p_{\vareps}(x_1)) - \nabla \log(p_{\vareps}(x_2))} \leq (1 + \Ktt_{\vareps}/\vareps) \norm{x_1 - x_2} / \vareps \eqsp .
	\end{equation}
	Reciprocally, if there $x \mapsto \nabla \log(p_{\vareps}(x))$ is
	Lipschitz-continuous then \rref{assum:cov}.
\end{lemma}

\begin{proof}
	Let $\vareps > 0$. We recall that for any $x \in \rset^d$ we have
	\begin{equation}
	p_{\vareps}(x) = \int_{\rset^d} \exp[-\norm{x - \tilde{x}}^2/(2 \vareps)] p(\tilde{x}) \rmd \tilde{x} \eqsp .
	\end{equation}
	Using the dominated convergence theorem we obtain that $\log(p_{\vareps}) \in \rmc^{\infty}(\rset^d, \rset)$.
	In particular we have for any $x \in \rset^d$
	\begin{align}
	\label{eq:formula_nabla2}
	\nabla^2 \log(p_{\vareps}(x)) &= - \vareps^{-1} \Id  + \vareps^{-2}\int_{\rset^{\dim}} (x - \tilde{x})^{\otimes 2} g_{\vareps}(\tilde{x}| x) \rmd \tilde{x}  - \vareps^{-2}\parenthese{\int_{\rset^{\dim}} (x - \tilde{x}) g_{\vareps}(\tilde{x} |x) \rmd \tilde{x} }^{\otimes 2}  \\
	&= -\vareps^{-1} \Id + \vareps^{-2} \int_{\rset^{\dim}} \parenthese{\tilde{x} - \int_{\rset^d} \tilde{x}' g_{\vareps}(\tilde{x}' | x) \rmd \tilde{x}'}^{\otimes 2} g_{\vareps}(\tilde{x} | x) \rmd \tilde{x}
	\end{align}
	Therefore, using \Cref{assum:cov} we obtain that for any $x \in \rset^d$ we have
	\begin{equation}
	\normLigne{\nabla^2 \log(p_{\vareps}(x))}_2 \leq \vareps^{-1} + \vareps^{-2} \Ktt_{\vareps} \eqsp ,
	\end{equation}
	which concludes the first part of the proof.  Reciprocally, since
	$x \mapsto \nabla \log(p_{\vareps}(x))$ is Lipschitz-continuous with constant
	$\Ktt \geq 0$ we get that for any basis vector
	$(\rme_i)_{i \in \{1, \dots, d\}}$ we have that
	$\rme_i^\top \nabla^2 \log(p_{\vareps}(x)) \rme_i \leq \Ktt$. Combining this
	result with \eqref{eq:formula_nabla2}, we get that
	\begin{equation}
	\vareps^{-2} \int_{\rset^{\dim}} \norm{\tilde{x} - \int_{\rset^d} \tilde{x}' g_{\vareps}(\tilde{x}' | x) \rmd \tilde{x}'}^{2} g_{\vareps}(\tilde{x} | x) \rmd \tilde{x} \leq \Ktt d  + \vareps^{-1} d \eqsp ,
	\end{equation}
	which concludes the proof.
\end{proof}

In what follows we prove \Cref{prop:bias_C}. The proof
of \Cref{prop:bias} is similar and left to the reader.

\begin{proof}[Proof of \Cref{prop:bias_C}]
	Let $\lambda >0$ such that
	$2\lambda (\Ltt_y + \alpha \Ltt
	/\vareps - \mtt) \leq 1$ and
	$\bdelta = (1/3)(\Ltt_y + \alpha \Ltt
	/\vareps + 1 / \lambda)^{-1}$. We divide the proof into two parts. First, we show that for any
	$\msc$ convex compact with $0 \in \msc$ there exists $B_{1, \msc} \geq 0$ such
	that for any $\delta \in \ocintLigne{0, \bdelta}$ and $R > 0$
	\begin{equation}
	\Vnorm{\pi_{\vareps, \delta}- \tilde{\pi}_{\vareps}} \leq B_{1, \msc} (\delta^{1/2} + \Mtt_R + \exp[-R])  \eqsp ,    
	\end{equation}
	with $\tilde{\pi}_{\vareps}$ given by
	\begin{equation}
	(\rmd \tilde{\pi}_{\vareps} / \rmd \Leb)(x) \propto \exp\parentheseDeuxLigne{- d^2(x, \msc)/(2 \lambda)}p(y|x)p_{\vareps}^{\alpha}(x) \eqsp ,
	\end{equation}
	Second, we show that there exists $B_0 \geq 0$ such that for any $\msc$ convex
	compact with $0 \in \msc$
	\begin{equation}
	\Vnorm{\pi_{\vareps}- \tilde{\pi}_{\vareps}} \leq B_0 \diameter^{-1/4}(\msc) \eqsp ,
	\end{equation}
	which concludes the proof upon using the triangle inequality.
	
	\begin{enumerate}[wide, labelwidth=!, labelindent=0pt, label=(\alph*)]    
		\item Let $\msc$
		convex compact with $0 \in \msc$. %
		We introduce
		$(\bbfX_t)_{t \geq 0}$ solution of the following Stochastic
		Differential Equation (SDE): $\bbfX_0 = X_0$ and
		\begin{align}
		\label{eq:pnpula_sde}
		\rmd \bbfX_t &= \bar{b}_{\vareps}(\bbfX_t) \rmd t + \sqrt{2}\rmd \bfB_t \eqsp , \\ 
		\bar{b}_{\vareps}(x) &= \nabla \log(p(y|x)) + \alpha \nabla \log(p_{\vareps}(x)) + \prox_{\lambda}(\iota_{\msc})(x) \eqsp ,
		\end{align}
		with $(\bfB_t)_{t \geq 0}$ a $d$-dimensional Brownian motion. 
		$\bar{b}_{\vareps}$ is Lipschitz continuous using \Cref{lemma:lip_peps}, hence 
		this SDE admits a unique strong solution for any initial condition $\bfX_0$
		with $\expeLigne{\normLigne{\bfX_0}^2} < + \infty$, see \cite[Chapter 5,
		Theorem 2.9]{karatzas1991brownian}. We denote by
		$(\Pker_{t, \vareps})_{t \geq 0}$ the semigroup associated with the
		strong solutions of \eqref{eq:pnpula_sde}. %
		Similarly to the proof of \Cref{prop:ergo}, replacing
		\cite[Corollary 2]{debortoli2020convergence} by \cite[Corollary
		22]{debortoli2020convergence}, there exist $\tilde{A}_{\msc} \geq 0$
		and $\tilde{\rho}_{\msc} \in \coint{0,1}$ such that that for any
		$x_1, x_2 \in \rset^{\dim}$ and $t \geq 0$
		\begin{align}
		\label{eq:ergo_cont}
		\Vnorm{\updelta_{x_1} \Pker_{t, \vareps}-  \updelta_{x_2} \Pker_{t, \vareps}} &\leq \tilde{A}_{\msc} \tilde{\rho}_{\msc}^t (V^2(x_1) + V^2(x_2))  \eqsp , \\
		\wassersteinD[1](\updelta_{x_1} \Pker_{t, \vareps}, \updelta_{x_2} \Pker_{t, \vareps}) &\leq \tilde{A}_{\msc} \tilde{\rho}_{\msc}^t\norm{x_1-x_2} \eqsp .
		\end{align}
		Combining \eqref{eq:ergo_cont}, \Cref{prop:ergo}, the fact that
		$(\Pens_1(\rset^{\dim}), \wassersteinD[1])$ is a complete metric
		space and the Picard fixed point theorem we obtain that for any
		$\delta \in \ocintLigne{0, \bdelta}$ there exist
		$\pi_{\vareps, \delta}, \tilde{\pi}_{\vareps} \in \Pens_1(\rset^d)$ such that
		$\pi_{\vareps, \delta} \Rker_{\vareps, \delta, \msc} = \pi_{\vareps, \delta}$ and
		for any $t \geq 0$, $\tilde{\pi}_{\vareps} \Pker_{t, \vareps} = \tilde{\pi}_{\vareps}$. Note
		that by \cite[Theorem 2.1]{roberts1996exponential} we have for any
		$x \in \rset^{\dim}$
		\begin{equation}
		(\rmd \tilde{\pi}_{\vareps} / \rmd \Leb)(x) \propto \exp\parentheseDeuxLigne{- d^2(x, \msc)/(2 \lambda)}p(y|x)p_{\vareps}^{\alpha}(x) \eqsp ,
		\end{equation}
		since $\prox_{\lambda}(\iota_{\msc})= \nabla d^2(\cdot, \msc)/(2 \lambda)$. 
		Let $f: \ \rset^{\dim} \to \rset$ measurable and such that for any
		$x \in \rset^{\dim}$, $\abs{f(x)} \leq V(x)$. Let $m \in \nsets$ such
		that $m \geq \bdelta^{-1}$, $x \in \rset^{\dim}$ and $k \in \nset$
		we have
		\begin{equation}
		\label{eq:decompo}
		\norm{\updelta_x \Rker_{\vareps, 1/m}^{k m}[f] - \updelta_x \Pker_{km, \vareps}^{k m}[f]} = \norm{\sum_{j=0}^{k-1} \updelta_x \Rker_{\vareps, 1/m}^{j m} (\Rker_{\vareps, 1/m}^{m} - \Pker_{1, \vareps}) \Pker_{k-j-1, \vareps}[f] }
		\end{equation}
		Using \eqref{eq:ergo_cont}, \Cref{lemma:drift_plus} and
		\Cref{lemma:bound_plus} there exists $B_a \geq 0$ such that for any
		$x \in \rset^{\dim}$ and $k \in \nset$ we have
		\begin{equation}
		\label{eq:inter}
		\norm{ \updelta_{x} \Pker_{k, \vareps, \msc}[f] - \tilde{\pi}_{\vareps}[f]}\leq B_a \tilde{\rho}_{\msc}^k V^2(x) \eqsp .
		\end{equation}  
		Let $T = 1$,
		$b_{1}(t, (w_t)_{t \in \ccint{0,T}}) = \sum_{j=0}^{m-1}
		\1_{\coint{j/m, (j+1)/m}}(t) b_{\vareps}(w_{j \delta})$ and
		$b_{2}(t, (w_t)_{t \in \ccint{0,T}} = \bar{b}_{\vareps}(w_t)$.  Let
		$\bfX_t^{(1)}$ and $\bfX_t^{(2)}$ the unique strong solution of
		$\rmd \bfX_t = b(t, (\bfX_t)_{t \in \ccint{0,1}}) + \sqrt{2} \bfB_t$ with $\bfX_0 =x$ with
		$x \in \rset^{\dim}$ and $b = b_1$, respectively $b= b_2$.
		Note that $(\bfX_t^{(2)})_{t \geq 0} = (\bar{\bfX}_t)_{t \geq 0}$
		and $(\bfX_{k/m}^{(1)}) = (X_k)_{k \in \nset}$. For any
		$i \in \{1, 2\}$, denote $\Pker_t^{(i)}$ the Markov semigroup
		associated with $\bfX_t^{(i)}$. For any $x \in \rset^{\dim}$ we have
		\begin{equation}
		\label{eq:ineq_duo}
		\tvnorm{\updelta_x \Rker_{\vareps, 1/m, \msc}^m - \updelta_x \Pker_{1, \vareps, \msc}} = \tvnorm{\updelta_x \Pker_{1}^{(1)} - \updelta_x \Pker_{1}^{(2)}} \eqsp .
		\end{equation}
		Using \rref{assum:neural_net}($R$) and the fact that for any $a, b\geq 0$,
		$(a+b)^2 \leq 2(a^2 + b^2)$, we have for any 
		$t \in \coint{j/m, (j+1)/m}$, $j \in \{0, \dots, m-1\}$ and
		$(w_t)_{t \in \ccint{0,1}} \in \rmc(\ccint{0,1}, \rset^d)$
		\begin{align}
		&\norm{b_1(t, (w_t)_{t \in \ccint{0,1}}) - b_2(t, (w_t)_{t \in \ccint{0,1}}) }^2 = \norm{b_{\vareps}(w_{j/m}) - \bar{b}_{\vareps}(w_t)}^2 \\
		& \qquad \qquad \leq 2 \norm{b_{\vareps}(w_{j/m}) - b_{\vareps}(w_{t})}^2 + 2 \norm{\bar{b}_{\vareps}
			(w_{t}) - b_{\vareps}(w_{t})}^2 \\
		& \qquad \qquad \leq 2 \Ltt_{b}^2 \norm{w_{j/m} - w_t}^2 + 4 \alpha^2 \Mtt_R^2 / \vareps^2 + 4 \alpha^2 \1_{\cball{0}{R}^{\complementary}}(\norm{w_t}) / \vareps^2  \eqsp ,
		\label{eq:ineq_uno}
		\end{align}
		where $\Ltt_{b}$ is the Lipschitz constant associated with $b_{\vareps}$.
		In addition using Itô's isometry we have for any $t \in \coint{j/m, (j+1)/m}$
		\begin{equation}
		\label{eq:ineq_trio}
		\textstyle{\expeLigne{\normLigne{\bfX_t^{(1)} - \bfX_{j/m}^{(1)}}^2} = 2\expeLigne{\normLigne{\int_{j/m}^t \rmd \bfB_t}^2} \leq 2 \dim \delta \eqsp .}
		\end{equation}
		Finally, using \Cref{lemma:drift_plus}, \Cref{lemma:bound_plus}, the
		logarithmic Sobolev inequality \cite[Theorem 5.5]{boucheron2013concentration},
		the Cauchy-Schwarz inequality and the Markov inequality, there exists $\tilde{B}_b \geq 0$
		such that for any $t \geq 0$ and $x \in \rset^d$
		\begin{align}
		\probaLigne{\normLigne{\bfX_t^{(1)}} \geq R} &\leq \exp[-2 R] \expe{\exp[2 \normLigne{\bfX_t^{(1)}}} \\
		& \textstyle{\leq \exp[-2 R] \expesq{\exp\parentheseDeux{4 \sqrt{2}\normLigne{\int_{\ell_t/m}^t \rmd \bfB_t }}} \expesq{\exp[4 \normLigne{X_{\ell_t}}}} \\
		&\leq \tilde{B}_b \exp[-2R] \exp[2 \Phi(x)] \eqsp ,
		\end{align}
		where $\ell_t = \floor{tm}$ and $\Phi(x) = \sqrt{1 +
			\normLigne{x}^2}$. Combining this result, \eqref{eq:ineq_uno},
		\eqref{eq:ineq_duo}, \eqref{eq:ineq_trio} and \Cref{lemma:girsanov}, we obtain
		that there exists $B_b \geq 0$ such that for any $x \in \rset^{\dim}$ and
		$R > 0$
		\begin{align}
		\Vnorm{\updelta_x \Rker_{1/m, \msc}^m - \updelta_x \Pker_{1, \msc}} &\leq 2 B_b (\sqrt{\delta} + \Mtt_R+ \exp[-R])(1 + \norm{x}^{4})\exp[\Phi(x)] \\
		& \leq 48 B_b (\sqrt{\delta} + \Mtt_R+ \exp[-R])\exp[2 \Phi(x)] \eqsp , 
		\end{align}
		Combining this result and \eqref{eq:inter} we obtain that for any
		$k \in \nset$, $j \in \{0, \dots, k-1\}$, $x \in \rset^d$ and $R > 0$ we have
		\begin{equation}
		\abs{(\updelta_x \Rker_{1/m, \msc}^{m} - \updelta_x \Pker_{1, \msc}) \Pker_{k-j-1, \msc}[f]} \leq B_aB_b (\sqrt{\delta} + \Mtt_R + \exp[-R]) \tilde{\rho}_{\msc}^{k-j-1} \exp[2\Phi(x)]  \eqsp .
		\end{equation}
		Using this result, \Cref{lemma:drift_plus}, \Cref{lemma:bound_plus} and
		\eqref{eq:decompo} we obtain that there exists $B_c \geq 0$ such that for any
		$m \in \nsets$ with $m^{-1} \geq \bdelta$
		\begin{equation}
		\Vnorm{\pi_{\vareps, 1/m, \msc} - \tilde{\pi}_{\vareps}} \leq \limsup_{k \to +\infty} \Vnorm{\updelta_0 \Rker_{\vareps, 1/m, \msc}^{k m} - \updelta_0 \Pker_{km, \vareps, \msc}^{k m}}  \leq B_c (\sqrt{\delta} + \Mtt_R + \exp[-R]) \eqsp .
		\end{equation}
		The proof in the general case where
		$\delta \in \ocintLigne{0, \bdelta}$ is similar and  we obtain that there exists $B_c \geq 0$ such that for
		any $\delta \in \ocintLigne{0, \bdelta}$
		\begin{equation}
		\Vnorm{\pi_{\vareps, \delta} - \tilde{\pi}_{\vareps}} \leq B_c (\sqrt{\delta} + \Mtt_R + \exp[-R]) \eqsp .
		\end{equation}
		\item For any
		$\msc$ compact convex with $0 \in \msc$ we define $\tilde{\pi}_{\vareps}$ and
		$\rho_{\vareps, \msc}$ such that for any $x \in \rset^{\dim}$
		\begin{equation}
		\rho_{\vareps, \msc}(x) =  \exp\parentheseDeuxLigne{ - d^2(x, \msc)/(2 \lambda)}p(y|x)p_{\vareps}^{\alpha}(x) \eqsp , \qquad (\rmd \tilde{\pi}_{\vareps} / \rmd \Leb)(x) = \left. \rho_{\vareps, \msc}(x) \middle/ \int_{\rset^{\dim}} \rho_{\vareps, \msc}(\tilde{x}) \rmd \tilde{x} \right. \eqsp .
		\end{equation}
		Similarly, define $\rho_{\vareps}$ and $\pi_{\vareps}$ such that for any
		$x \in \rset^{\dim}$
		\begin{equation}
		\rho_{\vareps}(x) = p(y|x) p_{\vareps}^{\alpha}(x) \eqsp , \qquad
		(\rmd \pi_{\vareps} / \rmd \Leb)(x) = \left. \rho_{\vareps}(x) \middle/ \int_{\rset^{\dim}} \rho_{\vareps}(\tilde{x}) \rmd \tilde{x} \right.   \eqsp .
		\end{equation}
		Since for any $x \in \rset^{\dim}$,
		$\rho_{\vareps, \msc}(x) \leq \rho_{\vareps}(x)$ we get
		$\int_{\rset^{\dim}} \rho_{\vareps, \msc}(\tilde{x}) \rmd \tilde{x} \leq
		\int_{\rset^{\dim}} \rho_{\vareps}(\tilde{x}) \rmd \tilde{x}$. Hence we obtain
		using the Cauchy-Schwarz inequality and the Markov inequality
		\begin{align}
		\KL(\pi_{\vareps} | \pi_{\msc}) &\leq \int_{\rset^{\dim}} \log(\rho_{\vareps}(\tilde{x}) / \rho_{\vareps, \msc}(\tilde{x})) \rmd \pi_{\vareps}(\tilde{x}) \\ &  \leq \int_{\msc^{\complementary}} \norm{\tilde{x}}^2 \rmd \pi_{\vareps}(\tilde{x}) \leq \probasq{X \notin \msc} \expeLignesq{\normLigne{X}^4} \leq \expeLigne{\normLigne{X}^4} R_{\msc}^{-2} \eqsp .
		\end{align}
		with $X$ a random variable with distribution $\pi_{\vareps}$. We
		conclude using the generalized Pinsker inequality \cite[Lemma
		24]{durmus2017nonasymp}.
	\end{enumerate}
\end{proof}

\section{Proofs of \Cref{sec:proj-altern}}
\label{sec:proofs-citecr-altern}

\subsection{Proof of \Cref{prop:ergo_C_proj}}
\label{prop:ergo_C_proj:proof}

Let $\alpha, \lambda, \vareps, \bdelta >0$, $\delta \in \ocintLigne{0, \bdelta}$
and $\msc \subset \rset^d$ convex and compact with $0 \in \msc$. For any
$x_1, x_2 \in \rset^d$ we have
\begin{equation}
  \norm{b_{\vareps}(x_1) - b_{\vareps}(x_2)} \leq (\Ltt_y + \alpha \Ltt / \vareps) \norm{x_1 - x_2} \eqsp .
\end{equation}
Denote $(X_n, Y_n)_{n \in \nset}$ the Markov chain obtained using the coupling
described in \cite[Section 3]{debortoli2020convergence} with initial condition
$(x_1, x_2) \in \msc$. Using \cite[Corollary 7-(b)]{debortoli2020convergence} we get
that for any $\ell \in \nset$
\begin{equation}
  \label{eq:contraction}
  \expe{\1_{\diag^\complementary}(X_{(\ell+1) \itgamma}, Y_{(\ell+1) \itgamma})} \leq (1 - \beta) \expe{\1_{\diag^\complementary}(X_{\ell \itgamma}, Y_{\ell \itgamma})} \eqsp ,
\end{equation}
where $\diag = \ensembleLigne{(x,x)}{x \in \rset^d}$ and $\beta \in \ooint{0,1}$ with
\begin{equation}
  \beta = 2\Phibf\defEnsLigne{-(1 + \bar{\delta})(1 + \Ltt_y + (\alpha \Ltt / \vareps))\diam(\msc)} \eqsp ,
\end{equation}
where $\Phibf$ is the cumulative distribution function of the univariate
Gaussian distribution with zero mean and unit variance. In addition, using that
the coupling is absorbing, we have that for any $k \in \nset$,
\begin{equation}
  \expe{\1_{\diag^\complementary}(X_{k}, Y_{k})} \leq \expe{\1_{\diag^\complementary}(X_{\floor{k/ \itgamma} \itgamma}, Y_{\floor{k/ \itgamma} \itgamma})} \eqsp ,
\end{equation}
Combining this result and \eqref{eq:contraction}, we get that for any $k \in \nset$
\begin{equation}
  \tvnorm{\updelta_{x_1} \Qker_{\vareps, \delta}^k - \updelta_{x_2} \Qker_{\vareps, \delta}^k} \leq 
  \expe{\1_{\diag^\complementary}(X_{k}, Y_{k})} \leq (1 - \beta)^{\floor{k/ \itgamma}} \eqsp .
\end{equation}
Using that $\floor{k/ \itgamma} \geq k\delta / (1 + \delta) -1$ concludes the
proof upon letting $\tilde{\rho}_{\msc} = (1 - \beta)^{1/(1+ \bdelta)}$ and
$\tilde{A}_{\msc} = (1 - \beta)^{-1}$.

\subsection{Proof of \Cref{prop:disc_invariant_K}}
\label{prop:disc_invariant_K:proof}

Let $\alpha, \lambda >0$, $\vareps \in \ocint{0, \vareps_0}$ such that
$2\lambda (\Ltt_y + \alpha \Ltt /\vareps - \min(\mtt, 0)) \leq 1$ and
$\bdelta_1 = (1/3)(\Ltt_y + \alpha \Ltt /\vareps + 1 / \lambda)^{-1}$. Recall that
for any $x_1, x_2 \in \rset^d$
\begin{equation}
  \norm{b_{\vareps}(x_1) - b_{\vareps}(x_2)} \leq (\Ltt_y + \alpha \Ltt / \vareps + 1/\lambda) \norm{x_1 - x_2} \eqsp . 
\end{equation}
Using this result, the fact that for any $x \in \rset^d$,
$\langle b_\vareps(x), x \rangle \leq -\tilde{\mtt} \norm{x}^2 + c$ and
\cite[Theorem 19.4.1]{douc:moulines:priouret:soulier:2018} there exist
$\bdelta_2 > 0$, $\tilde{B} \geq 0$ and $\tilde{\rho} \in \ocint{0,1}$ such that
for any $\delta \in \ocintLigne{0, \bdelta_2}$, $x \in \rset^d$ and
$k \in \nset$
\begin{equation}
  \Vnorm{\updelta_{x} \Rker_{\vareps, \delta}^k - \pi_{\vareps, \delta}} + \Vnorm{\updelta_{x} \Qker_{\vareps, \delta}^k - \pi_{\vareps, \delta}^{\msc}} \leq \tilde{B} \tilde{\rho}^{k \delta} V(x) \eqsp ,
\end{equation}
with $\tilde{B}$ and $\tilde{\rho}$ which do not depend on $R$.
In addition, using \Cref{lemma:drift_plus}, for any $k \in \nset$ and $\delta \in \ocintLigne{0, \bdelta_2}$ we have
\begin{equation}
  \Rker_{\vareps, \delta}^k V(x) \leq \tilde{\lambda}^{k \delta} V(x) + \tilde{c} \delta \eqsp ,
\end{equation}
with $\tilde{\lambda} \in \coint{0,1}$ and $\tilde{c} > 0$ which do not depend
on $R \geq 0$. For any $\delta \in \ocintLigne{0, \bdelta_2}$ we have
\begin{equation}
  \lambda^\delta + c\delta \leq \lambda^\delta (1 + c\delta \lambda^{-\bdelta_2}) \leq (\lambda \exp[c \lambda^{-\bdelta_2}])^{\delta} \eqsp .
\end{equation}
Let $A =\lambda \exp[c \lambda^{-\bdelta_2}]$, we have that for any
$x \in \rset^d$, $\Rker_{\vareps, \delta} V(x) \leq A^\delta V(x)$. Therefore we get that 
$(V(X_n) A^{-n})_{n \in \nset}$ is a supermartingale. Hence using Doob maximal
inequality and Markov inequality we get that
\begin{equation}
  \proba{\sup_{k \in \{0, \dots, n\}} \norm{X_k} \geq R} \leq V(x) A^{n\delta } \exp[-R] \eqsp . 
\end{equation}
Therefore, we get that for any $k  \in \nset$
\begin{equation}
\tvnorm{\pi_{\vareps, \delta} - \pi_{\vareps, \delta}^\msc} \leq (V(0) + \tilde{c} \bdelta_2) A^{k\delta}\exp[-R] + \tilde{B} \tilde{\rho}^{k \delta} V(0) \eqsp .
\end{equation}
We conclude upon letting $k = \floor{r / (2 \log(A) \delta)}$.

\section{Proofs of \Cref{sec:posterior_approx}}
\label{sec:proofs-citecr-1}

\subsection{Proof of \Cref{lemma:H1_check}}
\label{lemma:H1_check:proof}

The first part of the proposition is straightforward. Using Pinsker's inequality
\cite[Theorem 4.19]{boucheron2013concentration} we have for any 
$x \in \rset^{\dim}$
\begin{equation}
  \textstyle{\tvnorm{\mu - (\tau_{x})_{\#} \mu}^2 \leq 2 \KL((\tau_x)_{\#}\mu|\mu) \leq 2 \int_{\rset^{\dim}} \norm{U(\tilde{x} + x) - U(\tilde{x})} \rmd \mu(\tilde{x}) \leq 2 C_{\upgamma} \norm{x}^{\upgamma}  \eqsp .}
\end{equation}

For the second part of the proof, since there exist $c_1, \varpi > 0$
and $c_2 \in \rset$ such that for any $x \in \rset^{\dim}$,
$U(x) \geq c_1 \norm{x}^{\varpi} + c_2$ then for any $k \in \nsets$
and $\alpha >0$,
$\int_{\rset^{\dim}} (1 + \norm{x})^k p(x) < +\infty$.  Let
$q(x) = (1 + \norm{x})^{-(d+1)} / \int_{\rset^{\dim}} (1 +
\norm{\tilde{x}})^{-(d+1)} \rmd \tilde{x}$. Then using that for any
$t \geq 0$, $\abs{\rme^{t} - 1} \leq \abs{t} \rme^{\abs{t}}$ we get
that for any $x \in \rset^{\dim}$
\begin{multline}
  \int_{\rset^{\dim}} \abs{p(\tilde{x}) - p(x-\tilde{x})} q^{1 - 1/\alpha}(\tilde{x}) \rmd \tilde{x} \\ \leq C_{\upgamma} \norm{x}^{\upgamma} \exp[C_{\upgamma} \norm{x}^{\upgamma}] \int_{\rset^{\dim}} (1+\norm{\tilde{x}})^{(d+1)(1/\alpha-1)}p(\tilde{x}) \rmd \tilde{x} \parenthese{\int_{\rset^{\dim}}  (1 +
\norm{\tilde{x}})^{-(d+1)} \rmd \tilde{x}}^{1 - 1/\alpha} \eqsp ,
\end{multline}
which concludes the proof.

\subsection{Proof of \Cref{prop:posterior_eps}}
\label{prop:posterior_eps:proof}

First we show the following technical lemma.

\begin{lemma}
  \label{lemma:tech}
  For any $x,y \geq 0$ and $\beta > 0$,
  $(x+y)^{\beta} - x^{\beta} \leq 2^{\beta}(y^{\beta} +
  x^{(\beta-1)\wedge 0}y)$.
\end{lemma}

\begin{proof}
  The result is straightforward if $\beta \in \ocint{0,1}$, since in
  this case $(x+y)^{\beta} \leq x^{\beta} + y^{\beta}$. Assume that
  $\beta > 1$. If $x = 0$ the result holds. Now assume that $x >
  0$. If $y \geq x$ then
  $(x+y)^{\beta} - x^{\beta} \leq 2^{\beta} y^{\beta}$.  Assume that
  $y \leq x$. Since $\rmf: \ t \mapsto (1+t)^{\beta} - 1$ is convex we
  obtain that for any $t \in \ccint{0,1}$, $\rmf(t) \leq
  2^{\beta}t$. Using this result we have
  \begin{equation}
    (x+y)^{\beta} - x^{\beta} \leq x^{\beta}f(y/x) \leq 2^{\beta} x^{\beta-1} y \eqsp ,
  \end{equation}
  which concludes the proof.
\end{proof}

Before proving \Cref{prop:posterior_eps} we state the following lemma.
\begin{lemma}
  \label{lemma:tv_norm}
  Let $\pi_1, \pi_2$ two probability measures and
  $q_1, q_2: \ \rset^{\dim} \to \coint{0, +\infty}$ two measurables
  functions such that for any $x \in \rset^{\dim}$,
  $(\rmd \pi_i / \rmd \Leb)(x) = q_i(x) / c_i$ with
  $c_i = \int_{\rset^{\dim}} q_i(\tilde{x}) \rmd \tilde{x}$.  Denote
  $\mathrm{D} = \int_{\rset^{\dim}} |q_1(x) - q_2(x)|$. We have
  \begin{equation}
    \tvnorm{\pi_1 - \pi_2} \leq 2 c_1^{-1} \mathrm{D}   \eqsp .
  \end{equation}
\end{lemma}

\begin{proof}
  We have
  \begin{equation}
    \tvnorm{\pi_1 - \pi_2} = \int_{\rset^{\dim}} \abs{\dfrac{q_1(x)}{c_1} - \dfrac{q_2(x)}{c_2}} \rmd x 
                           \leq c_1^{-1} (\mathrm{D} + \abs{c_2 - c_1}) \eqsp ,
  \end{equation}
  which concludes the proof using that $\abs{c_2 - c_1} \leq \mathrm{D}$.
\end{proof}
We now give the proof  of \Cref{prop:posterior_eps}.

\begin{proof}
  Let $\alpha > 0$. For any $\vareps >0$ and $x \in \rset^{\dim}$ denote
  $\bar{p}(x) = p(y|x) p^{\alpha}(x)$ and
  $\bar{p}_{\varepsilon}(x) = (py|x) p_{\vareps}^{\alpha}(x)$, where we recall that for any $x \in \rset^{\dim}$
  \begin{equation}
    p_{\vareps}(x) = (2 \uppi \vareps)^{-d/2} \int_{\rset^{\dim}} p(\tilde{x}) \exp[-\norm{x - \tilde{x}}^2/(2\vareps)] \rmd \tilde{x} \eqsp .
  \end{equation} For any
  $\vareps >0$ we have
  \begin{equation}
    \label{eq:bound_easy}
    \int_{\rset^{\dim}} \abs{\bar{p}(x) - \bar{p}_{\varepsilon}(x)} \rmd x \leq \norm{p(y|\cdot)}_{\infty} \int_{\rset^{\dim}} \abs{p^{\alpha}(x) - p_{\vareps}^{\alpha}(x)} \rmd x \eqsp .
  \end{equation}
  Using \Cref{lemma:tech} and that
  $\| p_{\vareps}\|_{\infty} \leq \| p \|_{\infty} < + \infty$, we have for any
  $\vareps >0$ and $x \in \rset^{\dim}$
  \begin{align}
    \label{eq:ineq_uno_alpha}
    \int_{\rset^{\dim}} \abs{\bar{p}(x) - \bar{p}_{\varepsilon}(x)} \rmd x & \leq 2^{\alpha} \norm{p(y|\cdot)}_{\infty} (1 + \| p \|_{\infty}^{(\alpha - 1) \wedge 0}) \\
    & \qquad \qquad \times \defEns{\int_{\rset^{\dim}} \abs{p(x) - p_{\vareps}(x)} \rmd x + \int_{\rset^{\dim}} \abs{p(x) - p_{\vareps}(x)}^{\alpha}\rmd x   } \eqsp .
  \end{align}
  Using Jensen's inequality, for any
  $q: \ \rset^{\dim} \to \ooint{0, +\infty}$ with
  $\int_{\rset^{\dim}} q(\tilde{x}) \rmd \tilde{x} = 1$ we have 
  \begin{equation}
    \int_{\rset^{\dim}} \abs{p(x) - p_{\vareps}(x)}^{\alpha} \rmd x \leq \parenthese{\int_{\rset^{\dim}} \abs{p(x) - p_{\vareps}(x) q^{1 - 1/\alpha}(x) } \rmd
  x}^{\alpha} \eqsp .
\end{equation}
Combining this result with \eqref{eq:ineq_uno_alpha} we get that
\begin{multline}
  \int_{\rset^{\dim}} \abs{\bar{p}(x) - \bar{p}_{\varepsilon}(x)} \rmd x \leq 2^{\alpha} \norm{p(y|\cdot)}_{\infty} (1 + \| p \|_{\infty}^{(\alpha - 1) \wedge 0}) \\
    \times \defEns{\int_{\rset^{\dim}} \abs{p(x) - p_{\vareps}(x)} \rmd x + \parenthese{\int_{\rset^{\dim}} \abs{p(x) - p_{\vareps}(x)} q^{1 - 1/\alpha}(x) \rmd x }^{\alpha}  } \eqsp .
  \end{multline}
  If $\alpha \geq 1$, choosing $q$ such that $\| q \|_{\infty} \leq 1$ we get
  \begin{align}
    \int_{\rset^{\dim}} \abs{\bar{p}(x) - \bar{p}_{\varepsilon}(x)} \rmd x &\leq 2^{\alpha} \norm{p(y|\cdot)}_{\infty} (1 + \| p \|_{\infty}^{(\alpha - 1) \wedge 0}) \\
    & \qquad \qquad \times \defEns{\int_{\rset^{\dim}} \abs{p(x) - p_{\vareps}(x)}
      \rmd x + \parenthese{\int_{\rset^{\dim}} \abs{p(x) -
          p_{\vareps}(x)} (x) \rmd x }^{\alpha} } \eqsp .     \label{eq:ineq_p_bar}
  \end{align}
Hence since 
  $p \in \mathrm{L}^{1}(\rset^{\dim})$ and
  $\ensembleLigne{\tilde{x} \mapsto (2\uppi \vareps)^{-d/2}\exp[-
    \norm{\tilde{x}}^2/ (2 \vareps)]}{\vareps > 0}$ is a family of
  mollifiers, we have
  $\lim_{\vareps \to 0} \int_{\rset^{\dim}}\abs{p(x) -
    p_{\vareps}} \rmd x = 0 $. Combining this result,
  \eqref{eq:ineq_p_bar} and \Cref{lemma:tv_norm} concludes the first
  part of the proof.

  Now let $\alpha > 0$ and assume \rref{assum:posterior}($\alpha$). If
  $\alpha \geq 1$ then using \eqref{eq:ineq_uno_alpha} we have
  \begin{equation}
    \int_{\rset^{\dim}} \abs{\bar{p}(x) - \bar{p}_{\varepsilon}(x)} \rmd x \leq 2^{\alpha}(1+2^{\alpha -1}) \norm{p(y|\cdot)}_{\infty} (1 + \| p \|_{\infty}^{(\alpha - 1) \wedge 0}) \int_{\rset^{\dim}} \abs{p(x) - p_{\vareps}(x)} \rmd x \eqsp .
  \end{equation}
  If $\alpha < 1$ then using that $\| q \|_{\infty} < +\infty$, we get that
    \begin{multline}
    \int_{\rset^{\dim}} \abs{\bar{p}(x) - \bar{p}_{\varepsilon}(x)} \rmd x \leq 2^{\alpha} \norm{p(y|\cdot)}_{\infty} (1 + \norm{q}_{\infty}^{1/\alpha - 1})(1 + \| p \|_{\infty}^{(\alpha - 1) \wedge 0}) \\ \times \defEns{\int_{\rset^{\dim}} \abs{p(x) - p_{\vareps}(x)}q^{1-1/\alpha}(x) \rmd x + \parenthese{\int_{\rset^{\dim}} \abs{p(x) - p_{\vareps}(x)}q^{1-1/\alpha}(x) \rmd x}^{\alpha}} \eqsp .
  \end{multline}
  Hence, in any case, there exists $\tilde{C}_0 \geq 0$ such that
  \begin{multline}
    \int_{\rset^{\dim}} \abs{\bar{p}(x) - \bar{p}_{\varepsilon}(x)} \\
    \leq \tilde{C}_0 \defEns{\int_{\rset^{\dim}} \abs{p(x) -
        p_{\vareps}(x)}q^{\min(1-1/\alpha, 0)}(x) \rmd x +
      \parenthese{\int_{\rset^{\dim}} \abs{p(x) -
          p_{\vareps}(x)}q^{\min(1-1/\alpha, 0)}(x) \rmd x}^{\alpha}}
    \eqsp .
  \end{multline}
  Using Jensen's inequality and the
  change of variable $\tilde{x} \mapsto \vareps^{1/2} \tilde{x}$, we
  have for any $\vareps \in \ocintLigne{0, (4\constanteM)^{-1}}$
  \begin{align}
    &\int_{\rset^{\dim}} \abs{p(x) - p_{\vareps}(x)} q^{\min(1-1/\alpha, 0)}(x)\rmd x \\
    & \qquad \qquad \leq  \int_{\rset^{\dim}} \int_{\rset^{\dim}} \abs{p(x) - p(x - \tilde{x})}q^{\min(1-1/\alpha, 0)}(x) (2 \uppi \vareps)^{-d/2} \exp[-\norm{\tilde{x}}^2 / (2 \vareps)] \rmd x \rmd \tilde{x} \eqsp \\
                                                    &\qquad \qquad\leq \int_{\rset^{\dim}} \exp[\constanteM \norm{\tilde{x}}^2] \norm{\tilde{x}}^{\upbeta}  (2 \uppi \vareps)^{-d/2} \exp[-\norm{\tilde{x}}^2 / (2 \vareps)] \rmd \tilde{x} \\
                                                    &\qquad \qquad\leq \vareps^{\upbeta/2} (2\uppi)^{-d/2} \int_{\rset^{\dim}} \exp[\constanteM \vareps \norm{\tilde{x}}^2] \norm{\tilde{x}}^{\upbeta} \exp[-\norm{\tilde{x}}^2/2 ] \rmd \tilde{x} \\
                                                    &\qquad \qquad\leq \vareps^{\upbeta/2} (2\uppi)^{-d/2} \int_{\rset^{\dim}} \norm{\tilde{x}}^{\upbeta} \exp[-\norm{\tilde{x}}^2/4 ]  \rmd \tilde{x} \leq C_0 \vareps^{\upbeta/2} \eqsp , 
  \end{align}
  with
  $C_0 = (2\uppi)^{-d/2} \int_{\rset^{\dim}} \norm{\tilde{x}}^{\beta}
  \exp[-\norm{\tilde{x}}^2/4] \rmd
  \tilde{x}$. Hence, we have
  \begin{equation}
    \label{eq:duo_epsilon}
    \int_{\rset^{\dim}} \abs{\bar{p}(x) - \bar{p}_{\varepsilon}(x)} \rmd x  \leq C_1 (\vareps^{\upbeta/2} + \vareps^{\upbeta \alpha/2}) \eqsp ,
  \end{equation}
  with $C_1 = \tilde{C}_0 (C_0 + C_0^{\alpha})$.  Let
  $\vareps_1 = \min((c C_1)^{-2 / \upbeta}/2, (c C_1)^{-2 / (\upbeta
    \alpha)}/2, (4\constanteM)^{-1})$ and
  $c= \int_{\rset^{\dim}} \bar{p}(x) \rmd x$. Combining
  \eqref{eq:duo_epsilon} with \Cref{lemma:tv_norm}, we get that for any
  $\vareps \in \ocint{0, \vareps_1}$
  \begin{equation}
    \tvnorm{\pi - \pi_{\vareps}} \leq 2 c^{-1} C_1 (\vareps^{\upbeta/2} + \vareps^{\upbeta\alpha/2})\eqsp ,
  \end{equation}
which concludes the proof upon letting $A_0 = 2 c^{-1} C_1 $.
\end{proof}

\end{document}